\documentclass[reprint,amsfonts, amssymb, amsmath,  showkeys,pra,nofootinbib , superscriptaddress, twocolumn,longbibliography,aps]{revtex4-2}

\usepackage{float}
\makeatletter
\let\newfloat\newfloat@ltx
\makeatother
\usepackage[english]{babel}

\usepackage[utf8]{inputenc}
\usepackage{graphics}
\usepackage{selinput}
\usepackage[normalem]{ulem}
\usepackage[shortlabels]{enumitem}

\usepackage{braket}
\usepackage{amsthm}
\usepackage{mathtools}
\usepackage{physics}
\usepackage[table]{xcolor}
\usepackage{graphicx}
\usepackage[left=16mm,right=16mm,top=35mm,columnsep=15pt]{geometry} 
\usepackage{adjustbox}
\usepackage{placeins}
\usepackage[T1]{fontenc}
\usepackage{lipsum}
\usepackage{csquotes}
\usepackage{bm}

\usepackage[linesnumbered,ruled,vlined]{algorithm2e}
\SetKwInput{kwInit}{Init}

\usepackage{changes}
\definechangesauthor[name=Ilya Safro, color=red]{is}
\definechangesauthor[name=Ankit, color=orange]{ak}

\def\HC{\mathcal{H}}

\def\LC{\mathcal{L}}

\usepackage{mathrsfs} 

\def\ad{^{\dagger}}

\newcommand{\fsnull}[1]{}
\newcommand{\old}[1]{}

\usepackage[makeroom]{cancel}
\usepackage[toc,page]{appendix}
\usepackage[colorlinks=true,citecolor=blue,linkcolor=magenta]{hyperref}

\usepackage{tikz}
\tikzset{every picture/.style=remember picture}

\newcommand{\kett}[1]{|#1\rangle\!\rangle}
\newcommand{\braa}[1]{\langle\!\langle{#1}|}

\newcommand{\braakett}[1]{\langle\!\braket{#1}\!\rangle}

\newcommand{\poly}{\operatorname{poly}}

\newcommand{\Ebb}{\mathbb{E}}

\newcommand{\GC}{\mathcal{G}}

\newcommand{\NC}{\mathcal{N}}
\newcommand{\OC}{\mathcal{O}}
\newcommand{\PC}{\mathcal{P}}

\newcommand{\Var}{{\rm Var}}

\renewcommand{\geq}{\geqslant}
\renewcommand{\leq}{\leqslant}

\DeclareMathOperator*{\argmin}{arg\,min}
\DeclareMathOperator{\sinc}{sinc}

\renewcommand{\vec}[1]{\boldsymbol{#1}}  

\newcommand*{\id}{\openone}

\newcommand{\bs}{\textsf{BS}}

\renewcommand{\th}{\theta } 

\newcommand{\thv}{\vec{\theta}}

\def\be{\begin{equation}}
\def\ee{\end{equation}}
\def\bs{\begin{split}}
\def\e{\end{split}}
\def\ba{\begin{eqnarray}}
\def\bea{\begin{eqnarray}}

\def\tea{\end{eqnarray}}
\def\ea{\end{eqnarray}}
\def\eea{\end{eqnarray}}

\theoremstyle{theorem}
\newtheorem{theorem}{Theorem}
\newtheorem{lemma}{Lemma}

\theoremstyle{definition}
\newtheorem{example}{Example}

\theoremstyle{theorem}

\newtheorem{definition}{Definition}

\usepackage{dsfont}

\def\be{\begin{equation}}
\def\te{\end{equation}}
\def\ee{\end{equation}}
\def\ba{\begin{eqnarray}}
\def\bea{\begin{eqnarray}}

\def\tea{\end{eqnarray}}
\def\ea{\end{eqnarray}}
\def\eea{\end{eqnarray}}

\newcommand{\beq}{\begin{equation}}
\newcommand{\eeq}{\end{equation}}

\usepackage{graphicx}

\usepackage{enumitem}

\usepackage{dcolumn}
\usepackage{bm,bbm}
\usepackage{physics}
\usepackage{amsmath}
\usepackage{amsthm}
\usepackage{amsfonts}
\usepackage{hyperref}
\usepackage{xcolor}
\usepackage{float}
\usepackage{soul}
\sethlcolor{yellow}

\usepackage{lipsum}

\usepackage{cancel}

\newcommand{\1}{\mathbbm{1}}

\newcommand{\sx}{\sigma_x}
\newcommand{\sy}{\sigma_y}
\newcommand{\sz}{\sigma_z}

\makeatletter
\renewcommand\onecolumngrid{
\do@columngrid{one}{\@ne}
\def\set@footnotewidth{\onecolumngrid}
\def\footnoterule{\kern-6pt\hrule width 1.5in\kern6pt}
}
\renewcommand\twocolumngrid{
        \def\footnoterule{
        \dimen@\skip\footins\divide\dimen@\thr@@
        \kern-\dimen@\hrule width.5in\kern\dimen@}
        \do@columngrid{mlt}{\tw@}
}
\makeatother

\begin{document}

\preprint{APS/123-QED}

\title{Exponentially many initializations to avoid barren plateaus}

\author{Ankit Kulshrestha}
\thanks{The first two authors contributed equally}
\affiliation{Fujitsu Research of America, Santa Clara, CA 95054, USA}
\affiliation{University of Delaware, Newark, DE 19716, USA}

\author{Ricard Puig}
\thanks{The first two authors contributed equally}
\affiliation{Institute of Physics, Ecole Polytechnique F\'{e}d\'{e}rale de Lausanne (EPFL), CH-1015 Lausanne, Switzerland}

\affiliation{Centre for Quantum Science and Engineering, Ecole Polytechnique F\'{e}d\'{e}rale de Lausanne (EPFL), CH-1015 Lausanne, Switzerland}

\author{Diego Garc\'ia-Mart\'in}
\affiliation{Department for Quantum Information and Computation at Kepler (QUICK), Johannes Kepler University, Linz, Austria }

\author{Lukasz Cincio}
\affiliation{Theoretical Division, Los Alamos, NM, 87545, USA}

\author{Ilya Safro}
\affiliation{University of Delaware, Newark, DE 19716, USA}

\author{Zo\"e Holmes}
\affiliation{Institute of Physics, Ecole Polytechnique F\'{e}d\'{e}rale de Lausanne (EPFL), CH-1015 Lausanne, Switzerland}

\affiliation{Centre for Quantum Science and Engineering, Ecole Polytechnique F\'{e}d\'{e}rale de Lausanne (EPFL), CH-1015 Lausanne, Switzerland}

\author{M. Cerezo}
\thanks{cerezo@lanl.gov}
\affiliation{Information Sciences, Los Alamos National Laboratory, Los Alamos, NM 87545, USA}

\begin{abstract}
Barren plateaus are stated as an average-case phenomenon: pick an ansatz, initialize it naively, and concentration follows. This has led to the common view that a potential cure for barren plateaus is simply to initialize the parameters more carefully. Here we show that the situation is subtler. We introduce a first-moment framework that gives a simple operator-level diagnostic for when an initialization may escape the fully concentrated barren-plateau fixed point, and for comparing the biases induced by different initialization strategies. Our framework recovers several known initialization schemes such as identity and Gaussian initialization, but also shows that barren-plateau avoidance is highly non-unique. Indeed, many shifted, biased, and non-symmetric parameter distributions can avoid concentration, and these choices need not be equivalent. In fact, our results show that one can generate exponentially many families of inequivalent initialization strategies. Then, our numerics indicate that different first-moment-distinct initializations can lead to different attained minima, suggesting that avoiding barren plateaus via smart initializations can trade the exponential concentration problem for the challenge of selecting the right trainable pocket amongst many options.
\end{abstract}

\maketitle

\section{Introduction}

Over the past few years, a significant amount of effort has been put forward towards understanding if and how parametrized quantum circuits (PQCs) can be trained to accurately and efficiently find the global minima of loss functions for variational quantum algorithms~\cite{cerezo2020variationalreview,bharti2021noisy,endo2021hybrid,schuld2015introduction,biamonte2017quantum,cerezo2022challenges,di2023quantum}. This has led researchers to realize that quantum training landscapes can suffer from issues such as the barren plateau phenomenon~\cite{mcclean2018barren,larocca2024review}, where loss function values and their gradients concentrate exponentially; the presence of exponentially many local minima~\cite{you2021exponentially,fontana2022nontrivial,anschuetz2022beyond,anschuetz2021critical,larocca2021diagnosing}; or even the fact that training a variational model can be NP-hard~\cite{bittel2021training}.

While these results should certainly raise some eyebrows, they must also be taken with a (fairly large) grain of salt. Many no-go results in variational quantum computing are average-case or worst-case in nature, and therefore do not by themselves preclude the possibility that specific instances, architectures, or parameter regimes may remain trainable. In the particular case of barren plateaus, the average-case perspective is not just a technical nuisance, but rather  reflects a deeper issue for these models: the curse of dimensionality. Specifically, most loss functions used in the literature can be ultimately expressed as the Hilbert-Schmidt inner product between an input state and some observable, both living in an exponentially large operator space. As such, one can generically expect this overlap to become, on average, exponentially concentrated~\cite{larocca2024review}. 

Broadly speaking, current strategies for avoiding barren plateaus seem to fall into two different categories. The first one avoids the problem by taming the underlying curse of dimensionality itself. Here, one typically restricts the relevant dynamics to some smaller effective subspaces, for instance through shallow local architectures~\cite{cerezo2020cost,pesah2020absence,khatri2019quantum,zhao2021analyzing,liu2021presence,miao2023isometric,bermejo2024quantum,bach2024mlqaoa}, symmetries~\cite{larocca2022group,meyer2022exploiting,skolik2022equivariant,ragone2022representation,nguyen2022atheory,schatzki2022theoretical,zheng2021speeding,east2023all,chang2026practical,tsvelikhovskiy2026equivariant,shaydulin2021classical}, or small dynamical Lie algebras~\cite{larocca2021diagnosing,monbroussou2023trainability,cherrat2023quantum,fontana2023theadjoint,ragone2023unified,diaz2023showcasing,west2023provably,heredge2024prospects,aguilar2024full,kokcu2024classification,kazi2022landscape,tsvelikhovskiy2026reductions}. While this route has proven extremely successful at preventing barren plateaus, it also comes with an important caveat. Namely, that if the dynamics are constrained to some polynomially-sized subspace, then the same structure can also make the model more amenable to classical simulation~\cite{cerezo2023does,bermejo2024quantum,angrisani2024classically,lerch2024efficient,anschuetz2024arbitrary,mele2024noise,shin2024dequantising,ermakov2024unified,miller2025simulation, rudolph2026thermal,chang2026practical,barligea2026enabling}. The second route is more subtle. Rather than fundamentally shrinking the space explored by the PQCs, one can try to avoid the bad average-case scenario by initializing the parameters in a smart way~\cite{grant2019initialization,kulshrestha2022beinit,duffield2023bayesian,heyraud2023efficient,zhang2022escaping,park2023hamiltonian, wang2023trainability,park2024hardware,shi2024avoiding, puig2024variational,puig2026warm, lerch2026iqp}, thereby steering optimization toward special patches of the landscape where gradients remain large. In other words, some strategies avoid barren plateaus by constraining the space one can explore, while others try to avoid them by sampling a trainable part of the landscape, but without creating a barrier for the optimizer to see the full parameter space. While the second approach preserves the expressivity of the model, it introduces a new challenge due to the existence of local minima~\cite{anschuetz2022quantum, nemkov2024barren}, making the choice of ``where to start'' critical.

In this work, we focus precisely on the second route. We develop a simple framework to diagnose when a parameter initialization may avoid the fully concentrated barren-plateau regime, and to compare different initialization strategies. Our goal is not just to say that smart initializations can avoid barren plateaus, but rather to show that there can be exponentially many such initializations, and that they need not be equivalent. This, in turn, changes the initialization problem, as once exponentially small gradients are avoided, the relevant question becomes \emph{which initialization among many possible choices leads to a sufficiently favorable region of the landscape} (see Fig.~\ref{fig:schematic}). We support this picture both analytically and numerically, showing that different initialization families can  lead to different landscape patches and thus to different training dynamics, and different attained minima. 

This manuscript is organized as follows. In Section~\ref{sec:frmwrk}, we present the definitions and tools necessary to follow the results. In Sec.~\ref{sec:an_res}, we start by presenting two theorems which provide necessary conditions for an initialization strategy to avoid exponential concentration, and we use them to analyze several initializations. We then extend these tools to compare different non-concentrated initialization strategies, introduce the notions of symmetric and antisymmetric distributions, and explicitly show how to obtain exponentially many different (and potentially non-equivalent) initializations capable of avoiding barren plateaus. In Sec.~\ref{sec:numerics}, we present numerical simulations that can help determine whether these initialization strategies are practically distinguishable. Finally, we present our discussions in Sec.~\ref{sec:discussion}.

\section{Framework}\label{sec:frmwrk}

In this section we present the basic notions and tools that will be used throughout our work.

\subsection{Parameterized quantum circuit}

\begin{figure}
    \centering
    \includegraphics[width=.9\linewidth]{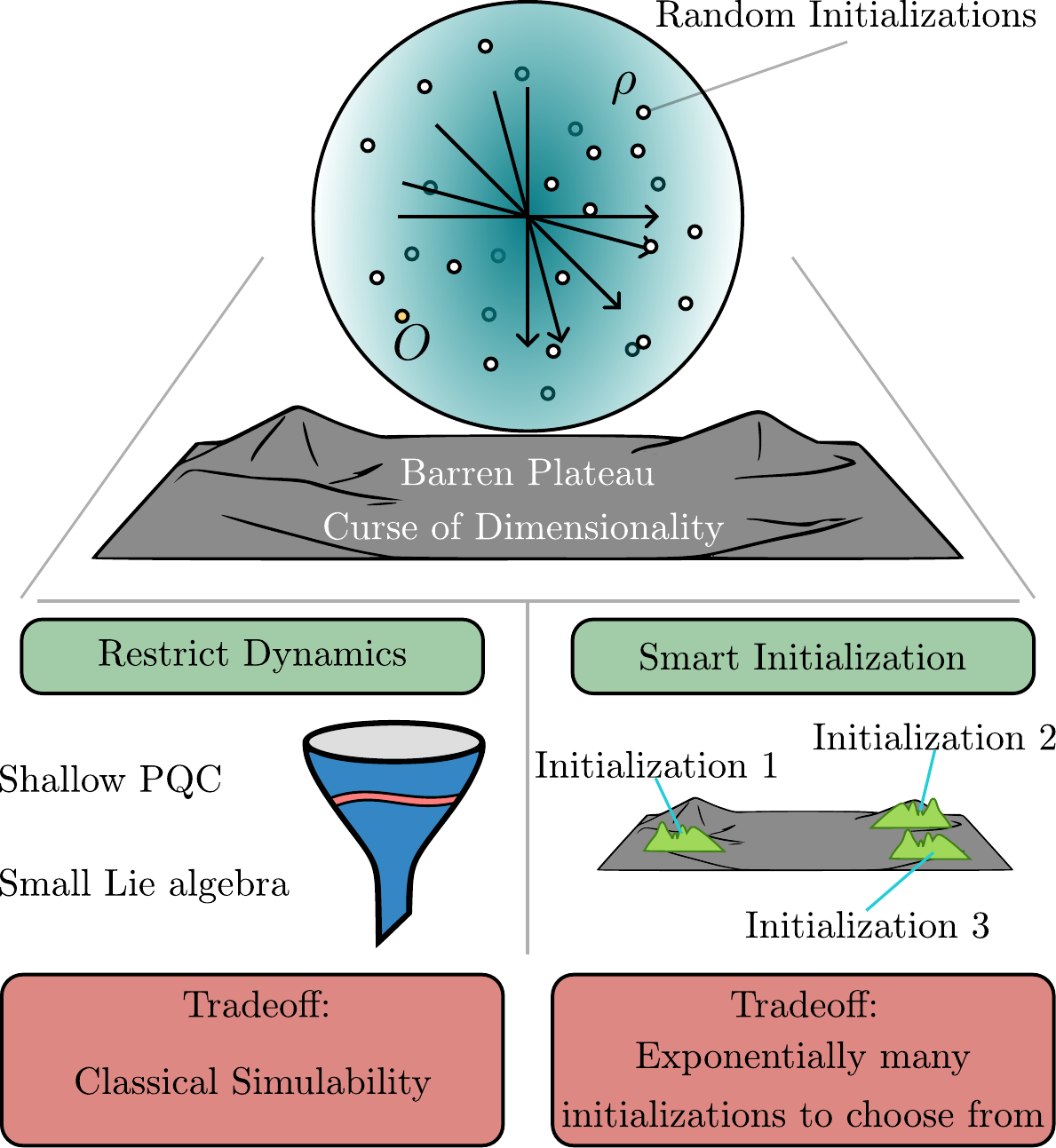}
\caption{\textbf{Schematic picture of the main message of this work.}  Barren plateaus arise as an average-case consequence of exploring an exponentially large space, leading to concentration of loss functions and gradients. One way to avoid this problem is to restrict the dynamics to a smaller effective subspace, e.g., through shallow local structure, symmetries, or small Lie algebras; however, this can make the model more amenable to classical simulation. A different route is to use smart parameter initializations, which do not necessarily shrink the accessible space, but instead bias optimization toward special trainable pockets of the landscape. The tradeoff is then no longer only how to avoid exponential concentration, but which among many candidate initializations should be chosen. 
}
    \label{fig:schematic}
\end{figure}

Let $\HC=(\mathbb{C}^2)^{\otimes n}$ denote an $n$-qubit Hilbert space. We will define a parameterized quantum circuit (PQC) acting on $\HC$ as a product of unitaries of the form
\begin{align}\label{eq:pauliPQC}
    U(\thv) = \prod_{l=1}^L e^{-i\th_l\sigma_l/2}V_l\,,
\end{align}
where $\sigma_l$ is a Pauli operator taken from a set of generators~$\GC$, $\thv=\{\theta_l\}_{l=1}^L$ a set of trainable parameters, and $V_l$ a non-parametrized gate. In what follows, we will assume that the PQC is universal, meaning that for any unitary $W\in \mathbb{U}(\HC)$, where $\mathbb{U}(\HC)$ is the unitary group of degree $2^n$, there exists a depth $L$ and set of parameters $\thv$ for which $U(\thv)=W$. However, we note that the techniques developed here can be extended to the case when $U(\thv)$ has symmetries and is only able to produce unitaries from a subgroup of  $\mathbb{U}(\HC)$. 

Within the context of VQAs, the goal is to optimize the vector $\thv$ to minimize a given loss function $\LC(\thv)$, defined by
\begin{equation}\label{eq:loss}
\LC({\thv})=\Tr[U(\thv)\rho U\ad(\thv) O]\,,
\end{equation}
for some traceless Hermitian observable $O$ such that $\|O\|_2\in\OC(2^n)$ (e.g., a Pauli operator, or a combination thereof), that quantifies the extent to which a given task has been resolved. That is, we optimize the PQC's parameters to solve the optimization problem
\begin{equation}\label{eq:min-prob}
    \argmin_{\thv } \LC({\thv})\,.
\end{equation}

Extensive research has characterized how the properties and inductive biases of $U(\thv)$ affect the optimization in Eq.~\eqref{eq:min-prob}. This includes the choice of gates (either chosen from a pool of implementable operations~\cite{kandala2017hardware}, or so that they encode information or respect a given symmetry  of the task at hand~\cite{grimsley2019adaptive,farhi2014quantum,nguyen2022atheory,larocca2022group,wiersema2020exploring,ragone2022representation,Chang2023Approximately,west2023provably,forestano2023comparison,west2023reflection,zheng2021speeding,schatzki2022theoretical,zheng2022super}), the way in which the parameters in the PQC are sampled and initialized~\cite{grant2019initialization,kulshrestha2022beinit,duffield2023bayesian,heyraud2023efficient,zhang2022escaping,park2023hamiltonian, wang2023trainability,park2024hardware,shi2024avoiding, puig2024variational,mhiri2025unifying,puig2026warm, lerch2026iqp}, and the classical optimizer, among others. The focus of this work is the strategy used to initialize training, which is defined by some parameter distributions from which the PQC parameters are sampled
\begin{align}
    \theta_l \sim P_{\gamma_l}(\theta_l)\,.
\end{align}
Here $P_{\gamma_l}(\theta_l)$ denotes a probability distribution of $\theta_l$, with $\gamma$ referring to the probability distribution.  In what follows, we will denote the joint distribution of the parameter vector as $P_{\bm{\gamma}}$.

\subsection{Barren plateaus}\label{sec:bp}
One of the main roadblocks for solving Eq.~\eqref{eq:min-prob} is that of barren plateaus. We say that a loss function exhibits a barren plateau~\cite{larocca2024review} if 
\begin{equation}\label{eq:exp_concentration}
    \Var_{\bm{\gamma}}[\LC(\thv)]\in\order{\frac{1}{b^n}}\,,
\end{equation}
with $b>1$ and where we recall that
\begin{align}
    \Var_{\bm{\gamma}}[\LC(\thv)] = \mathbb{E}_{\bm{\gamma}}[\LC(\thv)^2] - \mathbb{E}_{\bm{\gamma}}[\LC(\thv)]^2\, .  \nonumber
\end{align}
Above, $\mathbb{E}_{\bm{\gamma}}$ denotes the expectation value over the ensemble $\{U(\thv), P_{\bm{\gamma}}(\thv) \}$, i.e., the set of unitaries obtained by sampling parameters according to probability distribution $P_{\bm{\gamma}}(\thv)$. Typically, this ensemble is obtained by sampling parameters from the whole parameter range, and thus makes an average statement about the whole landscape. 

From Chebyshev's inequality,
\begin{align}
    \Pr[|\LC(\thv)-\mathbb{E}_{\bm{\gamma}}[\LC(\thv)]|\geq c ]\leq\frac{ \Var_{\bm{\gamma}}[\LC(\thv)]}{c^2}\,, \nonumber
\end{align}
for some $c>0$, we see that this means that the loss is exponentially concentrated around the mean. Hence, when initializing the PQC, the probability of large deviations from the mean are exponentially suppressed. This implies that exponential precision will typically be needed to resolve the gradients of the loss function that are required to variationally train the PQC. 

As evidenced from the fact that the expectation values above directly depend on the ensemble $\{U(\thv), P_{\bm{\gamma}}(\thv) \}$, the choice of initialization can determine whether the loss function concentrates. In what follows, we will assume that there exist a na\"ive parameter initialization distribution, denoted as $P_{{\rm BP}}$, for which a barren plateau exists, and for which the ensemble $\{U(\thv), P_{\bm{\gamma}}(\thv) \}$ forms a one-design over $\mathbb{U}(\HC)$ (see the next section for a definition of $t$-designs).  For example, consider a  standard variational quantum eigensolver task~\cite{peruzzo2014variational}, where  $O$ is a sum of local Pauli operators acting on $\OC(1)$ qubits, and $\rho$ is the all-zero state (or some other fiducial tensor-product initial state). Then, take $U(\thv)$ to be a hardware efficient ansatz~\cite{kandala2017hardware,cerezo2020variationalreview} composed of $\Theta(\poly(n))$ layers, where at each layer we apply single qubit rotations $e^{-i \theta Y }e^{-i \theta' X }$ to each qubit, followed by two-qubit entangling gates (e.g. CNOTs)  acting on nearest neighbors in a brick-layered fashion. Here, if $P_{{\rm BP}}$ is  obtained by sampling each parameter independently and identically from ${\rm Unif}[- \pi,  \pi]$, it has been proven that such loss will exhibit a barren plateau~\cite{cerezo2020cost}, i.e., $\Var_{{\rm BP}}[\LC(\thv)]\in\order{\frac{1}{b^n}}$. 

\subsection{Moment operators}\label{sec:exp_moment_loss}

To analyze how the parameter distribution affects the existence of barren plateaus, it is convenient to somewhat separate the influence of the PQC, from that of the initial state and the measurement operator. For this purpose, one defines the $t$-th order twirl superoperator
\begin{equation}
    \tau_{\bm{\gamma}}^{(t)}\left(\cdot\right) = \int d\thv P_{\bm{\gamma}}(\thv) U(\thv)^{\otimes t}(\cdot)\left( U(\thv)^\dagger \right)^{\otimes t}\,, 
\end{equation}
from which the loss function moments 
\begin{align}
    \mathbb{E}_{\bm{\gamma}}[\LC(\thv)^t] \,,
\end{align}
take the simple form 
\begin{align}
    \mathbb{E}_{\bm{\gamma}}[\LC(\thv)^t]  = \Tr\left[ \tau_{\bm{\gamma}}^{(t)}\left(\rho^{\otimes t}\right) O^{\otimes t} \right]\,.
\end{align}
Moreover in the vectorized picture, this becomes
\begin{align}\label{eq:vecmoment}
    \mathbb{E}_{\bm{\gamma}}[\LC(\thv)^t] = \langle\!\langle\rho^{\otimes t}|\widehat{\tau}^{(t)}_{\bm{\gamma}}|O^{\otimes t}\rangle\!\rangle\, .
\end{align}
For completeness,  we recall that the vectorization of a matrix $A = \sum_{i,j} A_{i,j}\ket{i}\bra{j}$ is defined as  $
   \kett{A} =  \sum_{i,j} A_{i,j}\ket{i}\ket{j}$; while given a channel $\Phi(\cdot)=\sum_\mu A_\mu (\cdot )B_\mu\ad$, its vectorization is the matrix $\widehat{\Phi}=\sum_\mu A_\mu \otimes  B^*_\mu$. As such, $\widehat{\tau}^{(t)}_{\bm{\gamma}}$ denotes the $t$-th order moment operator  
\begin{equation}\label{eq:twirl-vec}
    \widehat{\tau}^{(t)}_{\bm{\gamma}} = \int d\thv P_{\bm{\gamma}}(\thv) U(\thv)^{\otimes t}\otimes U^*(\thv)^{\otimes t}\, .
\end{equation}
We refer the reader for additional details to Appendix~\ref{app:vec_formalism}. 

For a PQC of the form in Eq.~\eqref{eq:pauliPQC}, it is helpful to write $\widehat{\tau}^{(t)}_{\bm{\gamma}}$ as a product of moment operators of the different gates. That is, for a circuit $U(\thv)=\prod_{l=1}^L U_l(\theta_l)$, where each parameter is independently initialized via $P_{\gamma_l}(\th_l)$, one obtains 
\begin{align}\label{eq:superoperator}
    \widehat{\tau}^{(t)}_{\bm{\gamma}}
    &= \prod_{l=1}^L  \widehat{\tau}^{(t)}_{\gamma_l}\, ,
\end{align}
where  
\begin{equation}\label{eq:layertwirl}
     \widehat{\tau}^{(t)}_{\gamma_l} = \int d\th_l P_{\gamma_l}(\th_l) U_l(\th_l)^{\otimes t}\otimes U_l^*(\th_l)^{\otimes t} \, .
\end{equation}

\subsection{A first-moment witness of escape from the fully concentrated barren-plateau fixed point}

As explained in Section~\ref{sec:bp}, on a landscape that suffers from a barren plateau, loss function values exponentially concentrate to a fixed point. Thus, one way of establishing that a particular initialization distribution $P_{\bm{\gamma}}(\thv)$ may not exhibit a barren plateau is to show that the average loss value for that distribution is substantially different to the fixed point for the full barren plateau landscape. More concretely, one can check whether the following minimal condition is satisfied
\begin{equation}
   | \mathbb{E}_{{\rm BP}}[\LC(\thv)] - \mathbb{E}_{\bm{\gamma}}[\LC(\thv)] | \in \Omega\left( \frac{1}{{\rm poly}(n)}\right)\,.
\end{equation} 
Of course, this diagnostic is not a complete trainability guarantee. In particular, a distribution could have the same first moment as the barren-plateau ensemble while still having large variance, or could have a shifted mean but poor gradients. Nevertheless, first-moment separation is a useful necessary ingredient for the particular mechanism we study here: avoiding the fully washed-out average operator induced by some naive initialization. 

Then, note that using Eq.~\eqref{eq:vecmoment}, this constraint becomes
\begin{equation}\label{eq:necessary}
  \Delta = | \braa{\rho}\widehat{\tau}^{(1)}_{{\rm BP}}\kett{O} - \braa{\rho}\widehat{\tau}^{(1)}_{\vec{\gamma}}\kett{O} |\in \Omega\left( \frac{1}{\poly(n)}\right)\,.
\end{equation} 
For ease of notation, we will henceforth omit the superscript $(1)$ and assume that a moment operator without superscript indicates $t=1$. One can readily see that a necessary, but not sufficient, condition for Eq.~\eqref{eq:necessary} to hold, is that 
\begin{equation}\label{eq:expresscondition}
   \| \widehat{\tau}_{{\rm BP}}\kett{O} - \widehat{\tau}_{\vec{\gamma}}\kett{O} \|_1\in \Omega\left( \frac{1}{\poly(n)}\right) \, ,
\end{equation} 
where $\|\cdot\|_p$ indicates the vector $p$-norm. The vectors $\widehat{\tau}_{{\rm BP}}\kett{O}$ and $\widehat{\tau}_{\vec{\gamma}}\kett{O} $ respectively represent the expected Heisenberg-evolved measurement operator obtained from the circuit ensembles arising from $P_{{\rm BP}}$ and $P_{\vec{\gamma}}$.  Thus, the norm above quantifies how much those two expected operators differ, a necessary condition for the expectation values over $\rho$ to be distinct. We provide a formal proof of these statements in Appendix~\ref{app:necessary_condition}.

At this point, we find it important to make several remarks. First, note that given two $t$-th order moment operators $\widehat{\tau}^{(t)}_{\vec{\alpha}}$ and $\widehat{\tau}^{(t)}_{\vec{\gamma}}$, the following bound holds for any matrix~$A$
\begin{equation}
  \| \widehat{\tau}^{(t)}_{\vec{\alpha}}\kett{A} - \widehat{\tau}_{\vec{\gamma}}^{(t)}\kett{A} \|_1\leq  \| \widehat{\tau}^{(t)}_{\vec{\alpha}} - \widehat{\tau}^{(t)}_{\vec{\gamma}} \|\cdot\|\kett{A}\|_1\,,\label{eq:ineq-norm}
\end{equation}
with $\|\cdot \|$ compatible with $\|\cdot \|_1$ (i.e., satisfying the inequality above). Indeed, this shows that if the moment operator themselves are similar, then differences between expected Heisenberg-evolved matrices will be small (and in turn the expectation value differences will also be small). Then, if we consider the case when $\widehat{\tau}^{(t)}_{\vec{\alpha}}$ and $ \widehat{\tau}^{(t)}_{\vec{\gamma}}$ are respectively obtained from sampling unitaries from a group $G$ according to the associated Haar measure, and from some other distribution, then the norm  $\| \widehat{\tau}^{(t)}_{\vec{\alpha}} - \widehat{\tau}^{(t)}_{\vec{\gamma}} \|$  has previously been called the ``expressivity'' of the PQC~\cite{sim2019expressibility}.   We refer the reader to~\cite{holmes2021connecting} for additional details on how expressivity measures relate to the barren plateau phenomenon.  In particular, if $\widehat{\tau}^{(t)}_{\vec{\gamma}}=\widehat{\tau}^{(t)}_{\vec{\alpha}} $, we say that the ensemble  $\{U(\thv),P_{\vec{\gamma}}(\thv)\}$ forms a $t$-design over $G$. In this case, one can leverage the Weingarten calculus to analytically evaluate the formula for  $\widehat{\tau}^{(t)}_{\vec{\gamma}}$~\cite{mele2023introduction}.

\subsection{Three questions about parameter initialization}

At this stage it is useful to distinguish three questions that are often conflated in discussions of parameter initialization. First, one may ask whether a given initialization strategy avoids the exponentially concentrated barren plateau regime, i.e., \textit{Is the expected observable separated from the BP ensemble?} For a fixed task, defined by the pair $(\rho,O)$, this requires that the average loss under $P_{\vec{\gamma}}$ remains inversely polynomially separated from the barren-plateau fixed point, as in Eq.~\eqref{eq:necessary}. A more tractable operator-level necessary condition is the one in Eq.~\eqref{eq:expresscondition}, namely that the expected Heisenberg-evolved observable under $P_{\vec{\gamma}}$ remains polynomially separated from that induced by $P_{\rm BP}$. 

Second, one may wonder: \textit{Do two such initializations induce different expected observables?} Indeed, even if two initialization strategies both avoid exponential concentration, they need not be equivalent. To address this question, we compare the expected Heisenberg-evolved observables induced by two parameter distributions directly.

\begin{definition}[Expected Heisenberg-evolved operator gap]\label{def:opdep_exp}
Let $U(\theta)$ be a parametrized quantum circuit, let $O$ be an observable, and let $\widehat{\tau}_{\vec{\alpha}}$ and $\widehat{\tau}_{\vec{\gamma}}$ be two first-moment operators associated with probability distributions $P_{\alpha}(\theta)$ and $P_{\gamma}(\theta)$. We define the expected Heisenberg-evolved operator gap as
\begin{equation}
    \xi_{\vec{\alpha},\vec{\gamma}}=\|\widehat{\tau}_{\vec{\alpha}}|O\rangle\!\rangle-\widehat{\tau}_{\vec{\gamma}}|O\rangle\!\rangle\|_1 \,.
\end{equation}
\end{definition}

When $\vec{\alpha}={\rm BP}$ (the initialization leading to a barren plateau), the quantity $\xi_{{\rm BP},\vec{\gamma}}$ measures how strongly $P_\gamma$ leads to a moment operator that remains separated from that associated to the fully concentrated barren-plateau ensemble. More generally, $\xi_{\vec{\gamma}',\vec{\gamma}}$ allows us to compare two potentially barren-plateau-free initialization strategies and ask whether they induce the same or different operator-space biases.

As for the third, more fundamental question,  one may ask: \textit{Do those differences persist to training outcomes?} That is, whether the different biases induced by the parameter distribution ultimately lead to different, or even good, optimization outcomes. In general, this depends not only on the circuit ensemble, but also on the overlap with the input state $\rho$ and on the global geometry of the loss landscape. Therefore, this question cannot be settled in full generality from first moments alone, which is why we complement our analytical framework, presented in the next section, with the numerical study in Section~\ref{sec:numerics}.

\section{Analytical results}\label{sec:an_res}

In this section we derive analytical results for the expected Heisenberg-evolved operator gap introduced above. Theorem~\ref{lem:moment-1-rotation-Pauli} characterizes the first moment of a single initialized gate as a contraction (i.e., multiplication of the moment operator by a prefactor that is less than one) followed by a rotation, while Theorem~\ref{th:sufficient_condition} uses these contraction factors to lower bound the first-moment separation from the barren-plateau ensemble. We then use these tools to compare inequivalent initializations, as well as to show how to obtain exponentially many barren plateau-free strategies.

\subsection{Lower bound on the expected Heisenberg-evolved operator gap}

We begin by presenting the following theorem. We refer the reader to Appendix~\ref{sec:twirt_rot_op} for a proof.
\begin{theorem}\label{lem:moment-1-rotation-Pauli}
    Let $U = e^{-i\theta\sigma_l/2}$, where $\theta$ is sampled according to $P_{\gamma_l}(\theta)$. Then, its associated first moment operator  can be written as a contraction matrix $\eta_{l}$ times a rotation $W_l\otimes W_l^*$. That is
\begin{equation}\label{eq:twirl_as_rot_app}
    \widehat{\tau}_{\gamma_l}  = \left(W_{l}\otimes W_{l}^*\right)\eta_{l}\,,
\end{equation}
where $W_{l}=e^{-i\phi_{l} \sigma_l/2}$, $\phi_{l} = \arccos \left(\frac{a_l}{\sqrt{a_l^2 + b_l^2}}\right)$ and 
\begin{align}
    a_l = \int P_{\gamma_l}(\th) \cos(\th) d\th \,,\quad 
    b_l =\int P_{\gamma_l}(\th)\sin(\th) d\th \,.\nonumber
\end{align}
In particular, the action of the matrix $\eta_{l}$ over Pauli operators is
\begin{equation}
    \eta_{l}\kett{\sigma_i} = 
    \begin{cases}
        \kappa_l \kett{\sigma_i} \, {\rm iff} \, \{\sigma_l,\sigma_i\} = 0\,,\\
        \kett{\sigma_i} \, {\rm iff} \, [\sigma_l,\sigma_i] = 0\,,
    \end{cases}
\end{equation}
where
\begin{equation}
    \kappa_l=\sqrt{a_l^2 + b_l^2}
\end{equation}
which shows that $\eta_{l}$ has positive eigenvalues smaller or equal to $1$, 
\end{theorem}

Theorem~\ref{lem:moment-1-rotation-Pauli} shows that the first moment operator of a single parametrized gate admits a simple operational decomposition, as it acts as a positive contraction followed by a rotation in operator space. More precisely, components of $\kett{O}$ that commute with the generator are left unchanged, whereas components that anticommute with the generator $\sigma_l$  are first suppressed by a factor  $\kappa_l$ and then rotated by $W_l$. Hence, the initialization distribution $P_{\gamma_l}(\theta)$ plays two distinct roles. On one hand, it determines how strongly averaging damps operator components. Then,  it determines the direction in operator space toward which those components are biased. Indeed, we will see next that such interpretation will be central below, as we will find that the damping factors $\kappa_l$ control our lower bound for avoiding barren plateaus, while the rotations $W_l$ can  explain why different barren-plateau-free initializations need not be equivalent.

The next theorem shows that the contraction factors $\kappa_l$ alone already control a lower bound on the expected Heisenberg-evolved operator gap (while the role of the associated rotations will become important when comparing inequivalent initialization strategies). As such, consider the following theorem whose proof can be found in Appendix~\ref{app:proof_th1}.
\begin{theorem}\label{th:sufficient_condition}
Let $U(\thv)$ be a PQC as defined in Eq.~\eqref{eq:pauliPQC} with associated distribution of parameters $P_{\vec{\gamma}}(\thv)$. The expected Heisenberg-evolved operator gap $\xi_{{\rm BP},\vec{\gamma}}$ is lower bounded as
    \begin{equation}
        \xi_{{\rm BP},\vec{\gamma}}\in\Omega\left(\kappa\right)
    \end{equation}
with
\begin{equation}
\kappa=\prod_{l=1}^L \kappa_l \,,\quad\kappa_l=\sqrt{a_l^2 + b_l^2}\,,
\end{equation}
and where 
\begin{align}\label{eq:def_a_main}
        a_l &= \int \cos(\th_l)P_{\gamma_l}(\th_l)d\th_l\,,\\
        b_l &= \int \sin(\theta_l)P_{\gamma_l}(\theta_l)d\th_l\,.\label{eq:def_b_main}
    \end{align}
\end{theorem} 
Theorem~\ref{th:sufficient_condition} is quite general and can be used to assess how a given parameter-initialization strategy affects the gap $ \xi_{{\rm BP},\vec{\gamma}}$ via the damping factors $\kappa_l$. When this lower bound is inverse-polynomial, the initialization passes our first-moment escape criterion. 

We now apply this result to several concrete distributions. As a first sanity check, we consider the trivial tensor-product PQC example studied in Ref.~\cite{mhiri2025unifying}. Although Ref.~\cite{mhiri2025unifying} analyzes this setting through second-moment operators, here we work at the level of first moments. While the two approaches do not address exactly the same question, the first moment still suffices to diagnose when barren plateaus can be avoided, while leading to significantly simpler calculations.

\begin{example}[Tensor product PQC]
    \textit{Consider a PQC composed of a tensor product of single qubit rotations
   \begin{equation}
        U(\thv) = \prod_{j=1}^{n}e^{-i\th_{1,j}\sx^{(j)}/2} e^{-i\th_{2,j}\sz^{(j)}/2}e^{-i\th_{3,j}\sx^{(j)}/2}\,,
    \end{equation}
i.e. $\GC=\{\sx^{(j)},\sz^{(j)}\}_{j=1}^n$ with $\sigma_\mu^{(j)}$ ($\mu=x,y,z$) a Pauli operator on qubit $j$, and let $O = \sigma_z^{\otimes n}$. Then, choose $P_{\gamma}(\th_j)= {\rm Unif}[- r,  r] \; \forall \; j$. Here we find for all parameters
\begin{equation}
    a= \sinc(r)\,,\quad b = 0\,,
\end{equation}
and hence 
    \begin{equation}
       \kappa = \left(\frac{1}{2r}\int_{-r}^r \cos(\th) \right)^{3n } = {\rm sinc}^{3n}(r)\,.
    \end{equation}
Note that here we removed the parameter dependence from $a,b,\kappa$ as they are all the same. Using the fact that $\sinc(r) = (\sin r)/r$, one obtains 
    \begin{equation}
        r\in\Theta\left(\frac{1}{\sqrt{n}}\right)\Rightarrow  \xi_{{\rm BP},\vec{\gamma}}\in \Omega(1)\,.
    \end{equation}}
\end{example}

Next, let us use Theorem~\ref{th:sufficient_condition} to re-derive more complex results from the literature. We begin with identity-inspired initializations~\cite{grant2019initialization,wang2023trainability, park2023hamiltonian,park2024hardware, zhang2022escaping, chang2024latent,shi2024avoiding,cao2024exploiting} and then show that the same barren-plateau scaling is achieved by a continuous family of shifted, generally non-identity initializations. As before, we stress that our technique only uses first moment operators, instead of second ones as in the original references.

\begin{example}[Close to Identity Initialization]\label{ex:reproducing_id}
\textit{Consider a PQC of the form
    \begin{equation}
        U(\thv) = \prod_{l=1}^L V \prod_{j=1}^n e^{-i\th_{j,l,x}\sigma_x^{(j)}/2}e^{-i\th_{j,l,y}\sigma_y^{(j)}/2}\,,
    \end{equation}
where  $\GC=\{\sx^j,\sy^j\}_{j=1}^n$,  $V$ is an entangling layer composed of a ladder of $C_z$ gates, and where $O = \sum_i \sigma_i$. Consider sampling the parameters independent and identically distributed (i.i.d.) from  ${\rm Unif}[- r,  r]$, then
\begin{equation}
    a = {\rm sinc}(r)\,,\quad b = 0\,, 
\end{equation}
which means we can bound the expected Heisenberg-evolved operator gap as
\begin{equation}\label{eq:lowerbound_reproducing_id}
        \xi_{{\rm BP},\vec{\gamma}} \in \Omega\left( a^{2 L n}\right)\,.
    \end{equation}
Thus, and similarly to the example above, we find
    \begin{equation}
        r\in\Theta\left(\frac{1}{\sqrt{Ln}}\right)\Rightarrow  \xi_{{\rm BP},\vec{\gamma}}\in \Omega(1)\,.
    \end{equation}}
\end{example}

\begin{example}[Gaussian initialization]\label{example-Gaussian}
\textit{Consider the same circuit and measurement operator as the ones presented in Example~\ref{ex:reproducing_id}. However, instead of considering a uniform distribution over a smaller region of the space, we assume that all the parameters are sampled i.i.d. from a normal (Gaussian) distribution with a mean $\mu = 0$ and a variance $\varsigma$, i.e., $\NC(0,\varsigma^2)$. Then, we find  for all parameters
\begin{equation}
    a = e^{-\varsigma^2/2}\,,\quad b=0\,, 
\end{equation}
 and hence
    \begin{equation}
       \xi_{{\rm BP},\vec{\gamma}}\in \Omega\left(e^{-Ln\varsigma^2}\right)\,,
    \end{equation}
from where
\begin{equation}
    \varsigma\in\Theta\left(\frac{1}{\sqrt{nL}}\right)\Rightarrow\xi_{{\rm BP},\vec{\gamma}}\in \Omega\left(1\right)\,.
\end{equation}}
\end{example}

A natural interpretation of Examples~\ref{ex:reproducing_id} and~\ref{example-Gaussian} is that avoiding barren plateaus may simply amount to centering the parameters around zero (and thus potentially simply initializing near the identity). However, this strategy does not always lead to avoiding barren plateaus~\cite{mhiri2025unifying}, and as shown by our next example is too restrictive of an interpretation.

\begin{example}[Initialization that is not centered around zero.]\label{example-not-centered}
\textit{Consider again the circuit and observables in Ex.~\ref{ex:reproducing_id}, but assume that all parameters are sampled i.i.d. from  ${\rm Unif}[\thv^*-r,\thv^*+r]$ for some $\thv^*$. Under these constraints we find for all parameters
\begin{align}
    a = \cos (\th^* )\sinc (r)\,, \quad 
    b  = \sin (\th^* ) \sinc (r)\, ,
\end{align}
leading to
\begin{equation}
     \xi_{{\rm BP},\vec{\gamma}}\in \Omega\left( \sinc^{2Ln}(r) \right)\, \,\, \, \forall \th^*  .
\end{equation}
Indeed, this is the exact same lower bound as that in Eq.~\eqref{eq:lowerbound_reproducing_id}. Thus we again find that
\begin{equation}
        r\in\Theta\left(\frac{1}{\sqrt{Ln}}\right)\Rightarrow  \xi_{{\rm BP},\vec{\gamma}}\in \Omega(1) \, .
\end{equation}}
\end{example}

Here, we find it important to highlight the significance of Example~\ref{example-not-centered}. In particular, we can see that although the coefficients $a$ and $b$ depend on the center $\theta^\ast$, the damping factor $\kappa=\sqrt{a^2+b^2}$, and hence the lower bound from Theorem~\ref{th:sufficient_condition}, depends only on the width $r$. This does not mean that all these points will avoid exponential concentration (as one can expect that the overlap with the initial state can be very small for most values of $\thv^*$).

To finish, we use  Theorem~\ref{th:sufficient_condition} to derive bounds for new potential initialization strategies. In particular, consider the following case.

\begin{example}[Sinusoidal initialization strategy]
\textit{Consider the PQC 
\begin{equation}
    U(\thv) = \prod_{l=1}^L V_l \prod_{j=1}^n e^{-i\th_{j,l,x}\sigma_x^{(j)}/2}e^{-i\th_{j,l,y}\sigma_y^{(j)}/2}e^{-i\th_{j,l,z}\sigma_z^{(j)}/2}\,,
\end{equation}
where $V_l$ are arbitrary non parametrized gates, a Pauli observable $O$, and the following probability distribution for all parameters
\begin{equation}
    P_{\gamma}(\theta)
    = \frac{\sqrt{\pi}\,\Gamma\!\left[\frac{\gamma+1}{2}\right]}
           {\Gamma\!\left[1+\frac{\gamma}{2}\right]}
      \sin^{\gamma}(\theta)\,,
    \qquad \theta\in[0,\pi]\,,
\end{equation}
where $\gamma>0$, and where $\Gamma(x) = \int_0^\infty t^{x-1}e^{-t}dt$ for any $x>0$ denotes the Gamma function. Then, we find
\begin{align}
    a &= 0\,,\quad
    b = \frac{\Gamma^2\!\left[\frac{\gamma+2}{2}\right]}
             {\Gamma\!\left[\frac{\gamma+1}{2}\right]
              \Gamma\!\left[\frac{\gamma+3}{2}\right]}\,.
\end{align}
and therefore  $\xi_{{\rm BP},\vec{\gamma}}\in\Omega(\kappa)$ with
\begin{equation}
    \kappa =  b^{3L n} = \left(\frac{\Gamma^2 \left[\frac{\gamma+2}{2}\right]}{\Gamma \left[\frac{\gamma+1}{2}\right] \Gamma\left[\frac{\gamma+3}{2}\right]}\right)^{3Ln}\,.
\end{equation}
Using that $\Gamma\left[\frac{\gamma+3}{2}\right] = \frac{\gamma+1}{2}\Gamma\left[\frac{\gamma+1}{2}\right]$, as well as Gautschi's inequality~\cite{wendel1948note} $\Gamma[\gamma/2+1]/\Gamma[\gamma/2+1/2]>\sqrt{\gamma/2}$, we can simplify the statement to
\begin{align}
    \kappa > \left(\frac{\gamma}{\gamma+1}\right)^{3Ln}\in\Omega(1)
\end{align}
whenever $\gamma\in\Theta(L n)$.}
\end{example}

We conclude this section by emphasizing both what Theorem~\ref{th:sufficient_condition} establishes and what it does not. On the one hand, the condition $\xi_{{\rm BP},\vec{\gamma}}\in\Omega(1)$ provides a simple and highly general certificate that an initialization strategy can avoid the fully concentrated barren-plateau regime, and the examples above show that this certificate is satisfied by many distinct families of parameter distributions  proposed in  the literature, not just by a single special choice. On the other hand, this condition is only a partial certificate, as guaranteeing that gradients do not vanish exponentially also requires accounting for the overlap between the twirled operators and the initial state $\rho$, so that Eq.~\eqref{eq:necessary} is satisfied. Determining when this overlap remains sufficiently large cannot be done in full generality, as it depends on the specific choice of initial state, observable, circuit, and initialization strategy. However, below we will present a specific practical case where a simple initial state can have large overlap with the expected Heisenberg-evolved
operators, and thus lead to non-concentrated landscapes, for exponentially many initializations.

Nevertheless, our lower bound on the expected Heisenberg-evolved operator gap already makes one central point clear: avoiding barren plateaus might not remove the initialization problem, but rather transforms it. Indeed, there may be many families of initializations that avoid exponential concentration, and the relevant question then becomes whether these strategies are effectively equivalent or whether one must choose among them because they bias the optimization toward different regions of the landscape. 

In this sense, our framework does more than recover known special initializations, it provides a systematic way of identifying broad classes of candidate barren-plateau-free strategies, while simultaneously motivating the need to understand how these different choices affect trainability beyond the initial gradient scale. In particular, once many barren-plateau-free initializations are available, trainability may no longer be limited by finding \emph{an} initialization with non-exponentially-concentrated loss function, but by deciding \emph{which} such initialization to use.

\subsection{Comparing initialization strategies}\label{sec:different_strategies}

The previous section showed that Theorem~\ref{th:sufficient_condition} can certify many distinct parameter distributions for which $\xi_{{\rm BP},\gamma}\in\Omega(1)$, and thus that avoiding barren plateaus is already a non-unique problem. The next question is whether these candidate initializations are effectively interchangeable.  Theorem~\ref{lem:moment-1-rotation-Pauli} suggests that they are not. In particular, two initialization strategies may have comparable damping factors $\kappa_l$, and hence similar barren-plateau scaling, while still inducing different rotations and therefore different operator-space biases. This motivates us to study the expected Heisenberg-evolved operator gap $\xi_{\gamma',\gamma}$ between two initialization strategies. A large value of $\xi_{\gamma',\gamma}$ means that the two strategies induce different expected observables and should therefore be regarded as different inductive biases on the optimization landscape.

To begin, we present a trivial single-qubit example of how two different non-concentrating initializations lead to different expected operators.

\begin{example}[Single-qubit PQC initialization strategies]\label{ex:trivial_diff_strategy}
\textit{Consider a single qubit PQC  $ U(\th) = e^{-i\theta\sigma_x/2}$, and take $O = \sigma_z$. Then, assume two deterministic initialization strategies such that in the first we always initialize to $\theta = 0$ and in the second to $\theta = \pi/2$: i.e., $P_1(\th) = \delta(\th)$ and $P_2(\th) = \delta(\th-\pi/2)$. This leads to $a_1 = 1,\, b_1 = 0$ and $a_2 = 0,\, b_2 = 1$ and  
\begin{align}
    \widehat{\tau}_1\kett{\sigma_z} =& \kett{\sigma_z}\, ,\\
    \widehat{\tau}_2\kett{\sigma_z} =-& \kett{\sigma_y}\, .
\end{align}
Clearly, there will be initial states for which the expectation values differ, meaning that the two  initialization strategies significantly bias the model to different landscape regions. This in-equivalence can be diagnosed by the expected Heisenberg-evolved operator gap   $\xi_{1,2} = \| \kett{\sz} +\kett{\sy}\|_1 = 2$. }
\end{example}

Example~\ref{ex:trivial_diff_strategy} is particularly useful because it isolates the role of the rotation in Theorem~\ref{lem:moment-1-rotation-Pauli}. Indeed, both initialization strategies are deterministic, so there is no averaging-induced damping to distinguish them; instead, their inequivalence arises entirely from the fact that they rotate the observable toward different Pauli directions. In this sense, the example shows explicitly that two initialization strategies can have comparable barren-plateau scaling while still inducing different operator-space biases through different first-moment rotations. In addition, at this point we also note that in Example~\ref{ex:trivial_diff_strategy} one distribution leads to $a=0$, while the other one to $b=0$. This distinction can be formalized via the helpful notion of symmetric and antisymmetric initialization strategies.

\begin{definition}[Symmetric/antisymmetric initialization strategy]\label{def:symmetric_anti}
    Let $P_{\gamma}(\theta)$ be an initialization strategy. We will say that it is symmetric when
    \begin{equation}
    b_\gamma = 0 \,,
    \end{equation}
    and antisymmetric if 
    \begin{equation}
     a_\gamma = 0\, .
    \end{equation}
\end{definition}

As such, in Example~\ref{ex:trivial_diff_strategy} we would say that $P_{1}$ is symmetric, whereas $P_{2}$ is antisymmetric. Importantly, we note that this distinction is not merely semantic. Example~\ref{ex:trivial_diff_strategy} already shows that symmetric and antisymmetric initialization strategies can induce markedly different expected observables, and hence substantially different inductive biases in the PQC. More generally, these two classes provide elementary local choices that can be assigned gate by gate, and therefore a natural route to constructing broad families of initialization strategies. Indeed, related notions have also appeared in the context of diagnosing barren plateaus~\cite{heyraud2023efficient}. In Table~\ref{tab:coefs_main} we summarize several representative examples of symmetric and antisymmetric examples, but we highlight that it is by no means comprehensive.

Then, we note that in Appendix~\ref{ap-sec:symmetric-strategis} we present a general recipe to compute $\widehat{\tau}_{\vec{\gamma}_{\rm sym}}$ when all PQC parameters are initialized according to symmetric distributions. In what follows, we illustrate the power of this framework by computing $\xi_{\gamma',\gamma}$ in three representative settings: when both initializations are symmetric, and when one is symmetric and the other is the identity initialization.  We refer the reader to Appendix~\ref{ap-sec:examples} for additional details.

\renewcommand{\arraystretch}{1.8}
\begin{table*}[t]
    \centering
    \begin{tabular}{|c|c|c|c|}
    \hline
        Initialization strategy & $a$ & $b$ & $\kappa=\sqrt{a^2+b^2}$\\
        \hline\hline
         $\sin(\theta)$, $\theta\in[0,\pi]$ & $0$ & $\frac{\pi}{4}$ & $\frac{\pi}{4}$ \\ 
         Unif$[- r, r]$ & $\frac{\sin r}{r}$ & $0$ & $\frac{\sin r}{r}$\\
         Unif$[\theta^\ast- r,\theta^\ast+r]$ & $\frac{\cos \theta^\ast \sin r}{r}$ & $\frac{\sin \theta^\ast \sin r}{r}$ &$\frac{\sin r}{r}$ \\
         $\mathcal{N}(0,\sigma)$ & $e^{-\sigma^2/2}$ & $0$  & $e^{-\sigma^2/2}$\\
         $\mathcal{N}(\mu,\sigma)$ & $e^{-\sigma^2/2}\cos\mu$ & $e^{-\sigma^2/2}\sin\mu$ & $e^{-\sigma^2/2}$\\
         Beta$\left( \frac{\theta}{2\pi}, \alpha, \beta \right)$ & ${}_2F_3^{(1)}$ &  $\pi ^2 \alpha  2^{-\alpha -\beta } \Gamma (\alpha +\beta ) \, {}_2\tilde{F}_3^{(2)}$ & $\sqrt{\left( {}_2F_3^{(1)} \right)^2 + \left( \pi ^2 \alpha  2^{-\alpha -\beta } \Gamma (\alpha +\beta ) \, {}_2\tilde{F}_3^{(2)} \right)^2}$\\
         \hline
    \end{tabular}
    \caption{\textbf{Representative initialization families.} Different initializations can be characterized through the first-moment coefficients $a=\int P_\gamma(\theta)\cos\theta\,d\theta$ and $b=\int P_\gamma(\theta)\sin\theta\,d\theta$ of Theorem~\ref{lem:moment-1-rotation-Pauli}, which in turn define the associated damping factor $\kappa=\sqrt{a^2+b^2}$. Symmetric and antisymmetric initializations correspond only to the special cases $b=0$ and $a=0$. More generally, many shifted or biased distributions have both $a\neq 0$ and $b\neq 0$, showing that the space of candidate barren-plateau-free initializations is substantially richer than this binary distinction. We use ${}_2F_3^{(1)} = {}_2F_3\left(\frac{\alpha }{2}+\frac{1}{2},\frac{\alpha }{2};\frac{1}{2},\frac{\alpha }{2}+\frac{\beta }{2},\frac{\alpha }{2}+\frac{\beta }{2}+\frac{1}{2};-\pi ^2\right)$ and ${}_2\tilde{F}_3^{(2)} = {}_2\tilde{F}_3\left(\frac{\alpha +1}{2},\frac{\alpha +2}{2};\frac{3}{2},\frac{1}{2} (\alpha +\beta +1),\frac{1}{2} (\alpha +\beta +2);-\pi^2\right)$, where ${}_2F_3$ denotes the Generalized hyper-geometric function. }
    \label{tab:coefs_main}
\end{table*}

\begin{example}[Symmetric initializations in the Hamiltonian variational ansatz]\label{ex:HVA}
\textit{Consider a Hamiltonian variational ansatz~\cite{wecker2015progress,park2023hamiltonian} for a one-dimensional $XYZ$ Heisenberg chain. Here we have $\GC=\{\sx^{(i)}\sx^{(i+1)},\sy^{(i)}\sy^{(i+1)},\sz^{(i)}\sz^{(i+1)}\}_{i=1}^{n-1}\cup\{\sz^{(i)}\}_{i=1}^n$. Next, assume that all the parameters are initialized according to the same symmetric initialization strategy $P_{\gamma_{\rm sym}}(\theta)$, and that the observable takes the form $O=\sum_i \sx^{(i)}\sy^{(i+1)}$. Then, as we  show in Appendix~\ref{app:hva}, we can express $\widehat{\tau}_{\vec{\gamma}_{\rm sym}}\kett{O}$ as
    \begin{align}
        \widehat{\tau}_{\vec{\gamma}_{\rm sym}}\kett{O} = &\sum_{i = 2}^{n-2}a^{9 L}\kett{\sx^{(i)}\sy^{(i+1)}} + \nonumber\\
        &a^{8 L}\kett{\sx^{(1)}\sy^{(2)}}+a^{8 L}\kett{\sx^{(n-1)}\sy^{(n)}}\,,\label{eq:sym-id-tau}
    \end{align}
    where $\vec{\gamma}_{\rm sym}$ is such that $a_{\vec{\gamma}_{\rm sym}}\neq0$ and $b_{\vec{\gamma}_{\rm sym}} = 0$.     Then, take a second initialization $P_{\vec{\gamma'}_{\rm sym}}(\thv)$ where all gates are initialized with $P_{\gamma_{\rm sym}}(\th)$, except for the one gate with parameter $\th_1$ generated by $\sz^{(1)}$ in the last layer, for which we use $P_{\gamma'_{\rm sym}}(\th_1)$. Then if we call $a' = \int P_{\gamma'_{\rm sym}}(\th_1) \cos \theta_1$, we  find
     \begin{align}
            \widehat{\tau}_{\vec{\gamma}'_{\rm sym}}\kett{O} = &\sum_{i = 2}^{n-2}a^{9 L}\kett{\sx^{(i)}\sy^{(i+1)}} + \\
            &a^{7 L}(a')^L\kett{\sx^{(1)}\sy^{(2)}}+a^{8 L}\kett{\sx^{(n-1)}\sy^{(n)}}\,.
    \end{align}
This means that the expected Heisenberg-evolved operator gap between these two probability distributions is
   \begin{align}\label{eq:sim-vs-sim}
    \xi_{\vec{\gamma}',\vec{\gamma}}
    = \|\widehat{\tau}_{\vec{\gamma}'_{\rm sym}}\kett{O}
      -\widehat{\tau}_{\vec{\gamma}_{\rm sym}}\kett{O}\|_1
    = a^{7L}\left|a^L-(a')^L\right|\,.
\end{align} 
Note that in this case, $\xi_{\vec{\gamma}',\vec{\gamma}}$ is 
 \begin{align}
     \xi_{\vec{\gamma}',\vec{\gamma}} = |\xi_{{\rm BP}, \gamma} - \xi_{{\rm BP}, \gamma'}|
 \end{align}
 and therefore, if the initializations do not concentrate (and $a^L\not=(a')^L$) we can guarantee that
 \begin{align}
     \xi_{\vec{\gamma}',\vec{\gamma}}\in\Omega(1)\,.
 \end{align}}

\end{example}
In contrast with Example~\ref{ex:trivial_diff_strategy}, where inequivalence arises purely from the rotational component of Theorem~\ref{lem:moment-1-rotation-Pauli}, Example~\ref{ex:HVA} shows that even within the symmetric class, different initialization strategies can still be inequivalent at the level of the expected Heisenberg-evolved observable and through the different contraction weights they assign to the surviving operator components. We can readily extend our analysis to compare a symmetric initialization with an identity initialization~\cite{grant2019initialization}, where all parameters are simply set to zero.

\begin{example}[Symmetric and identity initializations in the Hamiltonian variational ansatz]\label{ex:HVA-identity}
\textit{Consider the Hamiltonian variational ansatz of the previous example. Then, we can readily combine Eq.~\eqref{eq:sym-id-tau} with the fact that under an identity initialization $\widehat{\tau}_{\rm id}=\id\otimes \id$ (so that for any $O$ one has $\widehat{\tau}_{\rm id}\kett{O}=\kett{O})$ to find
\begin{align}
    \xi_{{\rm id},\vec{\gamma}}
    = \|\widehat{\tau}_{\rm id}\kett{O}
      -\widehat{\tau}_{\vec{\gamma}}\kett{O}\|_1
    = a^{7L}\left|a^L-1\right|\,.
\end{align}
Similarly, we find that if the two distributions avoid BPs, they are different
 \begin{align}
    \xi_{\vec{\gamma}',\vec{\gamma}} = |\xi_{{\rm BP}, \gamma} - \xi_{{\rm BP}, \gamma'}|\in\Omega(1)\,.
 \end{align}}
\end{example}

\subsection{Exponentially many initializations}

In this section we explicitly show how standard barren-plateau-free initialization strategies can be readily promoted to exponentially large families from which to choose from. For this purpose, let us consider an $n$-qubit PQC composed of general single-qubit gates interleaved with diagonal entangling gates (see for example Fig.~\ref{fig:qckt_ela}). Next, assume that each parameter in the single-qubit gates is sampled with respect to  $\NC(\theta^*,\varsigma^2)$, where is $\theta^*\in\{0,\pi\}$. That is, each parameter is sampled from a Gaussian centered at either $0$ or $\pi$. Since the center can be chosen independently for each gate, this produces an exponentially large family of initialization strategies.

\begin{figure}
    \centering
    \includegraphics[width=\linewidth]{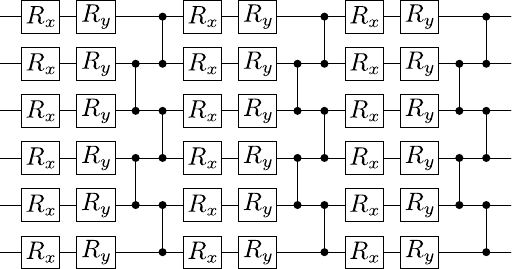}
    \caption{\textbf{Circuit ansatz.} We consider a PQC composed of single-qubit $R_x$ and $R_y$ rotations followed by nearest-neighbor $CZ$ entangling gates. The example shown corresponds to $n=6$ qubits and $L=3$ layers.}
    \label{fig:qckt_ela}
\end{figure}

Then, we can show that the following Theorem holds (see Appendix~\ref{app-prood-theo3} for a proof).
\begin{theorem}\label{theo:bp-free}
    Let $U(\thv)$ be a PQC composed of general single-qubit gates interleaved with diagonal entangling gates $V_l$, i.e., 
       \begin{equation}\nonumber
        U(\thv) =\prod_{l=1}^L \prod_{j=1}^{n}e^{-i\th_{1,j,l}\sx^{(j)}/2} e^{-i\th_{2,j,l}\sz^{(j)}/2}e^{-i\th_{3,j,l}\sx^{(j)}/2} V_l\,,
    \end{equation}
    let $O$ be sum of local Pauli operators acting on $\OC(1)$ qubits, and $\rho$ the all-zero state, $\ketbra{0}^{\otimes n}$. Then, for all exponentially many initialization choices such that the single-qubit gate's parameters are independently sampled with respect to  $\NC(\theta^*,\varsigma^2)$, where $\theta^*\in\{0,\pi\}$ and $ \varsigma\in\Theta\left(\frac{1}{\sqrt{L}}\right)$, the model is barren plateau-free, as
\begin{equation}\label{eq:exp_concentration}
    \Var_{\bm{\gamma}}[\LC(\thv)]\in\Omega\left(\frac{1}{\poly(n)}\right)\,.
\end{equation}
\end{theorem}
Here we find it important to note that Theorem~\ref{theo:bp-free} generalizes the Gaussian initialization strategy of~\cite{zhang2022escaping} (where $\theta^*=0$  for all gates) to an exponentially large family of non-concentrating strategies. Moreover, as discussed in the Appendix, the Gaussian choice for the parameter distributions is not essential. The same argument applies, for example, to uniform windows  ${\rm Unif}[\thv^*-r,\thv^*+r]$ with $r\in\Theta\left(\frac{1}{\sqrt{L}}\right)$ and $\theta^*\in\{0,\pi\}$.

An alternative route to a somewhat analogous conclusion might be drawn from prior work on local minima in barren plateau landscapes~\cite{nemkov2024barren, fontana2022nontrivial} and variance guarantees for warm starts~\cite{mhiri2025unifying}. Namely, for PQCs composed of entangling Clifford gates and local Pauli rotations, and for observables given by sums of local Pauli operators, the landscape has both a barren plateau and exponentially many local minima~\cite{nemkov2024barren, fontana2022nontrivial}. Ref.~\cite{mhiri2025unifying} then shows that, around any local minimum, if one restricts the initialization to a small angle range, the resulting initialization is barren-plateau free.
Putting these two facts together, we see that there are exponentially many regions of parameter space around which barren-plateau-free initializations can be found. While we include this argument to build an intuition, we stress that the proof of Theorem~\ref{theo:bp-free} is more subtle and is found by explicitly identifying points (that need not be, and in general will not be, local minima) with guaranteed gradients via generalizing the proof of Ref.~\cite{zhang2022escaping}. 

At a more conceptual level, Theorem~\ref{theo:bp-free} proves the existence of exponentially many barren-plateau-free initializations, but it does not guarantee nor tell us whether they are all inequivalent. Notably, our first-moment framework can help answer this question when we focus on a specific problem and  observable $O$. Thus, in the next example, we instantiate our results for the task of training a PQC to solve a Maximum Cut (MaxCut) problem~\cite{farhi2014quantum,kazi2022landscape}.

\begin{example}[Exponentially many initializations for MaxCut]\label{ex:MaxCut}
\textit{Consider a parameterized quantum circuit applied to a Max-Cut problem on a graph $G=(V,E)$, where $O= \sum_{(w,w')\in E}\sigma_z^{(w)}\sigma_z^{(w')}$. Then, we consider an $L$- layered PQC composed of local $R_y$ rotations interleaved with diagonal entangling gates $D$, given by
\begin{equation}
    U(\vec{\theta}) = V_{L+1}\left(\theta^V_{L+1}\right) \prod_{l=1}^{L} \left( \prod_{j=1}^{n} \text{R}_y^{(j)}\left(\theta^y_{l,j}\right)\, V_{l,j}\left(\theta^V_{l,j}\right) \right)\,.
\end{equation}
Let $\vec{\gamma}_{\vec{c}}$ be a combinatorial family of initializations governed by an $L \times n$ binary matrix $\vec{c} \in \{0,1\}^{L \times n}$. For a given $\vec{c}$, the initialization $\vec{\gamma}_{\vec{c}}$ draws the parameter for the $R_y$ gate on qubit $j$ at layer $l$ from a Gaussian distribution with mean $c_{l,j} \pi$, and variance $\varsigma$. Let $C_w(\vec{c}) = \sum_{l=1}^L c_{l,w} \pmod 2$ denote the total parity of the shifts applied to qubit $w$. As shown in Appendix~\ref{app:MaxCut}, we find
\begin{equation}
    \widehat{\tau}_{\vec{\gamma}_{\vec{c}}}|O\rangle\!\rangle = a^{2L}\sum_{(w,w')\in E}(-1)^{C_w(\vec{c}) + C_{w'}(\vec{c})}|\sigma_z^{(w)}\sigma_z^{(w')}\rangle\!\rangle\,.
\end{equation}
Consequently, for any two distinct initializations $\vec{\gamma}_{\vec{c}}$ and $\vec{\gamma}_{\vec{c}'}$, the operator gap is exactly determined by the edges where the parity assignments differ. Defining the edge parity as $s_e(\vec{c}) = C_w(\vec{c}) \oplus C_{w'}(\vec{c})$ for an edge $e=(w,w')$, we have:
\begin{equation}\label{eq:sim-vs-sim-maxcut}
    \xi_{\vec{\gamma}_{\vec{c}}, \vec{\gamma}_{\vec{c}'}} = 2 a^{2L} \big| \{ e \in E \mid s_e(\vec{c}) \neq s_e(\vec{c}') \} \big|\,.
\end{equation}
One can see that this strategy generates $2^{n-1}$ structurally distinct initializations for an all-to-all graph. Crucially, for a choice of $\varsigma$ that ensures $\xi_{\vec{\gamma}_{\vec{c}},{\rm BP}} \in\Omega(1)$, the separation between any two distinct initializations also does not vanish, then 
\begin{align}
\xi_{\vec{\gamma}_{\vec{c}},\vec{\gamma}_{\vec{c}'}}\in \Omega(1)\,.
\end{align}}
\end{example}

Theorem~\ref{theo:bp-free} applies also to Example~\ref{ex:MaxCut} and thus we demonstrate the existence of an exponential number of mutually distinct initializations that successfully avoid exponential loss function concentration. We refer the reader to the Appendix where we further discuss the effects of the initialization strategy around the local landscape.

Going beyond this, when considering Example~\ref{ex:MaxCut} for the Max-Cut Hamiltonian we can extend the analysis and find that certain parameters in the circuit can be given a much larger freedom in the initialization, while still avoiding barren plateaus  (as we show in Appendix~\ref{app:1d-slices}). In particular, we prove that the following theorem holds.

\begin{theorem}[Informal] \label{th:MaxCut-slices}
     Let $U(\thv)$ be a PQC composed of single-qubit gates interleaved with diagonal entangling gates $V_l$, i.e., 
       \begin{equation}\nonumber
        U(\thv) =\prod_{l=1}^L \prod_{j=1}^{n}e^{-i\th_{j,l}^z\sz^{(j)}/2} e^{-i\th_{j,l}^y\sy^{(j)}/2} V_l\,,
    \end{equation}
    where $L\in\OC\left(\poly(n)\right)$. Moreover, let $\rho=\ketbra{0}^{\otimes n}$ and let $O= \sum_{(w,w')\in E}\sigma_z^{(w)}\sigma_z^{(w')}$ be the Max-Cut Hamiltonian on a constant-odd-degree graph $G=(E,V)$, with $S\subseteq V$  a subset of vertices. Then, for each $j\in S$, let one $ \theta_{l_j,j}^y$ angle be initialized in a constant-size region, while the rest of the parameters are independently initialized following  Gaussian distributions $\NC(\theta^*,\varsigma^2)$, with $\theta^*\in\{0,\pi\}$ and $ \varsigma\in\Theta\left(\frac{1}{\poly(n)}\right)$. In this setting, all the $\theta^y_{l,j}$ parameters in the circuit are guaranteed to have gradients in $\Omega\left(\frac{1}{\poly(n)}\right)$ with high probability. 
\end{theorem}

Interestingly, Theorem~\ref{th:MaxCut-slices} shows that for certain graph families, one can pick a subset of $\Theta(n)$ $R_y$ angles and freely initialize them in a region whose volume in parameter space~\footnote{Notice that here we are referring to the volume of the $\Theta(n)$-dimensional parameter subspace itself, and not to the entire parameter space.} does not shrink with growing system size (in contrast to Theorem~\ref{theo:bp-free}). Having made this choice, there are then exponentially many choices for how to initialize the remaining parameters in the circuit such that the initialization provably avoids barren plateaus. 
Notice as well that the subset $S$ could be chosen as an independent vertex set of size $\Theta(n)$ and the rest of the parameters be fixed to $0$ or $\pi$, in which case one would find exponentially many initialization strategies that provably avoid barren plateaus, all of which lead to trivially decoupled optimization dynamics (as here each $R_y$ rotation would only affect different terms in the Hamiltonian). This latter example showcases a non-trivial problem (we recall that the Max-Cut problem remains NP-hard~\cite{berman2002some}) where care must actually be taken when choosing an initialization strategy, as there can exist many useless regions in the landscape with inverse polynomial gradients.

To finish, we note that up to this point we have generated exponentially large families of strategies purely based on symmetric initializations. However, the next example shows that changing a parameter's distribution from symmetric to antisymmetric can also lead to exponentially many choices.

\begin{example}[A local symmetric/antisymmetric toggle]\label{ex:different_initialisation_different_points} 
\textit{Consider a circuit composed of two layers of $C_z$ entangling gates and single-qubit parameterized gates, acting on the $j$-th qubit, of the form $U_{(j)}(\thv) = e^{-i \theta_{1} \sx^{(j)}/2}e^{-i \theta_{2} \sy^{(j)}/2}e^{-i \theta_{3} \sz^{(j)}/2}$. The total circuit looks like
\begin{equation}
    U(\thv) = \prod_{l=1}^{L}\prod_{i=1}^{n-2} C_{Z}^{(2i,2i+1)}C_{Z}^{(2i+1,2i+2)}\prod_{j=1}^n U_{(j)}(\thv_{j,l})\,.
\end{equation}
Then, consider $O= \sum_{i=1}^{n-1}\sigma_x^{(i)}\sigma_x^{(i+1)}$. To finish, we will consider the probability distributions $\PC_{\vec{\gamma'}}(\thv)  = P_{\bm{\gamma}_{\, \rm sym}}(\thv)$ and $\PC_{\vec{\gamma}}(\thv)$, obtained by initializing all the qubits as per $\PC_{\vec{\gamma'}}(\thv)$ except for the top-most qubit in the last layer, which we initialize with an antisymmetric distribution. Then, as detailed in Appendix~\ref{app:different_initialisation_different_points} we find
\begin{align}\label{eq:sim-vs-antisym}
    \xi_{\vec{\gamma}',\vec{\gamma}} = (|a| + |b|) a^{11L-1}\,.
\end{align}
Finally, for a choice of $P_{\vec{\gamma}_{\rm sym}}$ such that $a^{11L},ba^{11L-1}\in\Omega(1)$, guaranteeing $\xi_{\vec{\gamma}',{\rm BP}},\xi_{\vec{\gamma},{\rm BP}}\in\Omega(1)$, we have that 
\begin{align}
     \xi_{\vec{\gamma}',\vec{\gamma}}\in\Omega(1)\,.
\end{align}}
\end{example} 

\begin{figure*}[t]
    \centering
    \includegraphics[width=2\columnwidth]{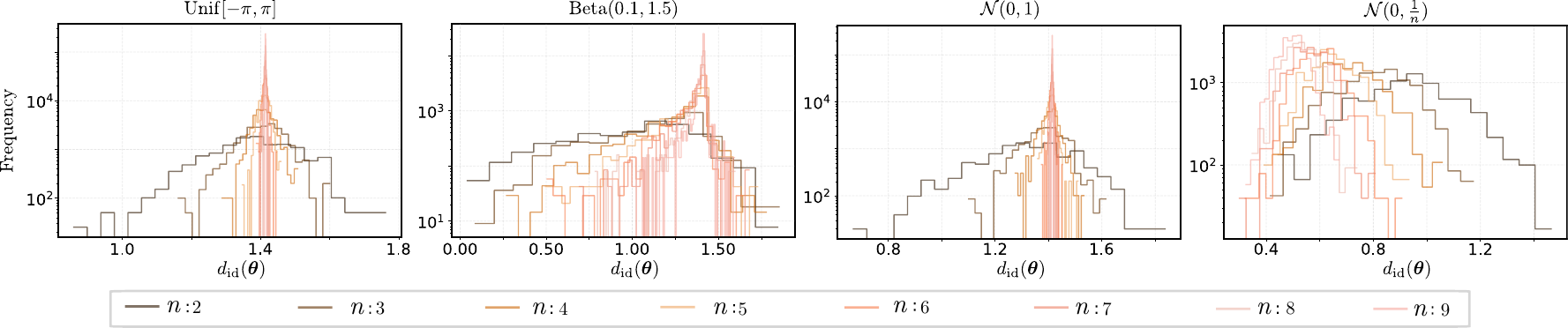}
\caption{\textbf{Distribution of  $d_{{\rm id}}(\bm{\theta})$ for several initialization families and system sizes}. We consider an $L=4$ layered circuit (see Fig.~\ref{fig:qckt_ela}), where the parameters are initialized i.i.d. from a uniform distribution $\mathrm{Unif}[-\pi,\pi]$, a Beta distribution $\mathrm{Beta}(\alpha,\beta)$, and two Gaussian initializations $\NC(0,\sigma)$, with $\sigma=1$ and $\sigma=1/\sqrt{n}$.}
    \label{fig:tracenorm_expts_big}
\end{figure*}

Example~\ref{ex:different_initialisation_different_points}  shows that the symmetric-to-antisymmetric substitution acts as a local binary toggle that replaces the mean choice  of the Gaussian initialization of  Theorem~\ref{theo:bp-free}. Indeed, changing the initialization of one gate from symmetric to antisymmetric already yields a nonzero expected Heisenberg-evolved operator gap. The same reasoning can be repeated for any other suitably chosen qubit, not only the first qubit on the last layer. In particular, if we select a set $S$ of qubits and, for each $j\in S$, independently choose whether the corresponding $z$-generated gate is initialized symmetrically or antisymmetrically, then we obtain a family $\{P_{\mathbf{s}}\}_{\mathbf{s}\in\{0,1\}^{|S|}}$ of $2^{|S|}$ initialization strategies. Since flipping any one bit in $\mathbf{s}$ changes the expected Heisenberg-evolved observable on a local Pauli component, distinct bit strings generically lead to $\xi_{\mathbf{s}',\mathbf{s}}>0$. Therefore, once $|S|=\Omega(n)$, the circuit admits exponentially many inequivalent initialization strategies. Moreover, if the local choices at each toggled gate are taken from families that each satisfy the conditions of Theorem~\ref{th:sufficient_condition}, then this construction yields exponentially many inequivalent barren-plateau-free initialization strategies.

Here we note that while the previous construction uses the centered/non-centered or symmetric/antisymmetric distinction as a straightforward mechanism for generating many inequivalent initialization strategies,  the available design space is considerably richer. In general, an initialization strategy is characterized by the pair $(a,b)$ introduced in Theorem~\ref{lem:moment-1-rotation-Pauli}, with symmetric and antisymmetric distributions corresponding only to the special cases $b=0$ and $a=0$, respectively. More broadly, there exist many families of distributions---including shifted, biased, and otherwise non-symmetric choices---for which both coefficients are nonzero.

Taken together, the previous examples show that once barren-plateau-free initializations are viewed as gatewise design choices, inequivalence becomes quite generic. In particular, the expected Heisenberg-evolved operator gap $\xi_{\gamma',\gamma}$ reveals that initialization strategies with comparable barren-plateau behavior can nevertheless steer the circuit toward different regions of the landscape. However, this operator-level inequivalence alone does not determine whether the corresponding strategies perform the same in practice. To address that question, one must study the final performance after training. This is precisely the goal of our numerical analysis, presented in the next section.

\section{Numerical Results}\label{sec:numerics}

As discussed above, our analytical framework provides necessary, but not sufficient, conditions for identifying initialization strategies that avoid barren plateaus and for certifying that two such strategies are operator-level inequivalent. The purpose of our numerical analysis is to determine whether these analytically distinct strategies result in practical differences in the training process: \textit{Given two distributions $P_{\vec{\alpha}}(\vec{\theta})$ and $P_{\vec{\gamma}}(\vec{\theta})$ with $\xi_{{\rm BP},\vec{\alpha}},\xi_{{\rm BP},\vec{\gamma}}\in\Omega(1)$, and $\xi_{\vec{\alpha},\vec{\gamma}}\in\Omega(1)$, do they lead to different final minima?} 

Here we note that while our results will mostly focus on the final minima quality, in Appendix~\ref{app-additional-numerics} we present additional simulations that further explore the local landscape properties around initialization and throughout training.

\subsection{Changing initializing distribution does not imply identity initialization}

Given the prominence and relevance of initialization the circuit to an identity~\cite{grant2019initialization}, it is natural to use the identity as a reference point when assessing other parameter-initialization strategies. However, as we show here, changing the initialization distribution while avoiding barren plateaus does not, in general, amount to initializing the circuit near the identity.

Importantly, our analytical framework captures ensemble-averaged behavior via moment operators. Thus, it only provides a coarse-grained description and does not reveal how close individual circuits are to the identity. To probe that question, we complement the moment-operator analysis with a finer-grained diagnostic at the level of the sampled unitaries themselves. For a given $\bm{\theta}$, we compute
\begin{equation}
    d_{{\rm id}} (\bm{\theta})= \frac{\|U(\bm{\theta})-\id\|_1}{2^n}\,,
\end{equation}
where $\|\cdot\|_1$ denotes the trace norm. If modifying the initialization distribution merely reproduced identity initialization, then the sampled circuits should cluster near $d_{{\rm id}}=0$.

Our experimental setup is as follows. We consider a PQC composed of single-qubit $R_x$ and $R_y$ rotations followed by nearest-neighbor $CZ$ entangling gates, repeated for $L=4$ layers (see Fig.~\ref{fig:qckt_ela}). We vary the number of qubits from $n=2$ to $n=9$, and for each system size and each initialization family we draw $500$ i.i.d., initial sets of parameters. We study four representative families of parameter distributions: a uniform distribution $\mathrm{Unif}[-r,r]$ with $r=\pi$ (which can be expected to concentrate), a Beta distribution $\mathrm{Beta}(\alpha,\beta)$ as in~\cite{kulshrestha2022beinit}, and two Gaussian initializations $\NC(0,\sigma)$, where $\sigma=1$ and $\sigma=1/\sqrt{n}$ (i.e., the former can be expected to concentrate, while the latter is barren plateau free as per~\cite{zhang2022gaussian}).

\begin{figure}
    \centering
    \includegraphics[width=.8\linewidth]{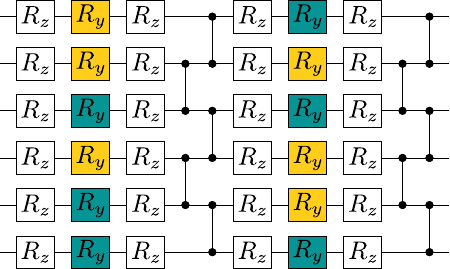}
    \caption{\textbf{Circuit ansatz for Max-Cut example.} We consider a PQC composed of single-qubit $R_z$,  $R_y$, $R_z$ rotations followed by nearest-neighbor $CZ$ entangling gates. The example shown corresponds to $n=6$ qubits and $L=2$ layers. Given that each $R_y$ can be initialized using two distributions, there are $2^{nL}$ choices, which we here schematically show by coloring gates according to which initialization was used.}
    \label{fig:qckt_ela-2}
\end{figure}

Figure~\ref{fig:tracenorm_expts_big} reports the distribution of $d_{{\rm id}} (\bm{\theta})$ for these initialization strategies. Here we can see that across all four families, the sampled circuits are generically not close to the identity. Thus, even when an initialization distribution is chosen so as to mitigate concentration, its effect is not simply to place the circuit near $\id$. What changes instead is how the initialization ensemble populates circuit space. For the barren plateau initialization strategies, i.e., $\mathrm{Unif}[-\pi,\pi]$ and $\NC(0,1)$, the distributions of $d_{\rm id}(\bm{\theta})$ become sharply concentrated around the same fixed non-zero value, as the system size grows, and even for fixed number of layers. In a way, this showcases the fact that random unitaries, leading to concentration phenomena, become typical and have properties that closely resemble those of Haar random ensembles.  By contrast, the Beta family and Gaussian initializations with $\sigma=1/\sqrt{n}$ retain a substantially broader spread. This broader support indicates that these strategies initialize the optimizer in genuinely different neighborhoods of the landscape.

 The next numerical experiments investigate whether these differences at initialization translate into different  training performance.

\subsection{Initialization affects the distribution of attained minima}

\begin{figure}[t]
    \centering
    \includegraphics[width=\linewidth]{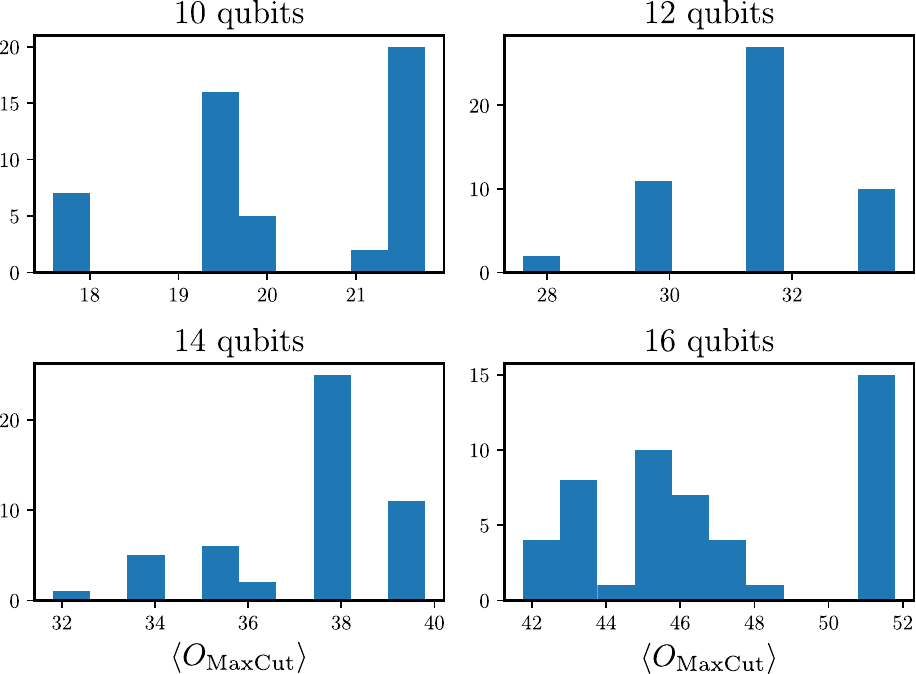}
    \caption{\textbf{Distribution of local minima after training for Max-Cut problem.} We show a histogram of the final expectation value for different system sizes.}
    \label{fig:minima_distribution}
\end{figure}

Here we study whether the initialization-induced biases discussed above persist all the way to the end of training. To that end, we study the distribution of attained minima for PQCs with $L$ layers, where each layer consists of three single-qubit parametrized rotations on every qubit ($R_z$, $R_y$, $R_z$), together with fixed nearest-neighbor $CZ$ gates (see Fig.~\ref{fig:qckt_ela-2}). We take $L=n$, initialize the $R_z$ angles uniformly in $[0,2\pi)$, and  the $R_y$ angles from a Gaussian distribution $\mathcal{N}(\theta^*,\varsigma^2)$ with $\varsigma^2=1/n$ and $\theta^*\in\{0,\pi\}$ chosen uniformly at random for each $R_y$ gate. One can readily see that all these initializations avoid barren plateaus as per Theorem~\ref{theo:bp-free}. Thus, each run samples from an exponentially large family of initialization strategies, obtained by selecting different centers for the $R_y$ rotations.

We consider qubit numbers $n=10,12,14,16$. For each $n$, we initialize the PQC to the all-zero state, and we draw a random Erd\"{o}s--R\'enyi graph with edge probability $p=0.5$. Then we follow the standard variational quantum eigensolver pipeline to minimize the expectation value of the corresponding Max-Cut Hamiltonian
\begin{equation}
    O_{\rm Max-Cut}=\sum_{(w,w')\in E}\sigma_z^{(w)}\sigma_z^{(w')}\, ,
\end{equation}
with $E$ the edge set of the Erd\"{o}s--R\'enyi graph. We perform training using vanilla gradient descent such that the optimization itself is deterministic and all variation between solutions on different runs comes from the initialization. The optimization is run for $500$ iterations, or until the gradient norm falls below the threshold set by the average gradient obtained from fully uniform initializations in $[0,2\pi)$.

The resulting histograms are shown in Fig.~\ref{fig:minima_distribution}. Several features are immediately apparent. First, for each system size the final objective values are broadly distributed and, in several cases, visibly multimodal. Hence, even within this single family of barren-plateau-free initializations, training does not repeatedly converge to one effectively unique minimum. Rather, the optimizer terminates in a collection of distinct local minima of different quality. Second, this spread is observed despite keeping fixed the ansatz, the optimizer, the stopping criterion, and the optimization budget. The only source of variability from run to run is therefore the initialization itself. In turn, this provides direct numerical evidence that not all the exponentially many initialization choices sampled here are practically interchangeable: they bias the optimizer toward different basins of attraction. As such, these results indicate that avoiding the fully concentrated barren-plateau regime does not collapse the optimization landscape into a single well-behaved region. Instead, it can reveal a landscape containing many trainable pockets, not all of which are equally favorable for the task at hand.

\section{Discussion}\label{sec:discussion}

One of the recurring reactions to barren plateaus is that the problem is fairly self-inflicted. That is, if one takes a sufficiently expressive~\cite{sim2019expressibility,holmes2021connecting} ansatz and blindly initializes its parameters at random, then of course we should expect concentration~\cite{larocca2024review}. And, to be fair, that criticism is partly right. Indeed, there are now many smart initialization strategies that can avoid the fully concentrated barren-plateau regime. The point of this work is not to argue otherwise. Rather, the point is that the statement ``good initializations solve barren plateaus'' is much more nuanced than it may first appear.

As we show here, avoiding barren plateaus is not a unique design principle. The first-moment framework developed here gives a simple way of identifying initialization strategies that remain separated from the barren-plateau ensemble. In particular, Theorem~\ref{lem:moment-1-rotation-Pauli} shows that, on average, a gate does two things at once: it damps some operator components and rotates others. Theorem~\ref{th:sufficient_condition} then shows that the damping part alone already gives a general lower bound on $\xi_{{\rm BP},\gamma}$. This is enough to recover several previously known schemes, but it also shows that these are only a small subset of a much broader design space. Shifted, biased, and more general non-symmetric distributions can all avoid the same concentration mechanism, and one can even have exponentially large strategies to explore (Theorems~\ref{theo:bp-free} and~\ref{th:MaxCut-slices}). So the lesson is not that there is a single clever way around barren plateaus. The lesson is that there are many.

But once that is true, a second issue appears immediately: these initializations need not be equivalent. Two distributions can have similar gradient scaling and still induce different first-moment rotations, different expected observables, and therefore different biases on the optimization landscape. This is precisely what the quantity $\xi_{\gamma',\gamma}$ studied here is meant to capture. Thus, avoiding barren plateaus does not remove the initialization problem, but changes its form. One may gain access to gradients that can be estimated with polynomial resources, but one still has to decide which of the many admissible initializations is the right one for the task at hand. In particular, since different initializations do not appear to lead to the same parts of the landscape, nor to minima of the same quality, the choice of initialization can play an extremely crucial role.

Moreover, we find it important to highlight that our results suggest a more careful interpretation of barren-plateau landscapes themselves. Barren plateaus are average-case statements about an initialization ensemble, yet they are often informally pictured as meaning that the landscape is flat everywhere. Our results point to a subtler picture. A landscape can be globally concentration-prone under naive random initialization and still contain many well-behaved regions that become accessible under better initialization strategies. In fact, our construction shows that one can generate exponentially many potentially inequivalent initialization choices, and the numerics suggest that these choices do not simply collapse onto one effectively unique basin. From this perspective, it is more useful to think of these landscapes as containing many trainable pockets, not all of equal quality, rather than as being uniformly featureless.

We do not want to oversell what first moments can do. The quantity $\xi_{{\rm BP},\gamma}$ gives a useful and very general certificate that an initialization remains separated from the fully concentrated barren-plateau ensemble, but it is not by itself a full trainability guarantee. The overlap with the input state, the observable, the ansatz, and the detailed geometry of the loss landscape still matter. So the message is not that our framework identifies the best initialization once and for all. Rather, it organizes the space of candidate barren-plateau-free strategies, and demonstrates that selecting the right initialization is an inherent part of the optimization problem.

Finally, a natural question to ask might be how warm starting strategies fit into our results here. In particular, there is a growing body of literature looking into classical surrogate methods (e.g. using tensor networks~\cite{lin2021real, ran2020encoding, rudolph2022decomposition, gibbs2024deep, chai2025resource, gibbs2025learning, jaderberg2025variational, jamet2023anderson, anselme2024combining, rudolph2022synergy, le2025riemannian, iaconis2024quantum, kanno2025tensor, ballarin2025efficient, gibbs2026low, szoldra2026scalable} or propagation methods~\cite{rudolph2025pauli, angrisani2024classically, lerch2024efficient, chakraborty2026scalable, shrikhande2025rapid}) to find approximate circuits for preparing ground states and initializing variational algorithms. It might initially seem that these methods offer a way around the paradox of choice identified here by picking out particular initialization points. However, this argument should be viewed with caution because we suspect that our strategies for constructing exponentially large/continuous families could to be equally applied also to many of such warm started methods. That said, it remains to be seen to what extent such families of initializations would lead to substantially inequivalent solutions. Or, indeed if they do, how hard it is to find a `good enough' initialization within those families. After all, heuristic optimization often works out better than one theoretically expects. 

Overall, our view is that parameter initialization should not be treated as a minor implementation detail in variational quantum algorithms. Once many barren-plateau-free initializations are on the table, the hard question is no longer only how to avoid exponentially vanishing gradients. It shifts from avoiding landscape concentration to choosing a high-quality initialization out of an exponentially large family. Put together, our results appear to imply that there is no free lunch for escaping the curse of dimensionality. 

\section{Acknowledgments}

A.K. and I.S. were supported in part by the DARPA ONISQ program. 
R.P. acknowledges the support of the SNF Quantum Flagship Replacement Scheme (Grant No. 215933). 
D.G.-M. acknowledges financial support from the European Research Council (ERC) via the Starting grant q-shadows (101117138) and from the Austrian Science Fund (FWF) via the SFB BeyondC (10.55776/FG7). L.C. was supported by the U.S. Department of Energy, Office of Science, Office of Advanced Scientific Computing Research under Contract No. DE-AC05-
00OR22725 through the Accelerated Research in Quantum Computing Program MACH-Q. Z.H. acknowledges support from the Sandoz Family Foundation-Monique de Meuron program for Academic Promotion.
M.C. was supported by the Laboratory Directed Research and Development (LDRD) program of Los Alamos National Laboratory (LANL) under project number 20260043DR and by LANL's ASC Beyond Moore’s Law project.

\bibliography{quantum.bib,biblio}

\onecolumngrid

\newpage

\appendix

\section{Vectorization formalism}\label{app:vec_formalism}

Given a linear operator $A\in\HC\otimes\HC^*$ expressed as $A = \sum_{i,j}^d A_{i,j}\ketbra{i}{j}$, the vectorization map $\mathrm{Vec}:\HC\otimes\HC^*\rightarrow \HC^{\otimes 2}$ is defined as follows
\begin{align}
\kett{A}=\mathrm{Vec}(A) = \sum_{i,j} A_{i,j}\ket{i}\ket{j}\,.
\end{align}
In particular, we recall that, given three linear operators $A$, $B$ and $C$, then the following property holds 
\begin{align}\label{eq:kroneker_vectorisation}
    \kett{ABC}  = A\otimes C^T\kett{B}
\end{align}
where $C^T$ denotes the transpose of $C$. Using this property we can see that 
\begin{align}
    \langle\!\langle A|B\rangle\!\rangle =\Tr[A^\dagger B]\, .
\end{align}

Finally, given a channel $\Phi:\HC\otimes\HC^*\rightarrow\HC\otimes\HC^*$ expressed as $\Phi(\cdot) = \sum_{i,j} A_i (\cdot) B_j^\dagger$, the vectorization of matrices induces the vectorization of channels as
\begin{align}
    \hat{\Phi}={\rm Vec}\left[\Phi(\cdot)\right] = \sum_{i,j} A_i\otimes B^*_j \, .
\end{align}
where $B^*$ is the complex conjugate of $B$.

\section{Necessary condition}\label{app:necessary_condition}
Our goal here is to show that a necessary condition for an initialization strategy to not lead to a barren plateau is the following
\begin{equation}\label{eq:expresscondition_appendix}
   {\|\widehat{\tau}_{{\rm BP}}\kett{O}- \widehat{\tau}_{\vec{\gamma}}\kett{O}   \|_1} \in \Omega\left( \frac{1}{\poly(n)}\right) \, .
\end{equation}

\begin{proof}
We start the proof by recalling that an initialization strategy $P_{\vec{\gamma}}(\thv)$ does not suffer from barren plateaus if (see Eq.~\eqref{eq:necessary})
\begin{equation}\label{eq:necessary_appendix}
   \| \braa{\rho}\widehat{\tau}_{\vec{\gamma}}\kett{O} - \braa{\rho}\widehat{\tau}_{{\rm BP}}\kett{O} \|\in \Omega\left( \frac{1}{\poly(n)}\right)\,.
\end{equation}

Without loss of generality, we can assume that we can decompose the vectorized observable $\kett{O}$ in the Pauli basis $\{\sigma_i\}_{i=1}^{4^n}$ as $\kett{O} = \sum_i \alpha_i \kett{\sigma_i}$. Similarly, we can decompose in the same basis (but with different coefficients) the operator $\kett{O}$ after applying the vectorized 1-twirl moment
\begin{equation}
    \kett{O} = \sum_i \alpha_i \kett{\sigma_i}, \,\quad \text{and}\quad  \widehat{\tau}_{\vec{\gamma}}\kett{O} = \sum_i \tilde{\alpha}_i \kett{\sigma_i}\,.
\end{equation}
We can use this to upper-bound the left hand side of Eq.~\eqref{eq:necessary_appendix} as follows
\begin{align}\label{eq:first_upperbound}
   \left| \braa{\rho}\widehat{\tau}_{\vec{\gamma}}\kett{O} - \braa{\rho}\widehat{\tau}_{{\rm BP}}\kett{O} \right| =&\left| \braa{\rho}\sum_i \tilde{\alpha}_i\kett{\sigma_i} - \braa{\rho}\sum_i \tilde{\alpha}^{{\rm BP}}_i\kett{\sigma_i} \right| \\
   =&\left| \braa{\rho}\sum_i \left(\tilde{\alpha}_i -  \tilde{\alpha}^{{\rm BP}}_i\right)\kett{\sigma_i} \right|\\
   \leq & \sum_i \left|(\tilde{\alpha}_i - \tilde{\alpha}_i^{{\rm BP}})\langle\!\langle \rho \kett{\sigma_i}\right|\\
   \leq & \sum_i \left|(\tilde{\alpha}_i - \tilde{\alpha}_i^{{\rm BP}})\right| = \|\widehat{\tau}_{\vec{\gamma}}\kett{O} - \widehat{\tau}_{{\rm BP}}\kett{O}\|_1\label{eq:last_upperbound}\,.
\end{align}
In the second equality, we have grouped the coefficients that correspond to the same Pauli. Then, in the first inequality we used the triangle inequality to show that $|\sum_i x_i|\leq \sum_i |x_i|$, and finally in the second inequality we used that $\max_{\rho}\langle\!\langle \rho \kett{\sigma_i}= \max_{\rho} \Tr[\rho \sigma_i] = 1 $, as  $\sigma_i$ are  Pauli operators. In the last equality we just use that the one norm is the sum of the absolute value of the coefficients, which completes the proof.

\end{proof}

\section{Proof of Theorem~\ref{lem:moment-1-rotation-Pauli}}\label{sec:twirt_rot_op}
In this section we prove the following theorem.
\setcounter{theorem}{0}
\begin{theorem}\label{lem:moment-1-rotation-Pauli-ap}
    Let $U = e^{-i\theta\sigma_l/2}$, where $\theta$ is sampled according to $P_{\gamma_l}(\theta)$. Then, its associated first moment operator  can be written as a contraction matrix $\eta_{l}$ times a rotation $W_{l}\otimes W_l^*$.  That is
\begin{equation}\label{eq:twirl_as_rot_app-ap}
    \widehat{\tau}_{\gamma_l}  = \left(W_{l}\otimes W_{l}^*\right)\eta_{l}\,,
\end{equation}
where $W_{l}=e^{-i\phi_{l} \sigma_l/2}$, $\phi_{l} = \arccos \left(\frac{a_l}{\sqrt{a_l^2 + b_l^2}}\right)$ and 
\begin{align}
    a_l = \int P_{\gamma_l}(\th) \cos(\th) d\th \,,\quad 
    b_l =\int P_{\gamma_l}(\th)\sin(\th) d\th \,.\nonumber
\end{align}
In particular, the action of the matrix $\eta_{l}$ over Pauli operators is
\begin{equation}
    \eta_{l}\kett{\sigma_i} = 
    \begin{cases}
        \kappa_l \kett{\sigma_i} \, {\rm iff} \, \{\sigma_l,\sigma_i\} = 0\,,\\
        \kett{\sigma_i} \, {\rm iff} \, [\sigma_l,\sigma_i] = 0\,,
    \end{cases}
\end{equation}
where
\begin{equation}
    \kappa_l=\sqrt{a_l^2 + b_l^2}
\end{equation}
which shows that $\eta_{l}$ has positive eigenvalues smaller or equal to $1$, 
\end{theorem}
\begin{proof}
We start by considering the action $\widehat{\tau}_{\gamma_l}\kett{O}$, for some generic operator $O$. In particular, we find it convenient to  separate the Paulis in the expansion of $\kett{O}$ into two groups: those which commute with $\sigma_l$, and those that anti-commute with the same generator. That is
\begin{align}
    \kett{O} = \sum_{i \, {\rm s.t.} \,  [\sigma_l,\sigma_i] = 0} \alpha_i\kett{\sigma_i} + \sum_{i \, {\rm s.t.} \, \{\sigma_l,\sigma_i\} = 0} \alpha_i\kett{\sigma_i}\,.
\end{align}

In particular, the moment operator for the parameterized gates is
\begin{align}
     \widehat{\tau}_{\gamma_l} = \int d\th P_{\gamma_l}(\th) e^{-i\th\sigma_l/2}\otimes e^{i\th\sigma^*_l/2} \,,
\end{align}
which can be simplified by using the fact that $\sigma_l^2 = \1$, and hence that 
\begin{align}
    e^{-i\theta\sigma_l/2} = \cos(\th/2)\1 -i\sin(\th/2)\sigma_l \, . 
\end{align}
Combining the previous results, we obtain 
\begin{align}
     \widehat{\tau}_{\gamma_l} = &I_{2,0}\1\otimes\1 + I_{0,2}\sigma_l\otimes\sigma^*_l -i I_{1,1}(\sigma_l\otimes\1 - \1\otimes\sigma^*_l) \, ,\label{eq:vectorized_twirl_without_int}
\end{align}
where we have used the notation
\begin{align}\label{eq:I_def}
    I_{k_1,k_2} =  \int d\th P_{\gamma_l}(\th) \cos^{k_1}(\th/2)\sin^{k_2}(\th/2)\,,
\end{align}
and the the normalization of the probability distribution, i.e.,  $\int d\th P_{\gamma_l}(\th)=1$. Eq.~\eqref{eq:vectorized_twirl_without_int} constitutes the basis building block for the proofs presented in the next section. 

Then, the action of $\widehat{\tau}_{\gamma_l}$ over the observable is
\begin{align}
    \widehat{\tau}_{\gamma_l}\kett{O} =& \sum_{i \, {\rm s.t.} \,  [\sigma_l,\sigma_i] = 0} \alpha_i\kett{\sigma_i} + (I_{2,0}\1\otimes\1 + I_{0,2}\sigma_l\otimes \sigma^*_l) \sum_{i \, {\rm s.t.} \, \{\sigma_l,\sigma_i\} = 0} \alpha_i \kett{\sigma_i}\\
    &-iI_{1,1}(\sigma_l\otimes\1 - \1\otimes\sigma^*_l) \sum_{i \, {\rm s.t.} \, \{\sigma_l,\sigma_i\} = 0} \alpha_i \kett{\sigma_i}\\
    =&\sum_{i \, {\rm s.t.} \,  [\sigma_l,\sigma_i] = 0} \alpha_i\kett{\sigma_i} + (I_{2,0} - I_{0,2})\sum_{i \, {\rm s.t.} \, \{\sigma_l,\sigma_i\} = 0} \alpha_i \kett{\sigma_i} -2I_{1,1}\sum_{i \, {\rm s.t.} \, \{\sigma_l,\sigma_i\} = 0} \alpha_i \kett{i\sigma_l \sigma_i}\label{eq:pre_final_apply_twirl}\,.
\end{align}
In the first equality we used that the gate $e^{-i\theta\sigma_l/2}$ (and thus the associated first moment operator) will only affect the operators $\sigma_i$ that do not commute with $\sigma_l$.  In the second equality we use Eq.~\eqref{eq:kroneker_vectorisation} to pass the Pauli operator $\sigma_l$ inside the vectorized $\sigma_i$. Furthermore, we use that $\{\sigma_l, \sigma_i\} = 0$ to further simplify the terms $\kett{\sigma_l \sigma_i\sigma_l}  =-\kett{\sigma_i}$ and $\kett{\sigma_i \sigma_l}= -\kett{\sigma_l \sigma_i}$.

Next we can define the following terms $a,b$ and simplify them by using Eq.~\eqref{eq:I_def}
\begin{align}
    a_l =& I_{2,0} - I_{0,2} = \int P_{\gamma_l}(\th) \left(\cos(\th/2) - \sin(\th/2)\right)d\th = \int P_{\gamma_l}(\th) \cos(\th) d\th \\
    b_l =& 2I_{1,1} = \int P_{\gamma_l}(\th) 2\sin(\th/2)\cos(\th/2) d\th =\int P_{\gamma_l}(\th)\sin(\th) d\th \,.
\end{align}
With this definitions, we can now further simplify Eq.~\eqref{eq:pre_final_apply_twirl}. Indeed we can introduce $a,b$ in the equation and rearrange the terms as follows
\begin{align}\label{eq:initial_apply_twirl}
    \widehat{\tau}_{\gamma_l}\kett{O} 
    =& \sum_{i \, {\rm s.t.} \,  [\sigma_l,\sigma_i] = 0} \alpha_i\kett{\sigma_i} + \sum_{i \, {\rm s.t.} \, \{\sigma_l,\sigma_i\} = 0} \alpha_i (a_l\kett{\sigma_i} - b_l\kett{i\sigma_l \sigma_i})\\
    =& \sum_{i \, {\rm s.t.} \,  [\sigma_l,\sigma_i] = 0} \alpha_i\kett{\sigma_i} + \sqrt{a_l^2 + b_l^2}\sum_{i \, {\rm s.t.} \, \{\sigma_l,\sigma_i\} = 0} \alpha_i \left(\frac{a_l}{\sqrt{a_l^2 + b_l^2}}\kett{\sigma_i} - \frac{b_l}{\sqrt{a_l^2 + b_l^2}}\kett{i\sigma_l \sigma_i}\right)
    \label{eq:final_apply_twirl}
\end{align}
where in the second equality we multiply and divide the second summand by $\sqrt{a_l^2 + b_l^2}$. 

Lets now focus on terms of the second summation in Eq.~\eqref{eq:final_apply_twirl}. Without loss of generality we can write 
\begin{equation}
    \left(\frac{a_l}{\sqrt{a_l^2 + b_l^2}}\kett{\sigma_i} - \frac{b_l}{\sqrt{a^2 + b^2}}\kett{i\sigma_l \sigma_i}\right) = \kett{W_{l} \sigma_i W_{l}^\dagger}\,,
\end{equation}
where $W_{l}=e^{-i\phi_{l} \sigma_l/2}$ and $\phi_{l} = \arccos \left(\frac{a_l}{\sqrt{a_l^2 + b_l^2}}\right)$. Therefore we can simplify the previous expression to find
\begin{equation}
    \widehat{\tau}_{\gamma_l}\kett{O} = W_{l}\otimes W_{l}^* \left( \sum_{i \, {\rm s.t.} \,  [\sigma_l,\sigma_i] = 0} \alpha_i\kett{\sigma_i} + \sqrt{a_l^2 + b_l^2}\sum_{i \, {\rm s.t.} \, \{\sigma_l,\sigma_i\} = 0} \alpha_i \kett{\sigma_i}  \right)\,.
\end{equation}

Then, let us define the matrix $\eta_{l}$, whose action is 
\begin{equation}
    \eta_{l}\kett{\sigma_i} = 
    \begin{cases}
        \sqrt{a_l^2 + b_l^2} \kett{\sigma_i} \, {\rm iff} \, \{\sigma_l,\sigma_i\} = 0\,,\\
        \kett{\sigma_i} \, {\rm iff} \, [\sigma_l,\sigma_i] = 0\,.
    \end{cases}
\end{equation}
Since the previous holds for any traceless operator $O$, we find that the first moment operator for a single rotation $U = e^{-i\theta\sigma_l/2}$, where $\theta$ is sampled according to $P_{\gamma_l}(\theta)$, can be expressed as 
$\widehat{\tau}_{\gamma_l} =\left(W_{l}\otimes W_{l}^*\right) \eta_{l}$.
    
\end{proof}

\section{Proof of Theorem~\ref{th:sufficient_condition}}\label{app:proof_th1}
In this section we provide a proof for Theorem~\ref{th:sufficient_condition}, which we recall here for convenience.

\begin{theorem}\label{th:sufficient_condition_appendix}
Let $U(\thv)$ be a PQC as defined in Eq.~\eqref{eq:pauliPQC} with associated distribution of parameters $P_{\vec{\gamma}}(\thv)$. The expected Heisenberg-evolved operator gap $\xi_{{\rm BP},\vec{\gamma}}$ is lower bounded as
    \begin{equation}
        \xi_{{\rm BP},\vec{\gamma}}\in\Omega\left(\kappa\right)
    \end{equation}
with
\begin{equation}
\kappa=\prod_{l=1}^L \kappa_l \,,\quad\kappa_l=\sqrt{a_l^2 + b_l^2}\,,
\end{equation}
and where 
\begin{align}\label{eq:def_a_app_tm1}
        a_l &= \int \cos(\th_l)P_{\gamma_l}(\th_l)d\th_l\,,\\
        b_l &= \int \sin(\theta_l)P_{\gamma_l}(\theta_l)d\th_l\,.\label{eq:def_b_app_tm1}
    \end{align}
\end{theorem} 

\begin{proof}
 In what follows, we will consider a vectorized observable $\kett{O}$ whose decomposition in the  Pauli Basis is $\kett{O}=\sum_i \alpha_i \kett{\sigma_i}$. Our goal is to bound the expected Heisenberg-evolved operator gap $\xi_{{\rm BP},\vec{\gamma}}$, defined as
 \begin{equation}
  \xi_{{\rm BP},\gamma} =  {\|\widehat{\tau}_{{\rm BP}}\kett{O} - \widehat{\tau}_{\vec{\gamma}}\kett{O}\|_1 }\,.
\end{equation}
One can readily see that the following chain of inequalities holds
\begin{align}
    \|\widehat{\tau}_{{\rm BP}}\kett{O} - \widehat{\tau}_{\vec{\gamma}}\kett{O}\|_1&\geq \|\widehat{\tau}_{{\rm BP}}\kett{O} - \widehat{\tau}_{\vec{\gamma}}\kett{O}\|_2\nonumber\\
    &\geq \max_{\tilde{\rho}} |\braa{\tilde{\rho}}\widehat{\tau}_{{\rm BP}}\kett{O} - \braa{\tilde{\rho}}\widehat{\tau}_{\vec{\gamma}}\kett{O}|\nonumber\,,
\end{align}
where the maximum is taken over all quantum states $\tilde{\rho}$. 

Next, using the fact that $\widehat{\tau}_{{\rm BP}}$ forms a $1$-design over $\mathbb{U}(\HC)$, we have~\cite{mele2023introduction}
\begin{equation}
    \widehat{\tau}_{{\rm BP}}=\frac{|\id_{2^n}\rangle\!\rangle\langle\!\langle \id_{2^n}|}{2^n}\,,
\end{equation}
where $\id_{2^n}$ denotes the $2^n\times 2^n$ identity over $\HC$, we have
\begin{equation}
    \braa{\tilde{\rho}}\widehat{\tau}_{{\rm BP}}\kett{O}=0\,.
\end{equation}
The previous equality holds since, by definition, $O$ is traceless and $\braa{\id_{2^n}}O\rangle\!\rangle=\Tr(O)=0$.

Then, consider the following lower-bound
\begin{align}\label{eq:lower_bound_twirl}
    \max_\rho |\braa{\rho}\widehat{\tau}_{\vec{\gamma}}\kett{O}| \geq \sqrt{\frac{1}{2^n}\braa{O} \widehat{\tau}_{\vec{\gamma}}^\dagger \widehat{\tau}_{\vec{\gamma}} \kett{O}}\,,
\end{align}
where we used that $ \max_\rho |\braa{\rho}\widehat{\tau}_{\vec{\gamma}}\kett{O}| = \|{\tau}_{\vec{\gamma}}(O)\|_{\infty}$  and $\braa{O} \widehat{\tau}_{\vec{\gamma}}^\dagger \widehat{\tau}_{\vec{\gamma}} \kett{O} = \Tr[{\tau}_{\vec{\gamma}}(O){\tau}_{\vec{\gamma}}(O)^\dagger]=\sum_i \lambda_i^2$ is the sum of all the eigenvalues of ${\tau}_{\vec{\gamma}}(O)$ squared. Therefore, we can readily see that 
\begin{align}
    \sqrt{\frac{1}{2^n}\sum_{i=1}^{2^n} \lambda_i^2}\leq \sqrt{\frac{1}{2^n}\max_{j}\lambda_j^2\sum_{i=1}^{2^n} } = \max_{j}|\lambda_j| =\|{\tau}_{\vec{\gamma}}(O)\|_{\infty}
\end{align}

Next, let us note that for a circuit of the form $U(\bm{\theta}) = \prod_{l=1}^L e^{-i\theta_l\sigma_l/2}V_l$, the associated moment operator can be expressed as 
\begin{equation}
    \widehat{\tau}_{\vec{\gamma}}=\prod_{l=1}^L \widehat{\tau}_{\gamma_l}\widehat{\tau}_{V_l}\,,
\end{equation}
where we recall that the associated moment operator for an unparametrized gate $V_l$ is simply 
\begin{align}
    \widehat{\tau}_{V_l}= V_l\otimes V_l^* \, ,
\end{align}
Then, leveraging Lemma~\ref{lem:moment-1-rotation-Pauli} we find 
\begin{equation}
    \widehat{\tau}_{\vec{\gamma}}=\prod_{l=1}^L \left(W_{l}\otimes W_{l}^*\right)\eta_{l}\left(V_l\otimes V_l^*\right)\,.
\end{equation}
This result, in turn, leads to
\begin{align}
   \braa{O} \widehat{\tau}_{\vec{\gamma}}^\dagger \widehat{\tau}_{\vec{\gamma}} \kett{O} = &  \braa{O}\prod_{k = L}^1 \left(V_k\ad\otimes V_k^t\right)\eta_{k}\left(W_{k}\ad\otimes W_{k}^t\right) \prod_{l=1}^L \left(W_{l}\otimes W_{l}^*\right)\eta_{l}\left(V_l\otimes V_l^*\right)\kett{O}\\
   = & \braa{O}\prod_{k = L}^2 \left(V_k\ad\otimes V_k^t\right)\eta_{k}\left(W_{k}\ad\otimes W_{k}^t\right)\left(\left(V_1\ad\otimes V_1^t\right)\eta_{1}\eta_{1}\left(V_1\otimes V_1^*\right)\right) \prod_{l=2}^L \left(W_{l}\otimes W_{l}^*\right)\eta_{l}\left(V_l\otimes V_l^*\right)\kett{O}\\
   \geq & (a_1^2 + b_1^2)\braa{O}\prod_{k = L}^2 \left(V_k\ad\otimes V_k^t\right)\eta_{k}\left(W_{k}\ad\otimes W_{k}^t\right)\prod_{l=2}^L \left(W_{l}\otimes W_{l}^*\right)\eta_{l}\left(V_l\otimes V_l^*\right)\kett{O}\label{eq:lowerbound_eta}\\
   \geq & \prod_{l=1}^L (a_l^2+b_l^2) \braakett{O}
\end{align}
 In Eq.~\eqref{eq:lowerbound_eta} we used the inequality $\expval{\eta_1}{v}\geq \lambda_{\rm min}(\eta_1)\braket{v}$, and the fact that the eigenvalues of $\eta_1$ are $1$ and $\sqrt{a_1^2 + b_1^2}$. In the last inequality we applied the same identity just explained in Eq.~\eqref{eq:lowerbound_eta} recursively. Finally, replacing this result in Eq.~\eqref{eq:lower_bound_twirl} we obtain 
\begin{align}
    \max_\rho |\braa{\rho}\widehat{\tau}_{\vec{\gamma}}\kett{O}| \geq \sqrt{\frac{1}{2^n}\braa{O} \widehat{\tau}_{\vec{\gamma}}^\dagger \widehat{\tau}_{\vec{\gamma}} \kett{O}}\geq \prod_{l=1}^L\sqrt{a_l^2 + b_l^2} \sqrt{\frac{\braakett{O}}{2^n}}
\end{align}
Thus, using the fact that, by assumption, $\frac{\braakett{O}}{2^n}\in\OC(1)$, we have shown that the expected Heisenberg-evolved operator
gap is lower-bounded by the product of the coefficients $\xi\in \Omega\left(\prod_{l=1}^L\sqrt{a_l^2 +b_l^2}\right)$. 

\end{proof}

\section{Proof of Theorem~\ref{theo:bp-free}}\label{app-prood-theo3}

Here we present a proof for Theorem~\ref{theo:bp-free}, which we recall and restate in a self-contained manner.
\begin{theorem}\label{theo:bp-free-ap}
    Let $U(\thv)$ be a PQC composed of general single-qubit gates interleaved with diagonal entangling gates, i.e., 
       \begin{equation}\nonumber
        U(\thv) =\prod_{l=1}^L \prod_{j=1}^{n}e^{-i\th_{1,j,l}\sx^{(j)}/2} e^{-i\th_{2,j,l}\sz^{(j)}/2}e^{-i\th_{3,j,l}\sx^{(j)}/2} V_l\,,
    \end{equation}
     let $O$ be sum of local Pauli operators acting on $\OC(1)$ qubits, and $\rho$ the all-zero state, $\ket{0}^{\otimes n}$.  Then, for all exponentially many initialization choices such that the single-qubit gate's parameters are independently sampled with respect to  $\NC(\theta^*,\varsigma^2)$, where  $\theta^*\in\{0,\pi\}$ and $ \varsigma\in\Theta\left(\frac{1}{\sqrt{L}}\right)$, the model is barren plateau-free, as
\begin{equation}\label{eq:exp_concentration}
    \Var_{\bm{\gamma}}[\LC(\thv)]\in\Omega\left(\frac{1}{\poly(n)}\right)\,.
\end{equation}
\end{theorem}

\begin{proof}
    The proof of this theorem follows closely the proof of Theorem 4.1 in~\cite{zhang2022escaping}. In particular, all we need to do is generalize Lemma B.1 and Lemma B.2 therein. Furthermore, we note that while the work in~\cite{zhang2022escaping} shows absence of barren plateaus by computing $\Var_{\bm{\gamma}}[\partial \LC(\thv)/\partial \theta_l]$ for all $\theta_l\in\thv$, this result directly implies $\Var_{\bm{\gamma}}[\LC(\thv)]~\cite{holmes2021connecting}$.

\begin{lemma}[Generalized version of Lemma B.1 from~\cite{zhang2022escaping}]
Let $\zeta\in\{0,1\}$, and let
\begin{equation}
    \theta=\pi\zeta+\delta\,, \qquad\text{where}\quad 
    \delta\sim\mathcal N(0,\varsigma^2)\,,
\end{equation}
and let
\begin{equation}
    U(\theta)=e^{-i\theta G/2}\,,
\end{equation}
where $G$ is a Hermitian unitary. Let $O$ be a Hermitian observable such that
\begin{equation}
    \{O,G\}=0\,.
\end{equation}
and $\rho$ a quantum state.  Then
\begin{equation}
    \mathbb E_{\theta} \Tr\!\left[
        OU(\theta)\rho U(\theta)^\dagger
    \right]^2
    \geq (1-\varsigma^2)\Tr[O\rho]^2 + \varsigma^2(1-\varsigma^2) \Tr[iGO\rho]^2\,.
\end{equation}
\end{lemma}

\begin{proof}
First, let us define the quantities
\begin{equation}
    A=\Tr[O\rho]\,, \qquad B=\Tr[iGO\rho]\,.
\end{equation}
Since $G^2=\id$,
\begin{equation}
    e^{-i\theta G/2}=\cos(\theta/2) \id -i\sin(\theta/2)G\,.
\end{equation}
Using $\{O,G\}=0$, we obtain
\begin{equation}
    \Tr\left[ OU(\theta)\rho U(\theta)^\dagger
    \right]=\cos(\theta)A+\sin(\theta)B\,.
\end{equation}
Substituting $\theta=\pi\zeta+\delta$, we find
\begin{equation}
    \cos\left(\delta+\pi \zeta\right)=(-1)^\zeta\cos(\delta)\,,
    \qquad
    \sin\left(\delta+\pi \zeta\right)=(-1)^\zeta\sin(\delta)\,,
\end{equation}
and hence we have
\begin{equation}
    \cos(\theta)A+\sin(\theta)B = (-1)^\zeta\left(\cos(\delta)A+\sin(\delta)B\right)\,.
\end{equation}
Therefore the square is independent of $\zeta$:
\begin{equation}
    \Tr\!\left[ OV(\theta)\rho V(\theta)^\dagger
    \right]^2  = \left(\cos(\delta)A+\sin(\delta)B\right)^2\,.
\end{equation}
Taking the Gaussian expectation over $\delta\sim\mathcal N(0,\varsigma^2)$ gives exactly the same
quantity as in Lemma B.1:
\begin{equation}
\begin{aligned}
    \mathbb E_{\delta} \left(\cos(\delta)A+\sin(\delta)B\right)^2
    &= \frac{1+e^{-2\varsigma^2}}{2}A^2 + \frac{1-e^{-2\varsigma^2}}{2}B^2 \\
    &\geq  (1-\varsigma^2)A^2  + \varsigma^2(1-\varsigma^2)B^2\,,
\end{aligned}
\end{equation}
where we used that $e^{-x}\geq 1-x$ and $e^{-x}\leq 1-x+\frac{x^2}{2}$ for all $x\geq 0$. This proves the claim.
\end{proof}

\begin{lemma}[Generalized version of Lemma B.2 from~\cite{zhang2022escaping}]
Let $\zeta\in\{0,1\}$, let
\begin{equation}
    \theta=\pi\zeta+\delta\,, \qquad \delta\sim\mathcal N(0,\gamma^2)\,,
\end{equation}
and let
\begin{equation}
    U(\theta)=e^{-i\theta G/2}\,,
\end{equation}
where $G$ is a Hermitian unitary. Let $O$ be a Hermitian observable such that
\begin{equation}
    \{O,G\}=0\,,
\end{equation}
and let $\rho$ be a density matrix. 
Then
\begin{equation}
    \mathbb E_{\theta}
    \left(\frac{\partial \LC(\theta)}{\partial \theta}\right)^2
    \geq  (1-\varsigma^2)
    \left(\frac{\partial \LC(\theta)}{\partial\theta}\bigg|_{\theta=\pi\zeta}
    \right)^2 + \varsigma^2(1-\varsigma^2)\Tr[O\rho]^2\,.
\end{equation}
\end{lemma}

\begin{proof}
As above, define
\begin{equation}
    A=\Tr[O\rho]\,, \qquad  B=\Tr[iGO\rho]\,.
\end{equation}
Then
\begin{equation}
    \LC(\theta)=\cos(\theta)A+\sin(\theta)B\,,
\end{equation}
and hence
\begin{equation}
    \frac{\partial \LC(\theta)}{\partial\theta} = -\sin(\theta)A+\cos(\theta)B\,.
\end{equation}
Substituting $\theta=\pi\zeta+\delta$, we obtain
\begin{equation}
\frac{\partial \LC\left(\pi\zeta+\delta\right)}{\partial\theta}
    = (-1)^\zeta\left(-\sin(\delta)A+\cos(\delta)B\right)\,.
\end{equation}
Therefore the squared derivative is again independent of $\zeta$. Taking the expectation gives
\begin{equation}
\begin{aligned}
    \mathbb E_{\theta} \left( \frac{\partial \LC(\theta)}{\partial\theta}
    \right)^2   &=  \mathbb E_{\delta}  \left( -\sin(\delta)A+\cos(\delta)B
    \right)^2  \\  &= \frac{(1-e^{-2\zeta^2})}{2}A^2  + \frac{(1+e^{-2\zeta^2})}{2}B^2 \\  &\geq
    \zeta^2(1-\zeta^2)A^2   + (1-\zeta^2)B^2\,.
\end{aligned}
\end{equation}
At the center $\theta=\pi\zeta$,
\begin{equation}
    \frac{\partial \LC(\theta)}{\partial\theta}\bigg|_{\theta=\pi\zeta}
    =
    (-1)^\zeta B,
\end{equation}
and hence
\begin{equation}
    \left(
        \frac{\partial \LC(\theta)}{\partial\theta}\bigg|_{\theta=\pi\zeta}
    \right)^2
    =
    B^2.
\end{equation}
Putting it all together leads to
\begin{equation}
    \mathbb E_{\theta} \left( \frac{\partial \LC(\theta)}{\partial \theta}
    \right)^2
    \geq  (1-\varsigma^2)  \left( \frac{\partial \LC(\theta)}{\partial\theta}\bigg|_{\theta=\pi\zeta}
    \right)^2 + \varsigma^2(1-\varsigma^2)\Tr[O\rho]^2\,,
\end{equation}
which proves the claim.
\end{proof}

From here, one simply follows the steps in~\cite{zhang2022escaping}.

\end{proof}

Finally, we note that while we here used Gaussian initialization, a similar result can be found with uniform initialization, where one simply needs to replace the expectation value over $\delta$.

\section{Simplified computations for symmetric strategies}\label{ap-sec:symmetric-strategis}
In this section we prove two lemmas that showcase different techniques for computing $\widehat{\tau}_{\vec{\gamma}_{\rm sym}}\kett{O}$ for the case when  $P_{\vec{\gamma}_{\rm sym}}(\thv)$ is composed of independent symmetric probability distributions on each parameter. 

\subsection{Lemma~\ref{lemma:symmetric_diagonal} and its proof}\label{app:lemma_sym_diag}
We devote this subsection to presenting and proving Lemma~\ref{lemma:symmetric_diagonal}.

\begin{lemma}\label{lemma:symmetric_diagonal}
     Let $U(\thv)$ be a PQC as in Eq.~\eqref{eq:pauliPQC} where all $V_l$  gates are Clifford operations, let $P_{\bm{\gamma}_{\rm sym}}(\thv)$ be a symmetric initialization strategy and $O$ a Pauli. The first moment operator takes the form      \begin{equation}\label{eq:result_lemma_diagonal}
         \widehat{\tau}_{\bm{\gamma}_{\rm sym}} \kett{O} =  U(\vec{0})\otimes U^*(\vec{0})\mu_{\bm{\gamma}_{\rm sym}}\kett{O}\, ,
     \end{equation}
     where $\mu_{\bm{\gamma}_{\rm sym}}$ is a diagonal matrix with positive eigenvalues smaller or equal to $1$.
\end{lemma}

\begin{proof}
    We start by noting that any circuit of the form of Eq.~\ref{eq:pauliPQC} can be written as a final Clifford gate after several parametrized gates made out of Pauli operators. Indeed, since for any Clifford gate $C$, and Pauli $\sigma$, one has $Ce^{-i \theta\sigma}C\ad=e^{-i \theta C\sigma C\ad}=e^{-i \theta\tilde{\sigma}}$ for a different Pauli $\tilde{\sigma}$ (recall that under conjugation, Clifford unitaries map Pauli operators to Pauli operators), we can ``push'' the action of the Cliffords to the end of the circuit. As such, we can express
\begin{align}\label{eq:rewrite_absorb_clifford}
        U(\thv) = U(\vec{0}) \prod_{l=1}^L e^{-i\th_l\tilde{\sigma}_l/2}\,,
    \end{align}
where we defined $U(\vec{0}) = \prod_{l=1}^L C_l$, the product of all Clifford gates in the circuit. 

With the previous in mind, let us next show that the result in the lemma is true for a single Pauli rotation.    Let $P_{\gamma_{l,{\rm sym}}}(\theta_l)$ be a symmetric initialization for the parameter $\theta_l$. From Eq.~\eqref{eq:vectorized_twirl_without_int}, the corresponding first moment operator is
    \begin{equation}
        \widehat{\tau}_{\gamma_{l,{\rm sym}}} = I_{2,0}\1\otimes\1 + I_{0,2}\sigma_l\otimes\sigma^*_l\,.
    \end{equation}
Then, its action on a vectorized Pauli $\kett{\sigma_i}$, where $\sigma_i$ is either $[\sigma_i,\sigma_l] = 0$, or $\{\sigma_i,\sigma_l\}=0$. Again, if the generator commutes with the observable, then  $\widehat{\tau}_{\gamma_{l,{\rm sym}}}\kett{\sigma_i}=\kett{\sigma_i}$, else
    \begin{align}
         \widehat{\tau}_{\gamma_{l,{\rm sym}}}\kett{\sigma_i} = I_{2,0}\kett{\sigma_i}+I_{0,2}\kett{\sigma_l\sigma_i\sigma_l}
         =(I_{2,0}-I_{0,2})\kett{\sigma_i} = a\kett{\sigma_i}
    \end{align}
    where we defined $a = I_{2,0} - I_{0,2}$. Thus, $\widehat{\tau}_{\gamma_{l,{\rm sym}}}$ is diagonal in the Pauli basis $\{\kett{\sigma_i}\}$, and its action can be expressed in terms of the diagonal matrix $\mu_{l}$. From this we can see that given a general observable $\kett{O} = \sum_i \alpha_i\kett{\sigma_i}$, where $\kett{\sigma_i}$ are the elements of the Pauli basis, we will have that $\widehat{\tau}_{\gamma_{l,{\rm sym}}}\kett{\sigma_i}\kett{O} = \mu_{l}\kett{O}$. Similarly, $\prod_l\widehat{\tau}_{\gamma_{l,{\rm sym}}}\kett{O} = \prod_l\mu_{l}\kett{O}$. Finally, adding the  Clifford gates allows us to write 
    \begin{align}
        \widehat{\tau}_{\bm{\gamma}_{\rm sym}}\kett{O} = U(\vec{0})\otimes U^*(\vec{0}) \prod_l \mu_{l}\kett{O} = U(\vec{0})\otimes U^*(\vec{0}) \mu_{\bm{\gamma}_{\rm sym}}\kett{O}\, ,
    \end{align}
    where we defined $\prod_l \mu_{l} = \mu_{\bm{\gamma}_{\rm sym}}$, and which recovers Eq.~\eqref{eq:result_lemma_diagonal} as intended.
\end{proof}

\subsection{Lemma~\ref{lemma:symmetric_compute_coef} and its proof}\label{app:lemma_compute_coef}
In this lemma we show how the matrix elements of $\mu_{\bm{\gamma}_{\rm sym}}$ (defined in the previous lemma) can be computed. Our strategy is based on the fact that $\mu_{\bm{\gamma}_{\rm sym}}$ is diagonal in the Pauli basis. 
\begin{lemma}\label{lemma:symmetric_compute_coef}
    For any circuit and initialization of the form in Lemma~\ref{lemma:symmetric_diagonal} with no Clifford gates\footnote{If we do have Clifford gates the same result can be applied, but we will have to consider the circuit with the permuted Clifford gates, as shown in Eq.~\eqref{eq:rewrite_absorb_clifford}.}, and for $\kett{O} = \sum_i\alpha_i\kett{\sigma_i}$, one obtains
    \begin{equation}
        \widehat{\tau}_{\vec{\gamma}}\kett{O} = \sum_i\alpha_i \prod_{l \, {\rm s.t.}\, \{\sigma_l,\sigma_i\}=0} a_l \kett{\sigma_i}
    \end{equation}
\end{lemma}

\begin{proof}
    The proof of this is trivial. Because symmetric initialization does not create branching (i.e. $\widehat{\tau}_{\bm{\gamma}_{\rm sym}}\kett{O} = \mu_{\bm{\gamma}_{\rm sym}} \kett{O}$), we will either get (or not get) a coefficient depending on the commutation between the generators and the initial terms in $O$.
\end{proof}

\section{Examples}\label{ap-sec:examples}
In this section we expand the calculations for Examples~\ref{ex:HVA}, ~\ref{ex:MaxCut} and~\ref{ex:different_initialisation_different_points} in the main text (we skip Example~\ref{ex:HVA-identity}, as its result follows directly from the proof of Examples~\ref{ex:HVA}). 

\subsection{Example~\ref{ex:HVA}}\label{app:hva}

    Let us consider the following example: a Hamiltonian Variational Ansatz for a one-dimensional $XYZ$ Heisenberg chain. Here we have $\GC=\{\sx^{(i)}\sx^{(i+1)},\sy^{(i)}\sy^{(i+1)},\sz^{(i)}\sz^{(i+1)}\}_{i=1}^{n-1}\cup\{\sz^{(i)}\}_{i=1}^n$. Thus the circuit at hand is 
    \begin{align}
        U(\thv) = \prod_{l=1}^L \Big[\prod_{j=1}^{n-1} e^{-i\theta_{j,l}/2 \sx^{(j)}\sx^{(j+1)}}
        \prod_{k=1}^{n-1} e^{-i\theta_{k,l}/2 \sy^{(k)}\sy^{(k+1)}}
        \prod_{a=1}^{n-1} e^{-i\theta_{a,l}/2 \sz^{(a)}\sz^{(a+1)}}
        \prod_{b=1}^{n-1} e^{-i\theta_{b,l}/2 \sz^{(b)}}\Big]\,.
    \end{align}
    And consider that all the parameters are initialized with the same symmetric initialization strategy $P_{\gamma_{\rm sym}}(\th)$. Furthermore, consider that we have an observable of the form $O=\sum_i \sx^{(i)}\sy^{(i+1)}$.

    This fulfills the conditions to apply Lemmas~\ref{lemma:symmetric_diagonal},~\ref{lemma:symmetric_compute_coef}. Indeed, from Lemma~\ref{lemma:symmetric_diagonal}, we have that the associated moment operator $\widehat{\tau}_{\vec{\gamma}}$ is diagonal in the Pauli basis. Moreover, we can apply Lemma~\ref{lemma:symmetric_compute_coef} to compute the scaling of the operator expressivity for our choice of measurement operator $O$. Specifically, Lemma~\ref{lemma:symmetric_compute_coef} states that we only need to consider the commutation relations between $\kett{\sigma_i}$ and the set of generators:
    \begin{align}
        &[\sx^{(i)}\sx^{(i+1)}, \sx^{(j)}\sy^{(j+1)}]\neq 0 \, {\rm for }\, i=j \, {\rm or }\, j+1 = i\,,\\
        &[\sy^{(i)}\sy^{(i+1)}, \sx^{(j)}\sy^{(j+1)}]\neq 0 \, {\rm for }\, i=j \, {\rm or }\, j-1 = i\,,\\
        &[\sz^{(i)}\sz^{(i+1)}, \sx^{(j)}\sy^{(j+1)}]\neq 0 \, {\rm for }\, i=j \,  {\rm or }\, j\pm 1 = 1\,,\\
        &[\sz^{(i)}, \sx^{(j)}\sy^{(j+1)}]\neq 0 \, {\rm for }\, i=j \,  {\rm or }\, j+1 = i\,.
    \end{align}

 From here, we can apply Lemma~\ref{lemma:symmetric_compute_coef} to obtain (note that care must be taken with the border condition)
    \begin{align}
        \widehat{\tau}_{\vec{\gamma}_{\rm sym}}\kett{O} = &\sum_{i = 2}^{n-2}a^{9 L}\kett{\sx^{(i)}\sy^{(i+1)}} + 
        a^{8 L}\kett{\sx^{(1)}\sy^{(2)}}+a^{8 L}\kett{\sx^{(n-1)}\sy^{(n)}}\,,
    \end{align}
    where we computed the coefficients $a,\, b$ (in this example $b=0$ because the initialization is symmetric) according to the definitions given in Eqs.~\eqref{eq:def_a_main}, and~\eqref{eq:def_b_main}.

    Similarly, consider the case where instead of initializing all the parameters with the same distribution $P_{\gamma_{\rm sym}}(\th)$, we now initialize all the gates with $\sz^{(1)}$ according to $P_{\gamma'_{\rm sym}}(\th)$. Then, defining $a' = \int P_{\gamma'_{\rm sym}}(\th) \cos \theta$, we find
     \begin{align}
            \widehat{\tau}_{\vec{\gamma}'_{\rm sym}}\kett{O} = &\sum_{i = 2}^{n-2}a^{9 L}\kett{\sx^{(i)}\sy^{(i+1)}} +a^{7 L}(a')^L\kett{\sx^{(1)}\sy^{(2)}}+a^{8 L}\kett{\sx^{(n-1)}\sy^{(n)}}\,.
    \end{align}
This means that 
\begin{align}
    \xi_{\vec{\gamma}',\vec{\gamma}}
    = \|\widehat{\tau}_{\vec{\gamma}'_{\rm sym}}\kett{O}
      -\widehat{\tau}_{\vec{\gamma}_{\rm sym}}\kett{O}\|_1
    = a^{7L}\left|a^L-(a')^L\right|\,.
\end{align}
as stated in the main text.

\subsection{Example~\ref{ex:MaxCut}}\label{app:MaxCut}
Let us consider a MaxCut problem defined on a graph $G=(V,E)$ with $n$ vertices. The measurement operator is the problem Hamiltonian $O = \sum_{(w,w')\in E}\sz^{(w)}\sz^{(w')}$. We evaluate a parameterized quantum circuit of the form
\begin{align}
    U(\vec{\theta}) &= V_{L+1}\left(\theta^V_{L+1}\right) \prod_{l=1}^{L} \left( \prod_{j=1}^{n} \text{R}_y^{(j)}\left(\theta^y_{l,j}\right)\, V_{l,j}\left(\theta^V_{l,j}\right) \right)\,,
\end{align}
where the operators $V_{l,j}$ and $V_{L+1}$ are diagonal in the computational basis (i.e., generated by products of $\sz$).

We define a family of initializations $\vec{\gamma}_{\vec{c}}$ parameterized by a binary matrix $\vec{c} \in \{0,1\}^{L \times n}$. Under this initialization strategy, the parameter $\theta^y_{l,j}$ for the $\text{R}_y$ gate on qubit $j$ at layer $l$ is independently drawn from a Gaussian distribution:
\begin{align}\label{eq:distribution_MaxCut}
    P_{\gamma_{l,j}}(\theta) &= \NC(\pi c_{l,j},\varsigma^2)\,.
\end{align}

The full moment operator for the circuit is the product of the individual moment operators of each layer acting on the observable $\kett{O}$:
\begin{align}
    \widehat{\tau}_{\vec{\gamma}_{\vec{c}}}\kett{O} &= \widehat{\tau}_{V_{L+1}} \prod_{l=1}^{L} \left( \prod_{j=1}^{n} \widehat{\tau}_{\gamma_{l,j}} \widehat{\tau}_{V_{l,j}} \right) \kett{O}\,.
\end{align}

To evaluate this under the distribution in Eq.~\eqref{eq:distribution_MaxCut}, we first track the action of a single $\text{R}_y$ moment operator on a Pauli string. The generators are $\sy^{(j)}$. For any edge term $\sz^{(w)}\sz^{(w')}$, the generator commutes with the operator unless $j=w$ or $j=w'$ (i.e., $[\sy^{(j)}, \sz^{(w)}\sz^{(w')}] \neq 0 \iff j \in \{w, w'\}$). If they commute, $\widehat{\tau}_{\gamma_{l,j}}$ acts as the identity. If they do not commute (for instance, when acting on $\sz^{(j)}$), we evaluate the twirl explicitly. Because the integration domain is symmetric around $c_{l,j}\pi$, the sine term vanishes ($\int P_{\gamma_{l,j}}(\theta)\sin(\theta)d\theta = 0$), yielding a strictly diagonal operator:
\begin{align}\label{eq:sy_applied_MaxCut}
\widehat{\tau}_{\gamma_{l,j}}\kett{\sz^{(j)}} &= \int d\th P_{\gamma_{l,j}}(\th) \kett{e^{-i\th\sy^{(j)}/2} \sz^{(j)} e^{i\th\sy^{(j)}/2}} \nonumber\\
    &= \int d\th P_{\gamma_{l,j}}(\th) \left(\cos(\th)\kett{\sz^{(j)}} + \sin(\th)\kett
    {\sx^{(j)}}\right) \nonumber\\
    &= a_{l,j} \kett{\sz^{(j)}}\,,
\end{align}
where the remaining coefficient evaluates to:
\begin{align}\label{eq:a_i_maxcut}
    a_{l,j} &= \frac{1}{\sqrt{2\pi \varsigma^2}}\int_{-\infty}^{\infty}e^{-\frac{(\theta - \pi c_{l,j})^2}{2\varsigma^2}} \cos(\theta) d\theta = (-1)^{c_{l,j}} e^{-\varsigma^2/2} = (-1)^{c_{l,j}} |a| \,,
\end{align}
with $|a| = e^{-\varsigma/2}$.

Because $V_{l,j}$ and $V_{L+1}$ are diagonal in the computational basis, they commute exactly with the observable $O$. For any operator $V$ that commutes with $O$, the twirl leaves the operator invariant, $\widehat{\tau}_{V}\kett{O} = \kett{V O V^\dagger} = \kett{O}$. Thus, the diagonal layers cancel out of the equation entirely:
\begin{align}
\widehat{\tau}_{\vec{\gamma}_{\vec{c}}}\kett{O} &= \prod_{l=1}^{L} \prod_{j=1}^{n} \widehat{\tau}_{\gamma_{l,j}} \kett{O}\,.
\end{align}

Finally, we can apply Eqs.~\eqref{eq:sy_applied_MaxCut} and~\eqref{eq:a_i_maxcut} on every qubit and  layer to the full observable $\kett{O}$. Here, each edge term $\kett{\sz^{(w)}\sz^{(w')}}$ accumulates a factor of $a_{l,w}$ and $a_{l,w'}$ at each layer $l$:
\begin{align}
    \widehat{\tau}_{\vec{\gamma}_{\vec{c}}}\kett{O} &= \sum_{(w,w')\in E} \left( \prod_{l=1}^L a_{l,w} a_{l,w'} \right) \kett{\sz^{(w)}\sz^{(w')}} \nonumber\\
    &= \sum_{(w,w')\in E} \left( \prod_{l=1}^L (-1)^{c_{l,w}} |a| \cdot (-1)^{c_{l,w'}} |a| \right) \kett{\sz^{(w)}\sz^{(w')}} \nonumber\\
    &= a^{2L} \sum_{(w,w')\in E} (-1)^{\sum_{l=1}^L (c_{l,w} + c_{l,w'})} \kett{\sz^{(w)}\sz^{(w')}}\,.
\end{align}
Defining $C_w(\vec{c}) = \sum_{l=1}^L c_{l,w} \pmod 2$ as the cumulative parity of shifts on qubit $w$, we can rewrite the previous equation as:
\begin{align}
    \widehat{\tau}_{\vec{\gamma}_{\vec{c}}}\kett{O} &= a^{2L} \sum_{(w,w')\in E}(-1)^{C_w(\vec{c}) + C_{w'}(\vec{c})}\kett{\sz^{(w)}\sz^{(w')}}\,.
\end{align}

From here, we compute the expected operator gap between two initializations $\vec{\gamma}_{\vec{c}}$ and $\vec{\gamma}_{\vec{c}'}$:
\begin{align}
    \xi_{\vec{\gamma}_{\vec{c}},\vec{\gamma}_{\vec{c}'}} 
    &= \left\|\widehat{\tau}_{\vec{\gamma}_{\vec{c}}}\kett{O} - \widehat{\tau}_{\vec{\gamma}_{\vec{c}'}}\kett{O}\right\|_1 \nonumber\\
    &= a^{2L} \sum_{(w,w')\in E} \left|(-1)^{C_w(\vec{c}) + C_{w'}(\vec{c})} - (-1)^{C_w(\vec{c}') + C_{w'}(\vec{c}')}\right|\,.
\end{align}
The absolute difference evaluates to $2$ if $C_w(\vec{c}) \oplus C_{w'}(\vec{c}) \neq C_w(\vec{c}') \oplus C_{w'}(\vec{c}')$, and $0$ otherwise. Thus:
\begin{align}
    \xi_{\vec{\gamma}_{\vec{c}},\vec{\gamma}_{\vec{c}'}} &= 2 a^{2L} \big| \{ (w,w') \in E \mid C_w(\vec{c}) \oplus C_{w'}(\vec{c}) \neq C_w(\vec{c}') \oplus C_{w'}(\vec{c}') \} \big|\,.
\end{align}

We now determine the exact number of structurally distinct initializations this generates. Under this strategy, the sign assigned to any edge $e = (w,w')$ is strictly $(-1)^{C_w \oplus C_{w'}}$. In general, tracking this quantity is highly complex, but we can simplify it by examining the all-to-all graph topology. 

Assume an all-to-all graph where $O = \sum_{i<j}\sz^{(i)}\sz^{(j)}$. In this scenario, the number of distinct initializations corresponds to the number of unique edge-sign configurations that can be generated. For any vertex $w$, let $v_w = (-1)^{C_w(\vec{c})}$. The effective sign on an edge is given by the product $v_w v_{w'}$. There are $2^n$ possible configurations for the vertex vector $\vec{v} \in \{-1, 1\}^n$. Because a global sign flip ($\vec{v} \to -\vec{v}$) leaves every edge product $v_w v_{w'}$ invariant, this global sign symmetry reduces the number of unique configurations by a factor of exactly 2. Thus, the total number of distinct observable initializations we can generate is $2^{n-1}$.

\subsection{Example~\ref{ex:different_initialisation_different_points}}\label{app:different_initialisation_different_points}

Here we consider a circuit composed of two layers of $C_z$ entangling gates and single-qubit parameterized gates, acting on the $j$-th qubit, of the form $U_{(j)}(\thv) = e^{-i \theta_{1}^{(j)} \sx/2}e^{-i \theta_{2}^{(j)} \sy/2}e^{-i \theta_{3}^{(j)} \sz/2}$. To ease the notation moving forward, we define $V$ as the two layers of $C_z$ that appear in between the single qubit gates, i.e.
\begin{equation}
    V = \prod_{i} C_{z}^{(2i,2i+1)}\prod_{j} C_{z}^{(2j+1,2j+2)}
\end{equation}

To make use of Lemmas~\ref{lemma:symmetric_diagonal} and~\ref{lemma:symmetric_compute_coef}, we push the action of the Clifford to the end of the circuit as in Eq.~\eqref{eq:rewrite_absorb_clifford}. For our PQC, this is especially easy in this case because $V = V^\dagger$\footnote{This can be seen just by inspection of $C_z^{(i,j)}$. Crucially $C_z^{(i,j)}$ is diagonal in the same basis $\forall \, i,j$, and thus they commute. Moreover $C_{z}^{(i,j)}C_{z}^{(i,j)} = \1$. }. Therefore, we can rewrite the circuit as follows 
\begin{align}
    U(\thv) =& \prod_{l=1}^L V \bigotimes_{j=1}^n U_{(j)}(\thv_{j,l})\label{eq:circuit_advanced_example}  = \prod_{l'=1}^{L/2} \left [V \bigotimes_{j=1}^n U_{(j)}(\thv_{j,2l'}) V \right] \bigotimes_{j=1}^n U_{(j)}(\thv_{j,2l'-1})\, ,
\end{align}
where we used that $V = V^\dagger$, and grouped it with the previous layer. Note that we made the implicit assumption that $L/2\in \mathbb{N}$, however the correction would be easy if $L$ was an odd number. With this rearrangement of the non-parametrized gates we can define $\widetilde{U}_{(j)}(\thv)$
\begin{align}
    \widetilde{U}_{(j)} &= V U_{(j)}(\thv) V  = \left( e^{-i \theta_{1} \sz^{(j-1)}\otimes\sx^{(j)}\otimes \sz^{(j+1)} /2} e^{-i \theta_{2} \sz^{(j-1)}\otimes\sy^{(j)}\otimes \sz^{(j+1)} /2} e^{-i \theta_{3} \sz^{(j)} /2}  \right)\,,
\end{align}
where we used that $V\sx^{(j)} V =  \sz^{(j-1)}\otimes\sx^{(j)}\otimes \sz^{(j+1)},\, V\sy V = \sz^{(j-1)}\otimes\sy^{(j)}\otimes \sz^{(j+1)}$. Thus we can rewrite the circuit in Eq.~\eqref{eq:circuit_advanced_example} as 
\begin{equation}
    U(\thv) = \prod_{l' = 1}^{L/2}\bigotimes_{j=1}^n U_{(j)}(\thv_{j,2l'})\widetilde{U}_{(j)}(\thv_{j,2l'-1})\,.
\end{equation}

Now take $O = \sum_{i=1}^{n-1} \sigma_x^{(i)}\otimes\sigma_x^{(i+1)}$. Using Lemma~\ref{lemma:symmetric_compute_coef} we get
\begin{align}
    \widehat{\tau}_{\vec{\gamma}_{\rm sym}}\kett{\sx^{(i)}\sx^{(i+1)}} = a^{12 L}\kett{\sx^{(i)}\sx^{(i+1)}}\,,
\end{align}
where we used that for a symmetric initialization $b = 0$ a. Hence, we can see that 
\begin{equation}
    \prod_{j=1}^n\widehat{\tau}_{\gamma_{j,\rm sym}}\kett{\sx^{(i)}\sx^{(i+1)}}  = a^4 \kett{\sx^{(i)}\sx^{(i+1)}}
\end{equation}
where we used that only the $U_{(i)}(\thv),\,U_{(i+1)}(\thv)$ do not commute with $\sigma_x^{(i)}\sigma_x^{(i+1)}$, and that $e^{-i\sigma_x}$ always commutes with $\sigma_x$, and therefore does not give a prefactor $a$. Similarly, we can compute 
\begin{equation}
    \prod_{j=1}^n\widehat{\tau}_{\gamma_{j,\rm sym}}\kett{\sx^{(i)}\sx^{(i+1)}}  = a^8 \kett{\sx^{(i)}\sx^{(i+1)}}\,,
\end{equation}
where we used that any gate that does not commute with the observable gives a prefactor $a$, and that in this case we used that the only terms that do not commute with $\sigma_x^{(i)}\otimes\sigma_x^{(i+1)}$ are the following
\begin{align}
    &\sigma_z^{(j-1)}\sigma_x^{(j)}\sigma_z^{(j+1)}, \, {\rm for}\, j \in \{i-1,i,i+1,i+2\}\,,\\
    &\sigma_z^{(j-1)}\sigma_y^{(j)}\sigma_z^{(j+1)}, \, {\rm for}\, j = i-1,\, j = i+2\,,\\
    &\sigma_z^{(j)}, \, {\rm for}\, j = i, \, j =i+1\,.
\end{align}

So the final observable after going through the circuit looks like
\begin{align}
    \widehat{\tau}_{\vec{\gamma}}\sum_i\kett{\sx^{(i)}\sx^{(i+1)}} = a^{11 L}\kett{\sx^{(1)}\sx^{(2)}} + a^{11 L}\kett{\sx^{(n-1)}\sx^{(n)}}+\sum_{i=2}^{n-2}a^{12 L}\kett{\sx^{(i)}\sx^{(i+1)}}\,,
\end{align}
where we have also applied edge conditions. 

If, however, we were to apply a different initialization distribution for the last parameterized gate $e^{-i\theta \sz^{(1)}}$ we find 
\begin{align}
    \widehat{\tau}_{\vec{\gamma}'}\sum_i\kett{\sx^{(i)}\sx^{(i+1)}} = ba^{11 L-1}\kett{\sy^{(1)}\sx^{(2)}}+ a^{11 L}\kett{\sx^{(n-1)}\sx^{(n)}}+\sum_{i=2}^{n-2}a^{12 L}\kett{\sx^{(i)}\sx^{(i+1)}}\,,
\end{align}
where we recall that $b = \int d\th P_{\gamma}(\th)\sin\th$ as defined in Eq.~\eqref{eq:def_b_main}. This leads to
\begin{equation}
        \xi_{\vec{\gamma}',\vec{\gamma}} = (|a| + |b|) a^{11L-1}\,.
\end{equation}

\section{Exponentially many high-dimensional slices in Max-Cut landscapes without barren plateaus} \label{app:1d-slices}

We here go beyond gradient analysis around critical points or involving small-angle initializations, and show that the optimization landscape of Example~\ref{ex:MaxCut} can exhibit an exponentially large number of high-dimensional slices with large gradients.  In particular, we prove Theorem~\ref{th:MaxCut-slices}, which we  recall for convenience.

\begin{theorem}[Informal] 
     Let $U(\thv)$ be a PQC composed of single-qubit gates interleaved with diagonal entangling gates $V_l$, i.e., 
       \begin{equation}\nonumber
        U(\thv) =\prod_{l=1}^L \prod_{j=1}^{n}e^{-i\th_{j,l}^z\sz^{(j)}/2} e^{-i\th_{j,l}^y\sy^{(j)}/2} V_l\,,
    \end{equation}
    where $L\in\OC\left(\poly(n)\right)$. Moreover, let $\rho=\ketbra{0}^{\otimes n}$ and let $O= \sum_{(w,w')\in E}\sigma_z^{(w)}\sigma_z^{(w')}$ be the Max-Cut Hamiltonian on a constant-odd-degree graph $G=(E,V)$, with $S\subseteq V$  a subset of vertices. Then, for each $j\in S$, let one $ \theta_{l_j,j}^y$ angle be initialized in a constant-size region, while the rest of the parameters are independently initialized following  Gaussian distributions $\NC(\theta^*,\varsigma^2)$, with $\theta^*\in\{0,\pi\}$ and $ \varsigma\in\Theta\left(\frac{1}{\poly(n)}\right)$. In this setting, all the $\theta^y_{l,j}$ parameters in the circuit are guaranteed to have gradients in $\Omega\left(\frac{1}{\poly(n)}\right)$ with high probability. 
\end{theorem}

\begin{proof}
    
Let us first note that if we choose a center $\pi\vec{c}\in\{0,\pi\}^{nL}$ for the joint distribution of $R_y$ angles, and we set $\vec{\theta}^y=\pi\vec{c}$ while letting $\vec{\theta}^D$ take some fixed arbitrary value, the circuit $U\left(\vec{\theta}\right)$ maps the all-zero state to a computational-basis state $\ket{\vec{z}(\vec{c})}\!\bra{\vec{z}(\vec{c})}$ (i.e., $U(\vec{c})|0\rangle^{\otimes n}=e^{i\phi}|\vec{z}(\vec{c})\rangle$) for some $\phi\in\mathbb{R}$), since
\begin{equation}
R_y(0)=\openone\,,\qquad R_y(\pi) = -iY\,,
\end{equation}
and additionally every diagonal gate preserves computational-basis states. That is, $R_y(0)$ preserves each computational-basis vector, while $R_y(\pi)$ flips the corresponding qubit, up to a phase.

Let us choose a coordinate $(l,j)$, vary only $\theta_{l,j}^y$ and compute the corresponding derivative. We can easily show that the derivative $\mathcal{L}(\vec{\theta})$ with respect to $\theta_{l,j}^y$ at $\vec{\theta}^y=\pi\vec{c}+\vec{\varepsilon}_{l,j}$ is given by
    \begin{equation}  \label{eq:derivative}
        \left.\frac{\partial}{\partial\theta_{l,j}^y} \mathcal{L}(\vec{\theta})\right|_{\vec{\theta}^y=\pi\vec{c}+\vec{\varepsilon}_{l,j}} = \frac{\sin \varepsilon_{l,j}}{2} \Delta_j(\vec{c})  \,.
    \end{equation}

    To see this, we first note that since the diagonal gates leave  computational-states intact, and all $R_y$ rotations have angles in $\{0,\pi\}$, the action of the circuit $U\left(\vec{\theta}\right)$ on $\ket{0}^{\otimes n}$ can be written as
    \begin{align}
       &U\left(\vec{\theta}\right) \ket{0}^{\otimes n} =  \cos(\varepsilon_{l,j}/2) \ket{\vec{z}(\vec{c})} + e^{i\phi_{l,j}} \sin(\varepsilon_{l,j}/2) \sigma_y^{(j)} \ket{\vec{z}(\vec{c})}\,,
    \end{align}
for some phase $\phi_{l,j}$. Hence, denoting $\ket{\vec{z}^j(\vec{c})} =\sigma_y^{(j)}\ket{\vec{z}(\vec{c})}$, with $\vec{z}=(z_1,\ldots, z_n)^T\in\{1,-1\}^n$, we arrive at
\begin{align}
    \mathcal{L}(\vec{\theta}) = \cos(\varepsilon_{l,j}/2)^2 \bra{\vec{z}(\vec{c})}O\ket{\vec{z}(\vec{c})}
  +\sin(\varepsilon_{l,j}/2)^2 \bra{\vec{z}^j(\vec{c})}O\ket{\vec{z}^j(\vec{c})}\,,
\end{align}
since $\vec{z}(\vec{c})\neq \vec{z}^j(\vec{c})$ and so the off-diagonal terms $\bra{\vec{z}^j(\vec{c})}O\ket{\vec{z}(\vec{c})} $ and  $\bra{\vec{z}(\vec{c})}O\ket{\vec{z}^j(\vec{c})} $ are both zero. Taking the derivative,
\begin{align}
    \left.\frac{\partial}{\partial\theta_{l,j}^y} \mathcal{L}(\vec{\theta})\right|_{\vec{\theta}^y=\pi\vec{c}+\vec{\varepsilon}_{l,j}}  &= - \cos(\varepsilon_{l,j}/2)\sin(\varepsilon_{l,j}/2) \bra{\vec{z}(\vec{c})}O\ket{\vec{z}(\vec{c})}+ \sin(\varepsilon_{l,j}/2)\cos(\varepsilon_{l,j}/2) \bra{\vec{z}^j(\vec{c})}O\ket{\vec{z}^j(\vec{c})}\nonumber \\  &= \frac{\sin \varepsilon_{l,j}}{2} \left( \bra{\vec{z}^j(\vec{c})}O\ket{\vec{z}^j(\vec{c})} -  \bra{\vec{z}(\vec{c})}O\ket{\vec{z}(\vec{c})}\right)\nonumber \\&=\frac{\sin \varepsilon_{l,j}}{2} \Delta_j  \,.
\end{align}

 We can rewrite $\Delta_j$ in graph-theoretic terms by denoting the  set of edges $E=(j,j')$ of the graph that are incident on vertex $j$ and such that $z_j=z_{j'}$,  as  $+_j$, and those incident to  vertex $j$ but such that $z_j\neq z_{j'}$, as $-_j$, i.e.,
\begin{equation}
+_j(\vec{z})=\{(j,j')\in E:\ z_j=z_{j'}\},\;\;
-_j(\vec{z})=\{(j,j')\in E:\ z_j\neq z_{j'}\}\,. \nonumber
\end{equation}
Then, only the edges incident on $j$ change sign when the $j$-th qubit is flipped, and hence when the graph is unweighted, we have
\begin{align} \label{eq:Delta=+_-}
\Delta_j(\vec{z}) = -2\sum_{j' :\, (j,j')\in E} z_j z_{j'} = -2\left(|+_j|-|-_j|\right)\,.
\end{align}
This interpretation of the gap $\Delta_j(\vec{z})$ will allow us to simplify the analysis and find families of graph where exponentially many barren-plateau avoiding initializations are possible. 
A particularly simple family of graphs that uniformly satisfy that  $\Delta_j(\vec{c})\in\Omega(1)$ consists of those with odd-degree in every vertex. We just need to realize that
\begin{equation}
|+_j|+|-_j|=\deg(j)\,.
\end{equation}
Because the degree $\deg(j)$ is odd, the sum $|+_j|+|-_j|$ is an odd integer, and therefore one summand is odd and the other even. We thus arrive at
\begin{equation}
{\rm abs}\left(|+_j|-|-_j|\right)\ge 1\,.
\end{equation}
Consequently, using Eq.~\eqref{eq:Delta=+_-} we can see that
\begin{equation} \label{eq:Delta-odd-graphs}
|\Delta_j(\vec{z})|\ge 2\,.
\end{equation}

\medskip

Next we extend the previous one-dimensional slice analysis to higher-dimensional slices. 
Let us fix a center $\pi\vec{c}\in\{0,\pi\}^{nL}$ for all $R_y$ angles, and let $S\subseteq V$ be a subset of qubits. Let us furthermore denote $\vec{x}\in\{0,1\}^{|S|}$. Then, we  perturb one $R_y$ angle for each $j\in S$ (at some layer $l_j$) by an amount $\varepsilon_{l_j,j}$, while keeping all other $R_y$ angles fixed at their $0$ or $\pi$ values. We denote this perturbation by $\vec{\varepsilon}$. We find
\begin{equation} 
    U(\vec{\theta}) \ket{0}^{\otimes n} = \sum_{\vec{x}\in\{0,1\}^S} e^{i\phi(\vec{x})} \prod_{j\in S}\left(\cos(\varepsilon_{l_j,j}/2)\right)^{1-x_j}\left(\sin(\varepsilon_{l_j,j}/2)\right)^{x_j} \ket{\vec{z}(\vec{c})^{\vec{x}}}\,,
\end{equation}
where $ \ket{\vec{z}(\vec{c})^{\vec{x}}}$ is the computational-basis state obtained from $\ket{\vec{z}(\vec{c})}$ by flipping the qubits $j\in S$ for which $x_j=1$. Notice that each qubit $j\in S$ is flipped with probability $\sin^2(\varepsilon_{l_j,j}/2)$ and left untouched with probability $\cos^2(\varepsilon_{l_j,j}/2)$. 
Hence, in the $|S|$-dimensional slice, the cost function takes the form
\begin{equation}
    \LC(\vec{\theta}) = \sum_{\vec{x}\in\{0,1\}^S} \prod_{j\in S}\left(\cos^2(\varepsilon_{l_j,j}/2)\right)^{1-x_j}\left(\sin^2(\varepsilon_{l_j,j}/2)\right)^{x_j} \left( \sum_{(w,w')\in E}z_w z_{w'}(-1)^{x_w+x_{w'}} \right)\,,
\end{equation}
where we set $x_j=0$ for $j\notin S$, and we obtain
\begin{equation}
    \LC(\vec{\theta})=\sum_{(w,w')\in E}z_w z_{w'}\,\mathbb{E}_{\vec{x}}\left[(-1)^{x_w+x_{w'}}\right]\, . 
\end{equation}
Since the bit flips are independent,
\begin{equation} 
    \mathbb{E}_{\vec{x}}\left[(-1)^{x_w+x_{w'}}\right]=\mathbb{E}_{\vec{x}}\left[(-1)^{x_w}\right]\mathbb{E}_{\vec{x}}\left[(-1)^{x_{w'}}\right]\, . 
\end{equation}
For $j\in S$, we have
\begin{equation} 
    \mathbb{E}_{\vec{x}}\left[(-1)^{x_j}\right]=\cos^2(\varepsilon_{l_j,j}/2)-\sin^2(\varepsilon_{l_j,j}/2)=\cos(\varepsilon_{l_j,j})\, ,
\end{equation}
while for $j\notin S$, 
\begin{equation} 
    \mathbb{E}_{\vec{x}}\left[(-1)^{x_j}\right]=1\, . \end{equation}
Defining
\begin{equation} 
    \Gamma_j=\begin{cases}
        \cos(\varepsilon_{l_j,j})\quad \text{for } j\in S \\ 1\quad\; \text{for } j\notin S
    \end{cases} \,,
\end{equation}
we arrive at the formula for the cost function on the $|S|$-dimensional slice,
\begin{equation} \label{eq:loss-gammas}
    \LC(\vec{\theta})=\sum_{(w,w')\in E}z_w z_{w'}\,\Gamma_w\Gamma_{w'}\, .
\end{equation}

We now differentiate this expression. For $j\in S$, we find
\begin{equation}\label{eq:slice-derivative}
    \left.\frac{\partial \LC(\vec{\theta})}{\partial \varepsilon_{l_j,j}}\right|_{\vec{\theta}^y=\pi\vec{c}+\vec{\varepsilon}}=-\sin(\varepsilon_{l_j,j})z_j\sum_{j' :\, (j,j')\in E}z_{j'}\Gamma_{j'}\, .
\end{equation}
Notice here that we recover the one-dimensional formula~\eqref{eq:derivative} as a special case, since if only the coordinate $j$ is perturbed, then $\Gamma_{j'}=1$ for all $j'$ above, and therefore
\begin{equation} 
\left.\frac{\partial \LC(\vec{\theta})}{\partial \varepsilon_{l_j,j}}\right|_{\vec{\theta}^y=\pi\vec{c}+\vec{\varepsilon}}=-\sin(\varepsilon_{l_j,j})z_j\sum_{j' :\, (j,j')\in E}z_{j'}\, .
\end{equation}
Using Eq.~\eqref{eq:Delta=+_-}, we find that Eq.~\eqref{eq:derivative} holds.

A particularly simple high-dimensional case arises when the qubit indices in $S$ are such that $(j,j')\notin E$ for every $j,j'\in S$ (in graph-theoretic terms, the edges $(j,j')$ for $j\in S$ are independent and hence $S$ is called an independent vertex set). Therefore $\Gamma_{j'}=1$ for all $j'$ such that $(j,j')\in E$, and we obtain
\begin{equation} 
    \left.\frac{\partial \LC(\vec{\theta})}{\partial \varepsilon_{l_j,j}}\right|_{\vec{\theta}^y=\pi\vec{c}+\vec{\varepsilon}}=\frac{\sin(\varepsilon_{l_j,j})}{2}\Delta_j(\vec{z})\quad \text{for every } j\in S\, .
\end{equation}
Consequently, if $|\Delta_j(\vec{z})|\in\Omega\left( 1/\poly(n)\right)$, and if $|\sin(\varepsilon_{l_j,j})|\in\Omega\left( 1/\poly(n)\right)$, then the gradient entry has inverse-polynomial magnitude.

For constant-degree graphs of size $\Theta(n)$, one can easily choose an independent set $S$ of linear size~\cite{henning2018tight}. In particular, for constant-odd-degree graphs $|\Delta_j(\vec{z})|\geq 2$ (see Eq.~\eqref{eq:Delta-odd-graphs}), and hence we obtain $\Theta(n)$-dimensional slices where the gradients are at most polynomially vanishing with $n$, as long as $\varepsilon_{l_j,j}$ is not exponentially close to $0$ or $\pi$ for every $j$. However, from the point of view of the optimization dynamics, these $\Theta(n)$-dimensional slices are completely decoupled from each other. Thus, if the optimization is constrained to stay within the slices, the actions of the $\Theta(n)$ variational parameters are independent as they all affect different terms in the cost function.

 Instead of choosing an independent set $S$ (which is the ultimate reason behind the parameters' optimization trajectories on the slices being independent), we can choose $S$ such that each vertex index couples to at most a constant number $K\in\OC(1)$ of other vertices in $S$. In this case, the light cones of different perturbed parameters are allowed to overlap, but only a constant number of times per gate.

Using Eq.~\eqref{eq:slice-derivative}, the gradient along a slice is found to be
\begin{equation}
    \left.\frac{\partial \LC(\vec{\theta})}{\partial \varepsilon_{l_j,j}}\right|_{\vec{\theta}^y=\pi\vec{c}+\vec{\varepsilon}}=-z_{j}\sin(\varepsilon_{l_j,j})\left(\sum_{j'\notin S:\, (j,j')\in E} z_{j'}+\sum_{j'\in S :\, (j,j')\in E}z_{j'}\cos(\varepsilon_{l_{j'},j'})\right)\,.
\end{equation}
Hence, each gradient component depends only on the parameters inside the induced light cone given by the indices $j'\in S$ such that $(j,j')\in E$, which is now of size $\OC(1)$. Let us denote
\begin{equation}
    F_j(\vec{\varepsilon}) = z_j\left(\sum_{j'\notin S:\, (j,j')\in E} z_{j'}+\sum_{j'\in S :\, (j,j')\in E}z_{j'}\cos(\varepsilon_{l_{j'},j'})\right)\,,
\end{equation}
We find (see Eq.~\eqref{eq:Delta=+_-})
\begin{equation}
    F_j(\vec{0})= -\frac{1}{2} \Delta_j(\vec{z})\,,
\end{equation}
and
\begin{equation} \label{eq:graph-degree-bound}
    \left| F_j(\vec{\varepsilon}) - F_j(\vec{0})\right| = \left|\sum_{j'\in S :\, (j,j')\in E}z_{j}z_{j'}\left(1-\cos(\varepsilon_{l_{j'},j'})\right)\right| \leq K \max_{j'\in S} \left(1-\cos(\varepsilon_{l_{j'},j'})\right)\,.
\end{equation}
Using the reverse triangle inequality,
\begin{equation} 
    |F_j(\vec{\varepsilon})| \geq | F_j(\vec{0})| - \left| F_j(\vec{\varepsilon}) - F_j(\vec{0})\right| \geq \frac{1}{2} |\Delta_j(\vec{z})| - K \max_{j'\in S} \left(1-\cos(\varepsilon_{l_{j'},j'})\right) \,.
\end{equation}
Then choosing the angles $\vec{\varepsilon}$ such that
\begin{equation}
    \alpha \min_{j\in S}|\Delta_j(\vec{z})| > K \max_{j'\in S} \left(1-\cos(\varepsilon_{l_{j'},j'})\right) \,,
\end{equation}
with $\alpha\in\left(0,\frac{1}{2}\right)$ a constant, then
\begin{equation}
    |F_j(\vec{\varepsilon)}|\geq \left(\frac{1}{2} -\alpha\right) \min_{j\in S}|\Delta_j(\vec{z})|\,,
\end{equation}
and hence
\begin{equation} \label{eq:derivative-bound}
    \left| \left. \frac{\partial \LC(\vec{\theta})}{\partial \varepsilon_{l_j,j}}\right|_{\vec{\theta}^y=\pi\vec{c}+\vec{\varepsilon}}\right|\geq  |\sin(\varepsilon_{l_j,j})| \left(\frac{1}{2} -\alpha\right)   \min_{j\in S}|\Delta_j(\vec{z})| \,.
\end{equation}
Consequently, whenever $\min_{j\in S}|\Delta_j(\vec{z})|\in\Omega(1/\mathrm{poly}(n))$ and $|\sin(\varepsilon_{l_j,j})|\in\Omega(1/\mathrm{poly}(n))$, all gradient components on the slice are at most inverse-polynomially vanishing.

In particular, for constant-odd-degree graphs, one has $|\Delta_j(\vec{z})|\geq 2$. Hence, if $K\in \OC(1)$, we obtain a constant-width hypercube around the computational-basis center $\pi \vec{c}$ in which the gradients are not exponentially small (except when $\varepsilon_{l_j,j}$ is exponentially close to $0$ or $\pi$), with half width $r\geq \max_j |\varepsilon_{l_j,j}|$ and $r\in \OC(1)$. Since the dimension of the slice is $|S|\in\Theta(n)$, this region has absolute volume $(2r)^{|S|}$. 
Thus, whenever $2r>1$, the trainable region in parameter space has exponentially growing absolute volume. For example, when $K=3$, if we restrict the parameters $\varepsilon_{l_j,j}$ such that $0< r \leq \pi$, we find
\begin{equation}
     \left| F_j(\vec{\varepsilon}) - F_j(\vec{0})\right| \leq 3 \left|\max_{j'\in S} (1-\cos(\varepsilon_{l_{j'},j'}))\right| \leq  3 \left| (1-\cos(r))\right|\,,
\end{equation}
so that choosing e.g., $\alpha=\frac{1}{4}$, we arrive at
\begin{equation}
    \frac{1}{2} > 3 |1-\cos(r)|\,.
\end{equation}
This condition can be satisfied with a constant $r< {\rm arccos\left(\frac{5}{6}\right)}\approx 0.586$, i.e., with a constant $r>\frac{1}{2}$, so that the corresponding coupled slices contain an exponentially large trainable region.

More generally, if $K\in\Theta(\log n)$, the same argument still applies, and one can choose $r\in\Theta\left(\frac{1}{\sqrt{\log n}}\right)$ to guarantee inverse-polynomial gradients, so the absolute volume shrinks as $n$ increases. To see this, we just notice that  $1-\cos(r)\leq \frac{r^2}{2}$, so that
\begin{equation}
   2\alpha > K \frac{r^2}{2} \quad\Rightarrow\quad r< \sqrt{\frac{4\alpha}{K}} \,.
\end{equation}

Finally, we show that every fixed $R_y$ parameter in the trainable slices can actually be initialized with an appropriate Gaussian distribution around $0$ or $\pi$, while still avoiding barren plateaus. Indeed, let the $R_y$ angles be
\begin{equation}
    \vec{\theta}^y=\pi \vec{c}+\vec{\varepsilon} + \vec{\eta} \,,
\end{equation}
where $\vec{\eta}$ are Gaussian perturbations for the fixed parameters in the slice (i.e., the analysis above corresponds to $\vec{\eta}=\vec{0}$). We will again make use of the reverse triangle inequality,
\begin{equation}\label{eq:reverse-triangle-ineq}
     \left|\left.\frac{\partial \LC}{\partial \varepsilon_{l_j,j}}\right|_{\vec{\theta}^y = \pi \vec{c}+\vec{\varepsilon}+\vec{\eta}} \right| \geq \left|\left.\frac{\partial \LC }{\partial \varepsilon_{l_j,j}}\right|_{\vec{\theta}^y = \pi \vec{c}+\vec{\varepsilon}+\vec{0}}\right| - \left| \left.\frac{\partial \LC}{\partial \varepsilon_{l_j,j}}\right|_{\vec{\theta}^y = \pi \vec{c}+\vec{\varepsilon}+\vec{\eta}} -  \left.\frac{\partial \LC}{\partial \varepsilon_{l_j,j}}\right|_{\vec{\theta}^y = \pi \vec{c}+\vec{\varepsilon}+\vec{0}}\right|\,.
\end{equation}
The idea is to show that 
\begin{equation}\label{eq:fraction-condition}
    \left| \left.\frac{\partial L}{\partial \varepsilon_{l_j,j}}\right|_{\vec{\theta}^y = \pi \vec{c}+\vec{\varepsilon}+\vec{\eta}} -  \left.\frac{\partial \LC}{\partial \varepsilon_{l_j,j}}\right|_{\vec{\theta}^y = \pi \vec{c}+\vec{\varepsilon}+\vec{0}}\right| \leq \frac{1}{2} \left|\left.\frac{\partial \LC}{\partial \varepsilon_{l_j,j}}\right|_{\vec{\theta}^y = \pi \vec{c}+\vec{\varepsilon}+\vec{0}}\right|\,,
\end{equation}
for conveniently chosen Gaussian perturbations $\vec{\eta}$, so that
\begin{equation}
     \left|\left.\frac{\partial \LC}{\partial \varepsilon_{l_j,j}}\right|_{\vec{\theta}^y = \pi \vec{c}+\vec{\varepsilon}+\vec{\eta}} \right| \geq \frac{1}{2} \left|\left.\frac{\partial \LC}{\partial \varepsilon_{l_j,j}}\right|_{\vec{\theta}^y = \pi \vec{c}+\vec{\varepsilon}+\vec{0}}\right| \geq \frac{1}{2} |\sin(\varepsilon_{l_j,j})| \left(\frac{1}{2} -\alpha\right)   \min_{j\in S}|\Delta_j(z)| \,,
\end{equation}
where we used Eq.~\eqref{eq:derivative-bound}~\footnote{The use of the constant $\frac{1}{2}$ fraction here is not essential. Any constant fraction would actually suffice.}.

By the fundamental theorem of calculus, and using the chain rule, we can write
\begin{equation} \label{eq:chain-rule}
     \left.\frac{\partial \LC}{\partial \varepsilon_{l_j,j}}\right|_{\vec{\theta}^y = \pi \vec{c}+\vec{\varepsilon}+\vec{\eta}} -  \left.\frac{\partial \LC}{\partial \varepsilon_{l_j,j}}\right|_{\vec{\theta}^y = \pi \vec{c}+\vec{\varepsilon}+\vec{0}} = \int_0^1 \frac{d}{dt}  \left.\frac{\partial \LC}{\partial \varepsilon_{l_j,j}}\right|_{\vec{\theta}^y = \pi \vec{c}+\vec{\varepsilon}+t\vec{\eta}}\,dt = \int_0^1\sum_{s}\eta_{l,k} \left.\frac{\partial^2 \LC}{\partial (t\eta_{l,k})\partial \varepsilon_{l_j,j}}\right|_{\vec{\theta}^y = \pi \vec{c}+\vec{\varepsilon}+t\vec{\eta}} \,dt\,.
\end{equation}
Now, the only derivatives in the sum in Eq.~\eqref{eq:chain-rule} that may be non-zero are those where $\eta_{l,k}$ corresponds to a gate in the light cone of the Heisenberg-evolved Max-Cut Hamiltonian. In particular, after $l$ layers the support of a given Heisenberg evolved 2-local term $\sigma_z^{(w)}\sigma_z^{(w')}$ is at most $2(2l+1)$ (because of the nearest-neighbors one-dimensional connectivity of the hardware efficient ansatz). Hence, since there is only one $R_y$ rotation per qubit per layer, we  can bound the maximum number of $R_y$ rotations that can affect a given Hamiltonian term as
\begin{equation} 
    \sum_{l=0}^{L-1} 2(2l+1) = 4 \sum_{l=0}^{L-1} l + 2L = 2L(L-1) + 2L = 2L^2 \,.
\end{equation}
Let us call $M$ the maximum number of local $2$-body terms in the Hamiltonian that can be affected by an $R_y$ rotation with $\varepsilon_{l_j,j}$ angle. Consequently, we find that there are at most $2ML^2$ non-vanishing local terms in the sum in Eq.~\eqref{eq:chain-rule}.

Let us next bound the second derivatives. For a single Hamiltonian term, we have
\begin{equation}
    f_{w,w'}(\vec{\theta})=\Tr\left[\rho\,U(\vec{\theta})^\dagger \sigma_z^{(w)}\sigma_z^{(w')} U(\vec{\theta})\right]\,.
\end{equation}
Consider a Gaussian-initialized  $R_y$ gate in the circuit, $e^{-i \eta_{l,k} \sigma_y/2}$. Writing the circuit as $U(\vec{\theta})=U_+ e^{-i t\eta_{l,k} \sigma_y/2}U_-$, we have
\begin{equation}
    f_{w,w'}(\vec{\theta})=\Tr\left[\rho\,U_-^\dagger e^{i t\eta_{l,k} \sigma_y/2} U_+^\dagger \,\sigma_z^{(w)}\sigma_z^{(w')}\, U_+ e^{-i t\eta_{l,k} \sigma_y/2}U_-\right]\,.
\end{equation}
Then
\begin{equation} \label{eq:first-derivative}
    \frac{\partial f_{w,w'}(\vec{\theta})}{\partial (t \eta_{l,k})}=\frac{i }{2}\Tr\left[\rho\, U_-^\dagger\,e^{i t\eta_{l,k} \sigma_y/2}\left[\sigma_y, \,U_+^\dagger\sigma_z^{(w)}\sigma_z^{(w')} U_+\right]e^{-i t\eta_{l,k} \sigma_y/2}U_-\right]\,.
\end{equation}

To compute the second derivative, suppose, for definiteness, that the gate with parameter $\eta_{l,k}$ appears after the gate with parameter $\varepsilon_{l_j,j}$ in the circuit (a complete analogous analysis follows if the order is reversed). We now write the circuit as $U(\vec{\theta})=U_3\, e^{-it\eta_{l,k} \sigma_y/2}\, U_2 e^{-i\varepsilon_{l_j,j}\sigma_y/2}\,U_1$. Thus,
\begin{equation}
    f_{w,w'}(\vec{\theta}) = \Tr\left[\rho\,U_1^\dagger e^{i\varepsilon_{l_j,j}\sigma_y/2} U_2^\dagger e^{it\eta_{l,k} \sigma_y/2} U_3^\dagger \sigma_z^{(w)}\sigma_z^{(w')} U_3 e^{-it\eta_{l,k} \sigma_y/2} U_2 e^{-i\varepsilon_{l_j,j}\sigma_y/2} U_1\right]\,.
\end{equation}
Using Eq.~\eqref{eq:first-derivative} we obtain
\begin{equation}
    \frac{\partial  f_{w,w'}(\vec{\theta}) }{\partial (t\eta_{l,k})}=\frac{i}{2}\Tr\left[\rho\,U_1^\dagger e^{i\varepsilon_{l_j,j}\sigma_y/2} U_2^\dagger e^{it\eta_{l,k}\sigma_y/2}\left[\sigma_y, U_3^\dagger \sigma_z^{(w)}\sigma_z^{(w')} U_3\right]e^{-it\eta_{l,k}\sigma_y/2} U_2e^{-i\varepsilon_{l_j,j}\sigma_y/2} U_1\right]\,.
\end{equation}
Let us call
\begin{equation}
B\equiv U_2^\dagger e^{it\eta_{l,k} \sigma_y/2} \left[\sigma_y, U_3^\dagger \sigma_z^{(w)}\sigma_z^{(w')} U_3\right]e^{-it\eta_{l,k} \sigma_y/2} U_2\,,
\end{equation}
so that
\begin{equation}
    \frac{\partial  f_{w,w'}(\vec{\theta}) }{\partial (t\eta_{l,k})}=\frac{i}{2}\Tr\left[\rho\,U_1^\dagger e^{i\varepsilon_{l_j,j}\sigma_y/2} B  e^{-i\varepsilon_{l_j,j}\sigma_y/2} U_1\right]\,.
\end{equation}

Differentiating once more, now with respect to $\varepsilon_{l_j,j}$, we find
\begin{align}
    \frac{\partial^2  f_{w,w'}(\vec{\theta})}{ \partial(t\eta_{l,k}) \partial \varepsilon_{l_j,j}} &=-\frac{1}{4}\Tr\left[\rho\,U_1^\dagger e^{i\varepsilon_{l_j,j}\sigma_y/2} [\sigma_y,B] e^{-i\varepsilon_{l_j,j}\sigma_y/2} U_1\right] \nonumber \\ &= -\frac{1}{4}\Tr\left[\rho\,U_1^\dagger e^{i\varepsilon_{l_j,j}\sigma_y/2}\left[\sigma_y,U_2^\dagger e^{it\eta_{l,k}\sigma_y/2}\left[\sigma_y,U_3^\dagger \sigma_z^{(w)}\sigma_z^{(w')} U_3\right]e^{-it\eta_{l,k}\sigma_y/2}U_2\right]e^{-i\varepsilon_{l_j,j}\sigma_y/2}U_1\right]\,.
\end{align}
Let us bound this expression. Since $\left|\Tr[\rho X]\right|\leq \|X\|$, where $\|\cdot\|$ denotes the operator norm, and this norm is invariant under unitary conjugation, we arrive at
\begin{equation}
    \left|\frac{\partial^2  f_{w,w'}(\vec{\theta})}{ \partial(t\eta_{l,k}) \partial \varepsilon_{l_j,j}}\right|\leq \frac{1}{4} \left\|\left[\sigma_y,U_2^\dagger e^{it\eta_{l,k}\sigma_y/2}\left[\sigma_y, U_3^\dagger \sigma_z^{(w)}\sigma_z^{(w')} U_3\right]e^{-it\eta_{l,k}\sigma_y/2}  U_2\right]\right\|\,.
\end{equation}
Using $\|\sigma_y\|=1$ and $\|[\sigma_y,X]\|\leq 2\|X\|$ (from the sub-additivity and sub-multiplicativity of the operator norm), we find
\begin{equation}
    \left|\frac{\partial^2  f_{w,w'}(\vec{\theta})}{ \partial(t\eta_{l,k}) \partial \varepsilon_{l_j,j}}\right|\leq \frac{1}{2}\left\|U_2^\dagger e^{it\eta_{l,k}\sigma_y/2}\left[\sigma_y, U_3^\dagger \sigma_z^{(w)}\sigma_z^{(w')} U_3\right]e^{-it\eta_{l,k}\sigma_y/2}  U_2\right\|\,.
\end{equation}
Moreover, by the invariance of the operator norm under unitary conjugation,
\begin{equation}
    \left\|U_2^\dagger e^{it\eta_{l,k}\sigma_y/2}\left[\sigma_y, U_3^\dagger \sigma_z^{(w)}\sigma_z^{(w')} U_3\right]e^{-it\eta_{l,k}\sigma_y/2}  U_2\right\| = \left\|\left[\sigma_y,U_3^\dagger \sigma_z^{(w)}\sigma_z^{(w')} U_3\right]\right\|\,.
\end{equation}
Using again $\left\|\left[\sigma_y,U_3^\dagger \sigma_z^{(w)}\sigma_z^{(w')} U_3\right]\right\| \leq 2\left\|U_3^\dagger \sigma_z^{(w)}\sigma_z^{(w')} U_3\right\|$, we obtain
\begin{equation}
    \left|\frac{\partial^2  f_{w,w'}(\vec{\theta})}{ \partial(t\eta_{l,k}) \partial \varepsilon_{l_j,j}}\right| \leq  \|U_3^\dagger \sigma_z^{(w)}\sigma_z^{(w')} U_3\| = \left\| \sigma_z^{(w)}\sigma_z^{(w')}\right\| = 1 \,.
\end{equation}
Therefore, we find that
\begin{equation}
    \left|\sum_{l,k} \left.\frac{\partial^2 \LC}{\partial (t\eta_{l,k})\partial \varepsilon_{l_j,j}}\right|_{\vec{\theta}^y = \pi \vec{c}+\vec{\varepsilon}+t\vec{\eta}} \right|\leq  \sum_{l,k}\left| \left.\frac{\partial^2 \LC}{\partial (t\eta_{l,k})\partial \varepsilon_{l_j,j}}\right|_{\vec{\theta}^y = \pi \vec{c}+\vec{\varepsilon}+t\vec{\eta}} \right| \leq 2ML^2\,,
\end{equation}
and 
\begin{equation} \label{eq:2norm-bound}
    \left|\sum_{l,k} \left(\left.\frac{\partial^2 \LC}{\partial (t\eta_{l,k})\partial \varepsilon_{l_j,j}}\right|_{\vec{\theta}^y = \pi \vec{c}+\vec{\varepsilon}+t\vec{\eta}}\right)^2 \right| \leq 2M^2 L^2\,.
\end{equation}

We then notice that the sum $\sum_{l,k}\eta_{l,k} \left.\frac{\partial^2 \LC}{\partial (t\eta_{l,k})\partial \varepsilon_{l_j,j}}\right|_{\vec{\theta}^y = \pi \vec{c}+\vec{\varepsilon}+t\vec{\eta}}$ can be interpreted as an inner product between $\vec{\eta}$ and a vector containing the second derivatives (whose Euclidean $2$-norm, denoted by $\|\cdot\|_2$, is bounded in Eq.~\eqref{eq:2norm-bound}). By the Cauchy--Schwarz inequality, and using Eq.~\eqref{eq:chain-rule},
\begin{equation}\label{eq:bound-integral}
    \left| \left.\frac{\partial \LC}{\partial \varepsilon_{l_j,j}}\right|_{\vec{\theta}^y = \pi \vec{c}+\vec{\varepsilon}+\vec{\eta}} -  \left.\frac{\partial \LC}{\partial \varepsilon_{l_j,j}}\right|_{\vec{\theta}^y = \pi \vec{c}+\vec{\varepsilon}+\vec{0}}\right| \leq \int_0^1 \sqrt{2} ML\|\vec{\eta} \|_2 dt = \sqrt{2} ML\|\vec{\eta} \|_2 \,.
\end{equation}

Let us now assume that the parameters $\eta_{l,k}$ are independently sampled according to a Gaussian distribution $\NC(0,\varsigma^2)$. We recall that in such case
\begin{equation}
    {\rm Pr}(|\eta_{l,k}|\geq x) \leq 2e^{-\frac{x^2}{2\varsigma^2}}\,.
\end{equation}
Using the union bound, and the fact there are at most $nL$ $R_y$ parameters in the circuit, we find that
\begin{align}
    {\rm Pr}\left(\max_{l,k}|\eta_{l,k}|\geq x\right) &\leq \sum_{l,k} {\rm Pr}\left(|\eta_{l,k}|\geq x\right) \nonumber \\ & \leq 2nL e^{-\frac{x^2}{2\varsigma^2}} \,.
\end{align}
We want this probability to be at most $\delta$, so
\begin{equation}
    \delta = 2nL e^{-\frac{x^2}{2\varsigma^2}} \quad\Rightarrow\quad x = \varsigma \sqrt{ 2\log\left(\frac{2nL}{\delta}\right)}\,.
\end{equation}
Hence, with probability at least $1-\delta$, it holds that
\begin{equation} \label{eq:Gaussian-tail-bound}
   \|\vec{\eta} \|_2 \leq \varsigma \sqrt{ 2nL\log\left(\frac{2nL}{\delta}\right)}\,.
\end{equation}
Plugging this bound into Eq.~\eqref{eq:bound-integral}, 
\begin{equation}
     \left| \left.\frac{\partial \LC}{\partial \varepsilon_{l_j,j}}\right|_{\vec{\theta}^y = \pi \vec{c}+\vec{\varepsilon}+\vec{\eta}} -  \left.\frac{\partial \LC}{\partial \varepsilon_{l_j,j}}\right|_{\vec{\theta}^y = \pi \vec{c}+\vec{\varepsilon}+\vec{0}}\right| \leq 2 ML\varsigma  \sqrt{ nL\log\left(\frac{2nL}{\delta}\right)}\,.
\end{equation}
Using Eqs.~\eqref{eq:derivative-bound} and~\eqref{eq:reverse-triangle-ineq}, we arrive at
\begin{equation}
    \left|\left.\frac{\partial \LC}{\partial \varepsilon_{l_j,j}}\right|_{\vec{\theta}^y = \pi \vec{c}+\vec{\varepsilon}+\vec{\eta}} \right| \geq |\sin(\varepsilon_{l_j,j})| \left(\frac{1}{2} -\alpha\right)   \min_{j\in S}|\Delta_j(z)| -  2 ML\varsigma  \sqrt{ nL\log\left(\frac{2nL}{\delta}\right)}\,.
\end{equation}
We now impose condition~\eqref{eq:fraction-condition},
\begin{equation}
     2 ML\varsigma  \sqrt{ nL\log\left(\frac{2nL}{\delta}\right)}\leq \frac{1}{2} |\sin(\varepsilon_{l_j,j})| \left(\frac{1}{2} -\alpha\right)   \min_{j\in S}|\Delta_j(z)| \,,
\end{equation}
and solve for $\varsigma$
\begin{equation}
    \varsigma \leq \frac{\frac{1}{2} |\sin(\varepsilon_{l_j,j})| \left(\frac{1}{2} -\alpha\right)   \min_{j\in S}|\Delta_j(z)|}{ 2 ML  \sqrt{ nL\log\left(\frac{2nL}{\delta}\right)}}\,.
\end{equation}
For constant-odd-degree graphs, it holds that $\min_{j\in S}|\Delta_j(z)|\geq 2$, so that choosing $|\sin(\varepsilon_{l_j,j})| \in\Omega(1)$, we conclude that if
\begin{equation}
    \varsigma \in \OC\left(\frac{1}{ML\sqrt{ nL\log\left(\frac{2nL}{\delta}\right)}}\right)\,,
\end{equation}
all the parameters in the slices avoid barren plateaus. Since $M\in\OC\left(\poly(n)\right)$, we see that the variance of the Gaussian-initialized angles is $\varsigma\in\Theta\left(1/\poly(n)\right)$ for $L\in\OC\left(\poly(n)\right)$.

Finally, we prove that the Gaussian-initialized angles themselves also have at least inverse polynomially decreasing gradients, under a simple condition on the diagonal $R_z$ angles. Namely, that these angles are close enough to $0$ or $\pi$. 
First, let us assume that the $R_z$ angles are in $\{0,\pi\}$. Then, up to an irrelevant global phase, these gates are single-qubit Pauli $\sigma_z$ matrices. Since $\sigma_z R_y(\theta)=R_y(-\theta)\sigma_z$, all such diagonal gates can be commuted to the end of the circuit by appropriately changing the signs of the $R_y$ angles. And because they commute with the MaxCut Hamiltonian, they do not affect the value of the loss function. 
Thus, for the purpose of evaluating the loss, the circuit is equivalent to one in which all $R_y$ rotations acting on the same qubit combine into a single effective angle. We denote this effective angle by
\begin{equation}
    \vartheta_k=\bar{\varepsilon}_{l_k,k}+\sum_l\xi_{l,k}\eta_{l,k}\,,
\end{equation}
where $\xi_{l,k}\in\{-1,1\}$, and $\bar{\varepsilon}_{l_k,k}$ is either $0$ or $\pm\varepsilon_{l_k,k}$. Therefore, Eq.~\eqref{eq:loss-gammas} becomes
\begin{equation}
    \LC(\vec{\vartheta})=\sum_{(k,k')\in E}z_k z_{k'}\cos(\vartheta_k)\cos(\vartheta_{k'})\,.
\end{equation}
Moreover, since $\partial\vartheta_k/\partial\eta_{lk}=\xi_{l,k}$,
\begin{equation}
    \frac{\partial \LC}{\partial\eta_{l,k}}=\xi_{l,k}\frac{\partial \LC}{\partial\vartheta_k}\,,
\end{equation}
and thus we find
\begin{equation}
    \left|\frac{\partial \LC}{\partial\eta_{l,k}}\right|=\left|\sin(\vartheta_k)\right|\left|z_k\sum_{k':(k,k')\in E}z_{k'}\cos(\vartheta_{k'})\right|\,.
\end{equation}
An analogous analysis as in Eq.~\eqref{eq:graph-degree-bound} leads to
\begin{equation}
    \left|z_k\sum_{k':(k,k')\in E}z_{k'}\cos(\vartheta_{k'})+\frac{1}{2}\Delta_k(\vec{z})\right|\leq \deg(k)\max_q(1-\cos \vartheta_q)\,.
\end{equation}
where we recall that $\deg(k)$ is the degree of vertex $k$. 
Hence, if
\begin{equation} \label{eq:conditional}
   \max_k \deg(k)\,\max_k(1-\cos \vartheta_k) \leq \alpha\min_k|\Delta_k(\vec{z})|\,,
\end{equation}
for some constant $\alpha$, then
\begin{equation}
    \left|z_k\sum_{k':(k,k')\in E}z_{k'}\cos(\vartheta_{k'})\right|\geq \left(\frac{1}{2}-\alpha\right)\min_k|\Delta_k(\vec{z})|\,.
\end{equation}
We can see that for constant-odd-degree graphs, both $\max_k \deg(k)$ and $\min_k|\Delta_k(\vec{z})|$ are constant, and so Eq.~\eqref{eq:conditional} can be satisfied with constant $\vartheta_k$.

It remains to lower bound $|\sin(\vartheta_k)|$. 
For any qubit $k$ with at least one Gaussian-initialized $R_y$ angle, $\sum_l\xi_{l,k}\eta_{l,k}$ is a sum of independent Gaussian variables, and therefore a Gaussian random variable itself, with variance at least $\varsigma^2$ and at most $L\varsigma^2$. It follows that
\begin{equation}
    \Pr(\left|\sum_l\xi_{l,k}\eta_{l,k}\right|\leq x) \leq \int_{-x}^x \frac{1}{\varsigma\sqrt{2\pi}} e^{-\frac{t^2}{2L\varsigma^2}} dt \leq \int_{-x}^x \frac{1}{\varsigma\sqrt{2\pi}}  dt  = \sqrt{\frac{2}{\pi}}\frac{x}{\varsigma}\,.
\end{equation}
Taking a union bound over the $n$ qubits and choosing
\begin{equation}
    x=\frac{\delta\varsigma}{2n}\sqrt{\frac{\pi}{2}}\,,
\end{equation}
we obtain 
\begin{equation}
  \Pr( \min_k \left|\sum_l\xi_{l,k}\eta_{l,k}\right| \geq \frac{\delta\varsigma}{2n}\sqrt{\frac{\pi}{2}}) \geq 1- \frac{\delta}{2}\,.
\end{equation}
On the other hand, using the same reasoning as in Eq.~\eqref{eq:Gaussian-tail-bound}, we have that with probability at least $1-\frac{\delta}{2}$,
\begin{equation}
    \max_k\left|\sum_l\xi_{l,k}\eta_{l,k}\right|\leq \varsigma\sqrt{2L\log\left(\frac{4n}{\delta}\right)}\,.
\end{equation}
Consequently, using the union bound again, we see that with probability at least $1-\delta$,
\begin{equation}
   \frac{\delta\varsigma}{2n}\sqrt{\frac{\pi}{2}}  \leq  \left|\sum_l\xi_{l,k}\eta_{l,k}\right| \leq \varsigma\sqrt{2L\log\left(\frac{4n}{\delta}\right)}\,.
\end{equation}
Thus, for $\varsigma,\delta\in\Theta(1/\poly(n))$, we see that condition~\eqref{eq:conditional} can be easily satisfied by appropriately choosing the angles $\varepsilon_{l_k,k}$ in a constant-width region.

We now put all the pieces together. First, if $\bar{\varepsilon}_k=0$, we can use that $|\sin x|\geq \frac{2|x|}{\pi}$ when $|x|\leq \frac{\pi}{2}$ to get (with probability at least $1-\delta$)
\begin{equation}
    |\sin(\vartheta_k)|\geq \frac{\delta\varsigma}{n\sqrt{2\pi}}\,.
\end{equation}
Then, if $\bar{\varepsilon}_k\neq 0$ we can again use the reverse triangle inequality,
\begin{equation}
|\sin(\vartheta_k)| = \left|\sin(\varepsilon_{l_k,k} + \sum_l\xi_{l,k}\eta_{l,k})\right| \geq |\sin(\varepsilon_{l_k,k})|- \left|\sin(\varepsilon_{l_k,k} + \sum_l\xi_{l,k}\eta_{l,k}) -\sin(\varepsilon_{l_k,k})\right| \,.
\end{equation}
And noticing that $|\sin(x)-\sin(y)|\leq |x-y|$ (i.e., the sine is a $1$-Lipschitz function),
\begin{equation}
    |\sin(\vartheta_k)| \geq |\sin(\varepsilon_{l_k,k})| - \left|\sum_l\xi_{l,k}\eta_{l,k}\right|\,,
\end{equation}
so that with probability at least $1-\delta$,
\begin{equation}
    |\sin(\vartheta_k)| \geq \min_k |\sin(\varepsilon_{l_k,k})| -\varsigma\sqrt{2L\log\left(\frac{4n}{\delta}\right)} \,.
\end{equation}

Combining all these bounds, we conclude that, with probability at least $1-\delta$, for every Gaussian-initialized $R_y$ angle,
\begin{equation}
    \left|\frac{\partial L}{\partial\eta_{l,k}}\right|\geq \left(\frac{1}{2}-\alpha\right)\min_k|\Delta_k(\vec{z})|\min\left\{\min_k |\sin(\varepsilon_{l_k,k})| -\varsigma\sqrt{2L\log\left(\frac{4n}{\delta}\right)},\frac{\delta\varsigma}{n\sqrt{2\pi}}\right\}\,.
\end{equation}
In particular, for constant-odd-degree graphs, $\min_k|\Delta_k(\vec{z})|\geq 2$, and therefore all Gaussian-initialized $R_y$ angles have gradients that vanish at most polynomially whenever $L\in\OC(\poly(n))$ and $\varepsilon_{l_k,k}$, $\varsigma$ and $\delta$ are lower bound by inverse-polynomial functions in $n$, except for pathological choices that render $\min_k |\sin(\varepsilon_{l_k,k})| -\varsigma\sqrt{2L\log\left(\frac{4n}{\delta}\right)}$ exponentially vanishing.

A completely analogous analysis shows that the $R_z$ angles can also be initialized with similar Gaussian distributions around $0$ or $\pi$, and that this does not make the gradients of the $R_y$ rotations exponentially vanishing, which concludes the proof. 

\end{proof}

\section{Additional numerics for local landscape characterization}\label{app-additional-numerics}

 \begin{figure}[ht]
    \centering
    \includegraphics[width=1\columnwidth]{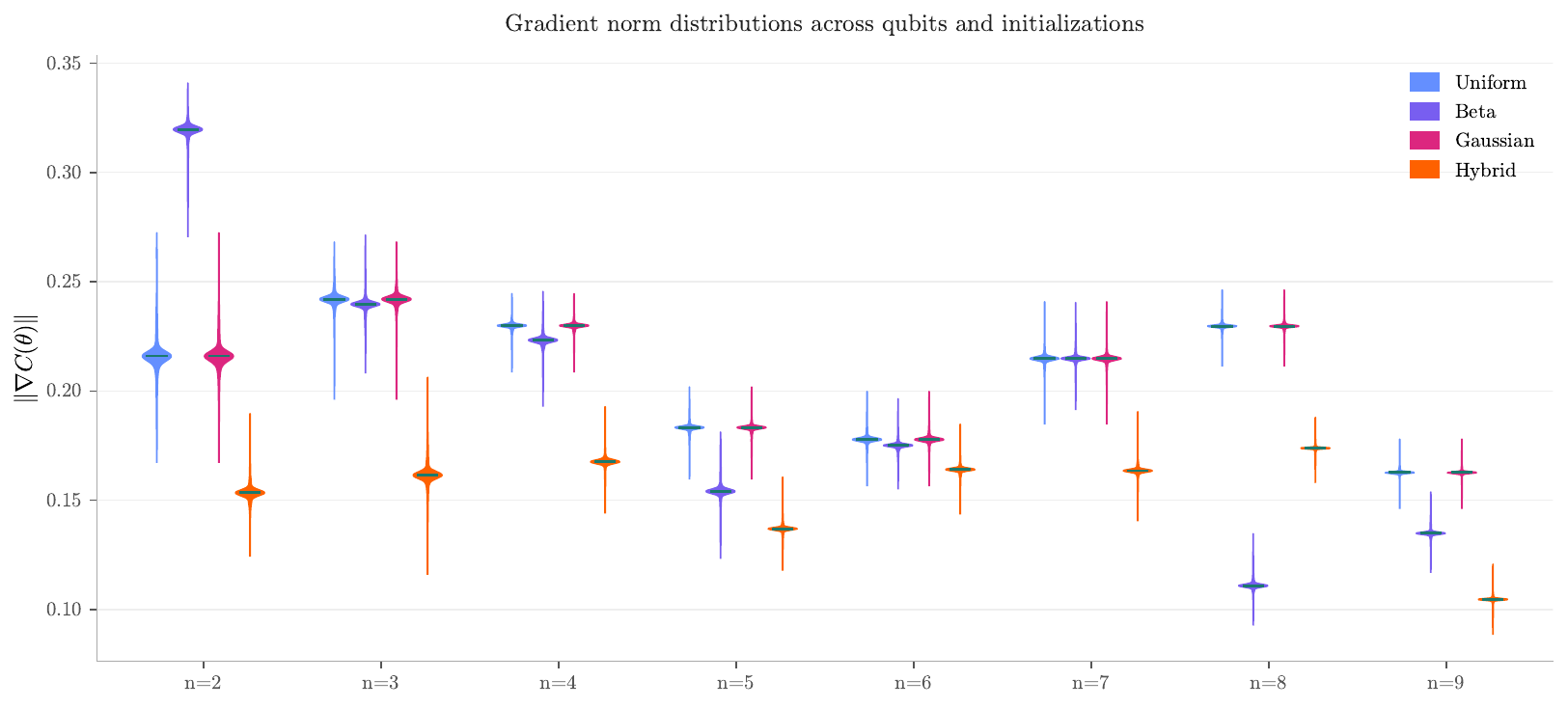}
    \caption{\textbf{Gradient norm distributions for different initialization families and system sizes.} We consider an $L=2n$ layered circuit (see Fig.~\ref{fig:qckt_ela}), where the parameters are initialized i.i.d. from a uniform distribution $\mathrm{Unif}[-\pi,\pi]$, a Beta distribution $\mathrm{Beta}(\alpha,\beta)$,  Gaussian initializations $\NC(0,\frac{1}{n})$ and a ``hybrid initialization'' strategy. We sample 500 circuits and consider circuits acting on $n=2,3,\ldots,9$ qubits.}
    \label{fig:gradnorm_distribs}
\end{figure}

In this section we present additional numerics to characterize the local landscape for different initialization strategies. While existing analyses of parameter initialization often focus on gradients at the initial point, our goal here is broader. Motivated by the analytical results of the previous sections, we now ask whether distinct initialization strategies continue to bias optimization after training has begun, by steering the optimizer toward different local neighborhoods in the loss function landscape. To probe this question, we do not restrict attention to the gradients at the initial point, but instead study the local landscape around a trained point $\bm{\theta}^\ast$ reached after optimization.

We again consider the ansatz of Fig.~\ref{fig:qckt_ela}, and perform simulations with qubit numbers ranging from $n=2$ to $n=8$ and depth $L=2n$. Here we focus on a variational quantum eigensolver task~\cite{peruzzo2014variational} where the goal is to find the ground state of the random one-dimensional Hamiltonian
\begin{equation}\label{eq:eg_hamiltonian}
    O = \sum_{i=1}^{n} \alpha_i \sigma_x^i + \sum_{i=1}^{n} \beta_i \sigma_z^i + \sum_{i=1}^{n-1} \sigma_x^i \sigma_x^{i+1}\,,
\end{equation}
where the coefficients $\alpha_i$ and $\beta_i$ are sampled uniformly at random from $[-1,1]$. For each system size, we input the all-zero state, and initialize the PQC parameters using uniform, Beta, and Gaussian distributions, and for each initialization family we perform $50$ independent training runs of $50$ gradient-descent steps with learning rate $10^{-3}$. Training is performed with gradient descent and gradients are computed via the parameter-shift rule~\cite{mitarai2018quantum}. For each run we record the corresponding $\bm{\theta}^t$ for $t=50$ and use it as the reference point for our local landscape analysis. All simulations are performed in the  noiseless regime (no shot noise and no hardware noise).

To explore the neighborhood around $\bm{\theta}^t$, we generate a perturbation set
\begin{equation}\label{eq:perturb_pts}
    \mathcal{P}(\bm{\theta}^t)=\{\bm{\theta}^{\,j}_p\}_j,
    \qquad
    \bm{\theta}^{\,j}_p = \bm{\theta}^t + \lambda_j \bm{\epsilon}_j,
\end{equation}
where each $\bm{\epsilon}_j$ is a random direction in parameter space and the radii $\lambda_j$ are logarithmically spaced between $10^{-3}$ and $10^{-1}$. Evaluating the loss and gradients on $\mathcal{P}(\bm{\theta}^t)$ allows us to probe the local geometry induced by each initialization strategy after training.

\subsubsection{A statistical picture of the local landscape}

To analyze the landscape, we consider the ansatz in Figure~\ref{fig:qckt_ela-2} and four different initialization strategies. The first three strategies initialize $\bm{\theta_z}, \bm{\theta}_y \sim P_\gamma(\bm{\theta})$ where $P_\gamma(\bm{\theta}) \in \{\text{Unif}(-r,r), {\rm Beta}(\alpha, \beta), \NC(0, \frac{1}{\sqrt{n}}) \}$. We also consider a ``hybrid'' initialization where $\bm{\theta}_z \sim \text{Unif}(-r, r)$ and $\bm{\theta}_y \sim \NC(\mu, \frac{1}{\sqrt{n}})$ and $\mu \sim \text{Unif}(0, \pi)$. Here, $\bm{\theta}_z$ refers to parameters of the $R_Z$ gate and $\bm{\theta}_y$ refers to parameters of the $R_Y$ gate in the ansatz. The loss function computes the cost w.r.t the observable $O$ in Equation~\ref{eq:eg_hamiltonian}.

Figure~\ref{fig:gradnorm_distribs} shows Gaussian KDE estimate of gradient norm distributions for four different initialization strategies described above. In the low qubit regime (i.e., $n=2,3, 4$) we see that all distributions show fairly large gradients even though hybrid initialization appears to be slightly more concentrated around median at $n=4$. For medium qubit regime (i.e. $n=5,6, 7$) we see marked difference in the behavior of distributions. In this regime, Beta and Gaussian initializations show a resistance to concentration even as median gradient norm decays significantly. Hybrid and Uniform distributions on the other hand show an onset of concentration. This is especially true for hybrid strategy where concentration becomes apparent at $n=6$. In high qubit regime (i.e. $n=8, 9$) we see some phenomena that correlates with our proposed theory. This regime exhibits all initializations concentrating heavily around their median showing that a gradient concentration pattern emerges (especially the hybrid initialization shows a strong concentration around it's median). However, Beta and Gaussian show larger tail  distributions which implies that $P(\Ebb[\nabla_\theta \LC(\bm{\theta})] > \delta)) > 0$ at some regions in the landscape. This directly correlates with our earlier analysis that showed that there exist non-trivial initialization strategies that does not lead a circuit into a ``barren'' part of the landscape.
\begin{figure}[t]
    \centering
    \includegraphics[width=.5\linewidth]{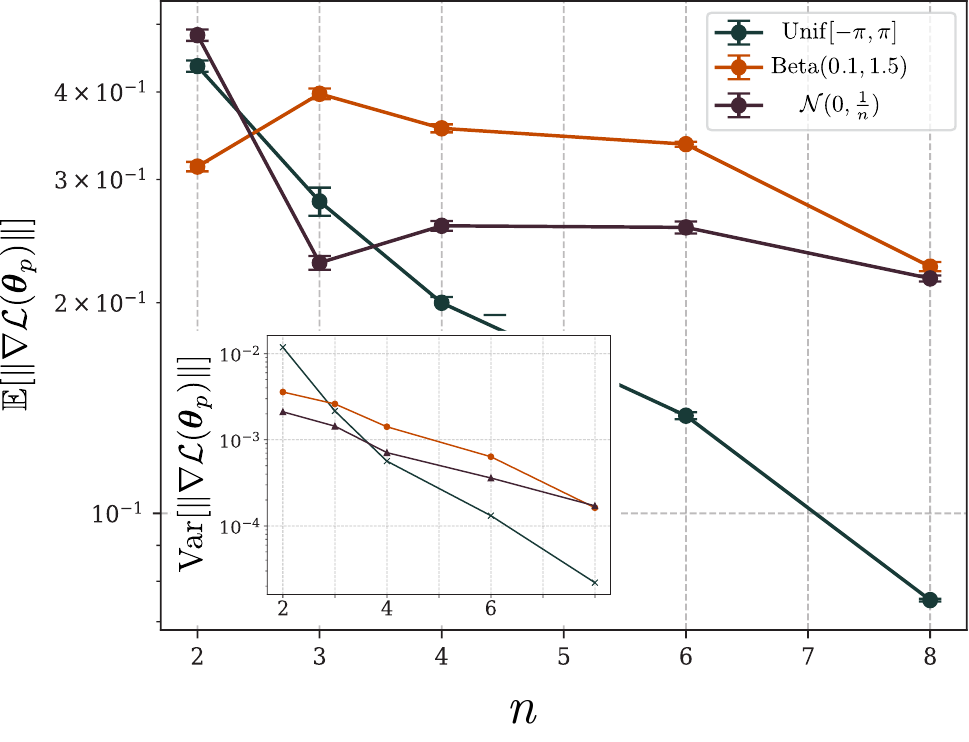}
    \caption{\textbf{Order statistics of gradient distributions around convergent point.} We show the expectation value and the variance (inset) for the values of $\|\nabla \mathcal{L}(\bm{\theta}_p)\|$ in Fig.~\ref{fig:gradnorm_distribs}. }
    \label{fig:gradnorm_stats}
\end{figure}

\begin{figure}[t]
    \centering
    \includegraphics[width=.5\linewidth]{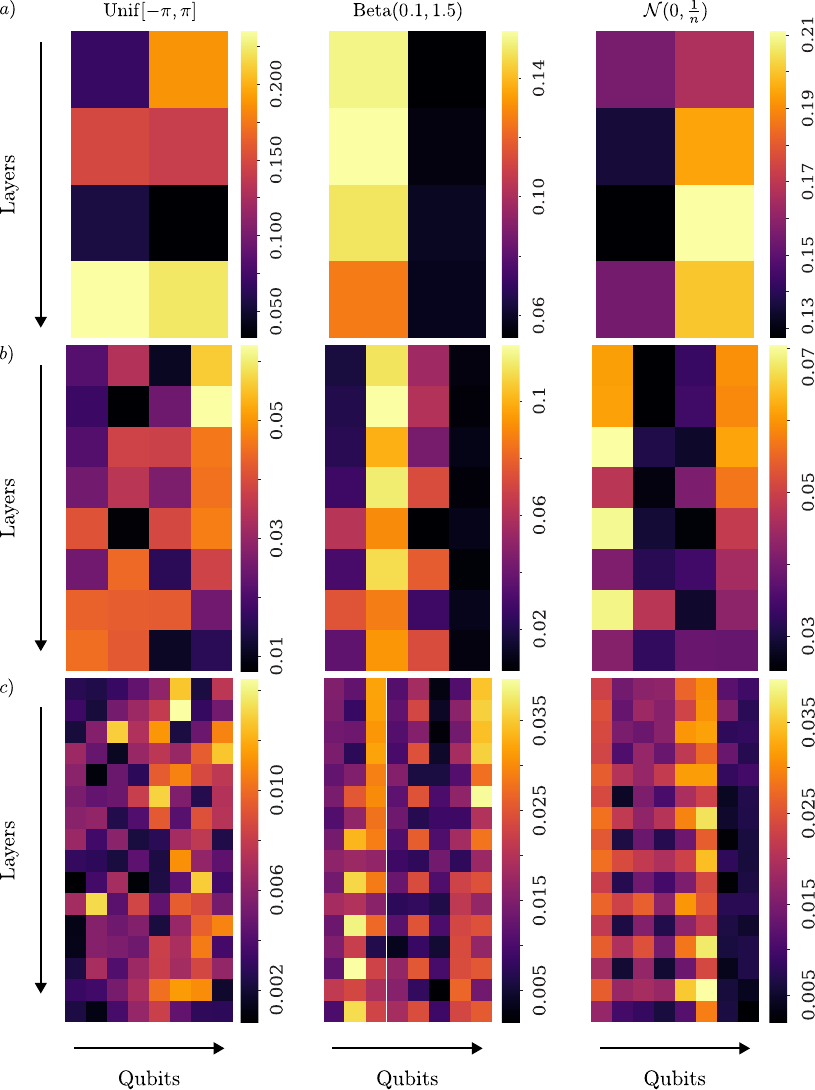}
    \caption{\textbf{Magnitude of gradient norms across different qubits and layers for three different qubit sizes.} We consider system sizes of $n=2,4,8$ qubits (respectively panels a), b) and c)), a number of layers $L=2n$, and initialization strategies  $\mathrm{Unif}[-\pi,\pi]$, $\mathrm{Beta}(\alpha,\beta)$, and  $\NC(0,\frac{1}{n})$. }
    \label{fig:gradnorm_heatmaps}
\end{figure}

Figure~\ref{fig:gradnorm_stats} completes this picture through the expectation value and the variance of the gradient norms. The uniform initialization exhibits a rapid decay in both the mean and the variance as the system size increases, consistent with increasingly featureless local geometry. The Beta and Gaussian initializations display a much slower decay, with Beta retaining the largest average gradients over most of the explored range. Importantly, we do not interpret these post-training statistics as a proof that one neighborhood exhibits a barren plateau while another does not. Rather, the point is that analytically distinct initialization strategies remain practically distinct after optimization has started. That is, they appear to lead the optimizer to local neighborhoods with different gradient scales and different fluctuation structure.

\subsubsection{A fine-grained picture of the local landscape}

To move beyond aggregate statistics, we next resolve the gradient information across qubits and layers. Figure~\ref{fig:gradnorm_heatmaps} reports heatmaps of the gradient magnitudes associated with the perturbed points in $\mathcal{P}(\bm{\theta}^\ast)$, organized by qubit and layer.

\begin{figure}[t]
    \centering
    \includegraphics[width=.5\columnwidth]{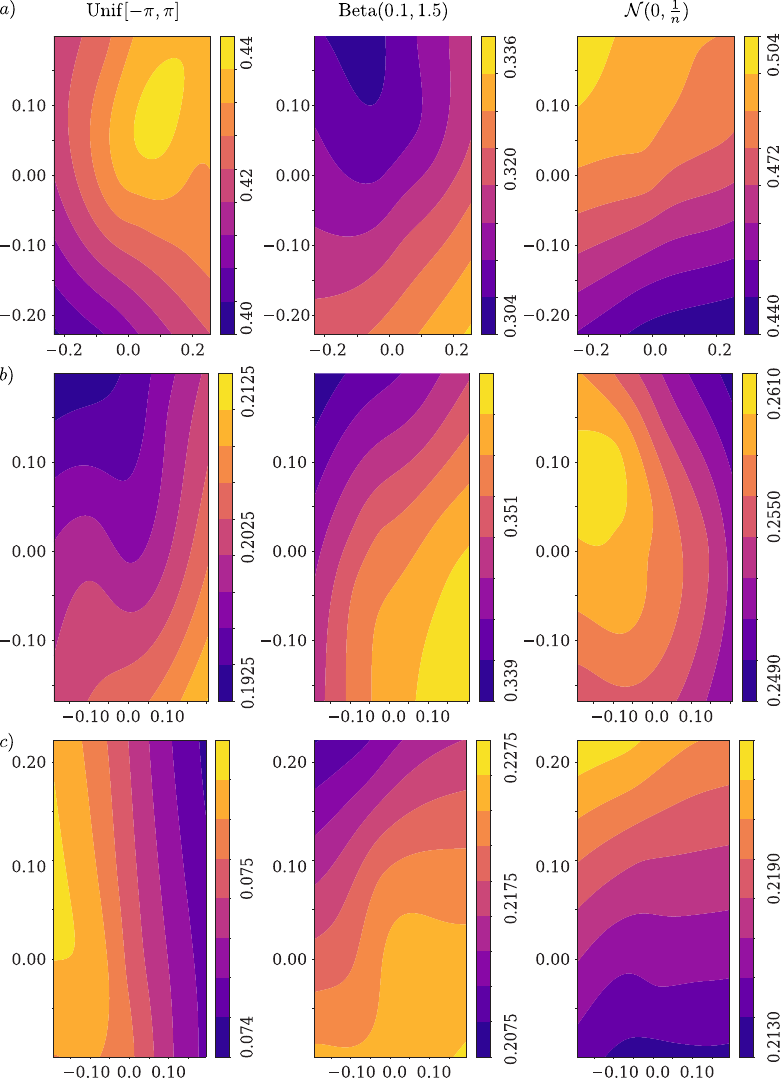}
    \caption{\textbf{Gradient landscape around some local optimum for different parameter initializations and system sizes.} We consider system sizes of $n=2,4,8$ qubits (respectively panels a), b) and c)), a number of layers $L=2n$, and initializations strategies  $\mathrm{Unif}[-\pi,\pi]$, $\mathrm{Beta}(\alpha,\beta)$, and  $\NC(0,\frac{1}{n})$. }
    \label{fig:grad_landscape} 
\end{figure}

The qualitative contrast between initialization strategies is again clear. For the uniform initialization, the gradients become both smaller and more homogeneous as the system size increases. In particular, for $n=8$ the heatmaps display increasingly little variation across qubits and layers, indicating that the corresponding trained points lie in locally flat and nearly isotropic neighborhoods. By contrast, the Beta and Gaussian initializations retain visibly heterogeneous gradient profiles. Even when their overall gradient magnitudes decrease with system size, these two strategies continue to exhibit anisotropy across layers and subsystems, which in turn suggests the persistence of nontrivial descent directions.

This distinction is particularly suggestive in the comparison between Beta and Gaussian initializations. Although both behave substantially better than the uniform baseline, the Beta initialization tends to preserve larger gradients over a broader portion of the circuit, while the Gaussian initialization develops a more localized structure. We find this difference especially interesting because the Gaussian width is chosen according to the principled scaling $\sigma^2=1/n$, whereas the Beta initialization is specified through fixed hyperparameters. One possible interpretation is that the non-symmetric character of the Beta family contributes to the observed robustness, although establishing this connection more systematically is left for future work.

Figure~\ref{fig:grad_landscape} complements the previous analysis by showing contour plots of the local gradient landscape around representative trained points (obtained by Eq.~\eqref{eq:perturb_pts}). These plots reinforce the message that different initializations do not merely rescale the gradients, but rather, that they can alter the qualitative structure of the local landscape itself. For small systems, the three initialization strategies already display distinct local geometries. As the system size increases, the uniform initialization develops progressively flatter and less structured neighborhoods, whereas the Beta initialization continues to exhibit more pronounced basin-like features. The Gaussian initialization lies in between, as it retains clear structure at intermediate sizes, but for larger systems its local landscape also becomes noticeably flatter.

Taken together, Figs.~\ref{fig:gradnorm_distribs}--\ref{fig:grad_landscape} show that initialization affects much more than the gradient at step zero. Even after training, different parameter distributions lead the optimizer to local neighborhoods with different gradient scales, different anisotropies, and different geometric structure.

\end{document}